\documentclass[preprint]{elsarticle}
\usepackage{amssymb}
\usepackage{amsmath}
\usepackage{amsthm}
\usepackage{amsfonts}
\usepackage{array}
\usepackage{bm}
\usepackage{booktabs}
\usepackage{cases}
\usepackage{float}
\usepackage[margin=1.0in]{geometry}
\usepackage{hyperref}
\usepackage[capitalize]{cleveref}
\usepackage{mathtools}
\usepackage{multirow}
\usepackage{graphicx}
\usepackage{subcaption}
\usepackage{tabularx}
\usepackage{xcolor}
\usepackage{xspace}

\usepackage[%
  disable,
  colorinlistoftodos]{todonotes}

\usepackage[
  shortcuts,%
  acronym,%
  nonumberlist,%
  nogroupskip]{glossaries}

\usepackage[
  detect-none,%
  binary-units]{siunitx}[=v2]

\newacronym{bdf}{BDF}{backward difference formulas}
\newacronym{bgk}{BGK}{Bhatnagar--Gross--Krook}

\newacronym{cfd}{CFD}{computational fluid dynamics}
\newacronym{cfl}{CFL}{Courant--Friedrichs--Lewy}

\newacronym{dag}{DAG}{directed acyclic graph}
\newacronym{dfs}{DFS}{depth-first-search}
\newacronym{dg}{DG}{discontinuous Galerkin}
\newacronym{dgfem}{DG-FEM}{discontinuous Galerkin finite-element method}
\newacronym{dirk}{DIRK}{diagonally implicit Runge-Kutta}
\newacronym[shortplural={DOFs},longplural={degrees of freedom}]{dof}{DOF}{degree of freedom}

\newacronym{fem}{FEM}{finite-element method}

\newacronym{gcd}{GCD}{Graphics Compute Die}

\newacronym{hpc}{HPC}{high-performance computing}

\newacronym{mi250x}{MI250X}{AMD Instinct 250X}

\newacronym{numa}{NUMA}{non-uniform memory access}

\newacronym{olcf}{OLCF}{Oak Ridge Leadership Computing Facility}
\newacronym{ornl}{ORNL}{Oak Ridge National Laboratory}

\newacronym{rk}{RK}{Runge-Kutta}

\newacronym{sgs}{SGS}{symmetric Gauss-Seidel}

\newacronym{ts}{TS}{Topological Sweep}
\newacronym{tgs}{TGS}{Topological Generation Sweep}

\glsdisablehyper
\setacronymstyle{long-short}
\makeglossaries

\definecolor{CiteColor}{rgb}{0, 0, 0.55}
\definecolor{LinkColor}{rgb}{0.2, 0.2, 0.2}
\definecolor{URLColor}{rgb}{0.62745098, 0.1254902 , 0.94117647}
\hypersetup{
  linkcolor=LinkColor,
  citecolor=CiteColor,
  urlcolor=URLColor,
  colorlinks=true
}

\biboptions{sort&compress,numbers}

\DeclareSIUnit\flops{Flops}

\newtheorem{prop}{Proposition}
\newtheorem{remark}{Remark}
\newtheorem{ass}{Assumption}
\numberwithin{prop}{section}
\numberwithin{remark}{section}
\numberwithin{ass}{section}
\numberwithin{equation}{section}
\numberwithin{table}{section}
\numberwithin{figure}{section}

\newcolumntype{Y}{>{\raggedright\arraybackslash}X}
\newcolumntype{T}[1]{>{\hsize=#1\hsize}X}
\newcolumntype{U}[1]{>{\hsize=#1\hsize}Y}

\newcommand{\amd}{AMD\xspace}
\newcommand{\nvidia}{NVIDIA\xspace}

\newcommand{\ve}[1]{\ensuremath{\mathbf{#1}}}

\newcommand{\grad}[1]{\ensuremath{\nabla #1}}

\newcommand{\mbb}{\mathbf{b}}
\newcommand{\mbc}{\mathbf{c}}

\newcommand{\mbe}{\mathbf{e}}

\newcommand{\mbj}{\mathbf{j}}

\newcommand{\mbA}{\mathbf{A}}

\newcommand{\mbJ}{\mathbf{J}}

\newcommand{\mbU}{\mathbf{U}}

\newcommand{\cE}{\mathcal{E}}
\newcommand{\cF}{\mathcal{F}}
\newcommand{\cG}{\mathcal{G}}

\newcommand{\cM}{\mathcal{M}}

\newcommand{\cP}{\mathcal{P}}
\newcommand{\cQ}{\mathcal{Q}}

\newcommand{\cS}{\mathcal{S}}
\newcommand{\cT}{\mathcal{T}}

\newcommand{\cV}{\mathcal{V}}

\newcommand{\bbM}{\mathbb{M}}

\newcommand{\bbP}{\mathbb{P}}
\newcommand{\bbQ}{\mathbb{Q}}
\newcommand{\bbR}{\mathbb{R}}

\newcommand{\dt}{\Delta t}

\newcommand{\bmu}{{\bm{u}}}
\newcommand{\bmv}{{\bm{v}}}

\newcommand{\bmj}{{\bm{j}}}
\newcommand{\bmx}{{\bm{x}}}

\newcommand{\bmn}{{\bm{n}}}

\newcommand{\avg}[1]{\{\!\!\{#1\}\!\!\}}
\newcommand{\jmp}[1]{[\![#1]\!]}
\newcommand{\ledge}{\big<\!\!\big<}
\newcommand{\redge}{\big>\!\!\big>}
\newcommand{\EIx}{\mathcal{E}_x^{\mathrm{I}}}

\newcommand{\dx}{\mathrm{d}}
\renewcommand{\phi}{\varphi}

\newcommand{\fhj}{f_{h,\mbj}}
\newcommand{\fhs}{f^{(s)}_{h}}
\newcommand{\fhjs}{f^{(s)}_{h,\mbj}}
\newcommand{\fhjr}{f^{(r)}_{h,\mbj}}

\newcommand{\fhjslp}{f^{(s),\ell+1}_{h,\mbj}}
\newcommand{\Fhj}{\cF_{h,\mbj}}
\newcommand{\mhjs}{M^{(s)}_{h,\mbj}}
\newcommand{\mhj}{M_{h,\mbj}}
\newcommand{\nhj}{N_{h,\mbj}}

\newcommand{\qhjs}{q^{(s)}_{h,\mbj}}
\newcommand{\kx}{k_{\bmx}}
\newcommand{\Kx}{K_{\bmx}}
\newcommand{\Tx}{\cT_{\bmx}}
\newcommand{\Wx}{W_{\bmx}}
\newcommand{\Qx}{\cQ_{\bmx}}
\newcommand{\Qex}{\cQ^e_{\bmx}}
\newcommand{\kv}{k_{\bmv}}
\newcommand{\Kv}{K_{\bmv}}
\newcommand{\Tv}{T_{\bmv}}
\newcommand{\Wv}{W_{\bmv}}
\newcommand{\Qv}{\cQ_{\bmv}}
\newcommand{\vj}{\bmv_{\mbj}}
\newcommand{\wj}{\omega^{\bmv}_{\mbj}}
\newcommand{\ej}{\mbe_\mbj}
\newcommand{\intv}{\int_{\mathbb{R}^d}}
\newcommand{\dofx}{\text{dof}_{\bmx}}
\newcommand{\dofv}{\text{dof}_{\bmv}}

\newcommand{\p}{\partial}
\newcommand{\ddd}{, \dots ,}

\newcommand{\fic}{f^{\rm{ic}}}
\newcommand{\nic}{n^{\rm{ic}}}
\newcommand{\uic}{\bm{u}^{\rm{ic}}}
\newcommand{\tic}{\theta^{\rm{ic}}}
\newcommand{\fbc}{f^{\rm{bc}}}

\newcommand{\nl}{n^{\rm{left}}}
\newcommand{\ul}{\bm{u}^{\rm{left}}}
\newcommand{\tl}{\theta^{\rm{left}}}

\begin{document}

\journal{Journal of Computational Physics}

%%---------------------------------------------------------------------------%%
\begin{frontmatter}

    \title{%
        Sweep-based, implicit solutions of the multidimensional BGK
        equation on unstructured grids
        \tnoteref{ornl-cr}}
    \tnotetext[ornl-cr]{%
        This manuscript has been authored by UT-Battelle, LLC, under contract
        DE-AC05-00OR22725 with the US Department of Energy. The United States
        Government retains and the publisher, by accepting the article for
        publication, acknowledges that the United States Government retains a
        nonexclusive, paid-up, irrevocable, worldwide license to publish or
        reproduce the published form of this manuscript, or allow others to do so,
        for United States Government purposes. DOE will provide access to these
        results of federally sponsored research in accordance with the DOE Public
        Access Plan (http://energy.gov/downloads/doe-public-access-plan)
    }

    % \author[ornl]{Eirik Endeve\fnref{mmd}}\ead{endevee@ornl.gov}
    \author[ornl]{Thomas M. Evans\corref{cor1}\fnref{acmes}} \ead{evanstm@ornl.gov}
    \cortext[cor1]{Corresponding Author. Tel: +1 865 576 3535}
    \author[ornl]{Ryan Glasby\fnref{acmes}}\ead{glasbyrs@ornl.gov}
    \author[ornl,utk]{Cory Hauck\fnref{mmd}}\ead{hauckc@ornl.gov}
    \author[ornl]{Stefan R. Schnake\fnref{mmd}}\ead{schnakesr@ornl.gov}
    \author[ornl]{Kyle J. Schwiebert\fnref{mmd}}\ead{schwiebertkj@ornl.gov}
    \author[ornl]{Lawton Shoemake\fnref{acmes}}\ead{shoemakewl@ornl.gov}
    \author[ornl]{Stuart Slattery\fnref{acmes}}\ead{slatterysr@ornl.gov}

    \fntext[acmes]{%
        Computational Sciences and Engineering Division}

    \fntext[mmd]{%
        Computer Science and Mathematics Division
    }

    \affiliation[ornl]{%
        organization={Oak Ridge National Laboratory},
        addressline={1 Bethel Valley Rd.},
        city={Oak Ridge},
        state={TN},
        postcode={37831},
        country={U.S.A.}
    }

    \affiliation[utk]{%
        organization={Mathematics Department, University of Tennessee},
        city={Knoxville},
        state={TN},
        postcode={37996},
        country={U.S.A.}
    }

    \begin{abstract}

        We present a nodal \acl{dg} method for solving the \acf{bgk} kinetic
        equation on multi-dimensional, unstructured grids. The method uses
        implicit, sweep-based solvers and a moment-preserving projection of the
        Maxwellian source to enable high-order accuracy in time while avoiding
        restrictive time steps imposed by boundary layers and other
        geometry-induced features. We verify that the method is correct in the
        continuum limit by comparing to closed-form and high-order solutions of
        the Sod shock problem on 2 and 3D unstructured grids. Linear $L^2$
        stability is demonstrated for a B-stable \acl{dirk} method of third
        order. The solver uses a hybrid parallel scheme based on spatial domain
        decomposition with local sweeps performed on CPU and GPU hardware.
        Platform-portability is demonstrated through the development of new
        GPU-friendly, graph-based sweep algorithms that are implemented using
        the Kokkos performance portability library and achieve greater than 20
        times speedup on NVIDIA H100 GPUs compared to 64-core AMD EPYC 9654
        CPUs. Finally, we show results on the Frontier supercomputer at the
        \acl{olcf} for a boundary value problem with 2.77 trillion phase
        space \aclp{dof} that executed on \num{1536} nodes utilizing \num{6144}
        AMD MI250X GPUs.

    \end{abstract}

    %\begin{highlights}
    %    \item A nodal \acl{dg} scheme for the \acf{bgk} kinetic equation has been
    %    developed and implemented on multi-dimensional, unstructured grids.
    %    \item The nodal discretization includes a moment-preserving projection of the
    %    Maxwellian term in the \acs{bgk} operator and enables the use of sweep algorithms for
    %    implementing fully implicit time integrators.
    %    \item We verify that the method is correct in the continuum limit by comparing to
    %    closed-form solutions of the Sod shock problem on 2 and 3D unstructured grids.
    %    \item Linear $L^2$ stability is established and demonstrated for a family of
    %    \acf{dirk} methods up to third order.
    %    \item Platform-portability is demonstrated through the development of new
    %    GPU-friendly, graph-based sweep algorithms that are implemented using the
    %    Kokkos performance portability library and achieve greater than 20 times
    %    speedup on \nvidia H100 GPUs compared to 64-core AMD EPYC 9654 CPUs.
    %    \item We solve a duct flow problem in 2D and 3D to demonstrate that the algorithm can
    %    reach trillions of \aclp{dof}.
    %\end{highlights}

    %% Keywords

\end{frontmatter}

%%---------------------------------------------------------------------------%%

%%\tableofcontents

\listoftodos

%%---------------------------------------------------------------------------%%
\section{Introduction}
\label{sec:introduction}

The \ac{bgk} equation is a kinetic equation that is used to simulate the
behavior of rarefied gases. It is a computationally cheaper alternative to the
Boltzmann equation, the latter of which models particle collisions with an
integral operator. In the \ac{bgk} equation, the Boltzmann collision operator is
replaced by a nonlinear relaxation model that can be derived heuristically from
the Boltzmann operator by making several near-equilibrium assumptions
\cite{liboff2003kinetic}. The target of this relaxation model is a Maxwellian
distribution that shares the same mass, momentum, and energy density as the
input kinetic distribution. Like the Boltzmann operator, the \ac{bgk} operator
conserves mass, momentum, and energy; it dissipates the Maxwell-Boltzmann
entropy, and its kernel lies in the manifold of Maxwellian distributions
parameterized by density, bulk velocity, and temperature. Like the Boltzmann
equation, the \ac{bgk} equation formally recovers the system of compressible
Euler equations that govern the evolution of these parameters in the limit of
infinite collisions \cite{cercignani2013mathematical}. However, the \ac{bgk}
equation requires modifications to recover proper transport coefficients for the
compressible Navier-Stokes equations \cite{shakhov1968generalization,
andries2000gaussian,holway1965kinetic} in near-equilibrium regimes.

The \ac{bgk} equation was first introduced in \cite{bhatnagar_model_1954};
rigorous mathematical results about the existence and uniqueness of solutions
can be found in \cite{perthame_1989,perthame_1993}. Simulations of the \ac{bgk}
equation typically rely on semi-implicit time integration strategies that treat
particle advection explicitly and particle collisions implicitly. Because the
moments of the \ac{bgk} operator are constant during the implicit update, the
inversion of this nonlinear operator becomes a trivial computation
\cite{coron}.

In many engineering applications, the simulation of fluids and gases is
performed on domains which must also resolve fine physical features
\cite{erwin2013flows,moser1999channelflow,nguyen2019turbine}. In such cases,
specialized meshes, often unstructured, are needed to resolve thin layers
around domain boundaries or internal components that may be moving in the flow.
In order to resolve important geometry-induced features, the size of the mesh
may vary by orders of magnitude, both in terms of the volume between the
smallest and largest cells and the aspect ratio of individual cells. As a
consequence, an explicit treatment of the advection term in the \ac{bgk}
equation leads to time step restrictions that are not practical. In the kinetic
setting, the maximum wave speed is set by the largest microscopic velocity,
which is always greater than the fluid acoustic wave speed, in some cases by an
order of magnitude or more \cite{shin2024collision}.

The most obvious way to address the time step restrictions induced by the
advection operator is with fully implicit time integration. For grid-based
methods, there are two basic approaches: semi-Lagrangian and Eulerian.

In the semi-Lagrangian approach, the advection operator is evaluated by
backtracking along characteristics to the previous time step. Typically, the
location of the characteristic at the previous time step does not align with
the grid, and an interpolation procedure is required. However, one can show
that, up to consistency errors, the mass, momentum, and energy moments of the
distribution are conserved during the characteristic update. As a result, the
Maxwellian can be treated explicitly in time, making the inversion of the
\ac{bgk} operator a trivial linear update.

The semi-Lagrangian approach to the solving the \ac{bgk} equation was proposed
in \cite{filbet2009semilagrangian}, with a rigorous convergence analysis given
in \cite{russo2012convergence}. A high-order version of the method can be found
in \cite{groppi2016high}. The strategy is relatively simple to implement, easy
to parallelize, and requires no time step restriction to maintain stability.
However, in its original formulation, the approach lacks local conservation, in
the sense of preserving collision invariants of the \ac{bgk} operator and in the
sense of conservative discretizations for spatial gradients. The problem of
preserving collision invariants can be addressed by modifying the Maxwellian
parameters through an entropy minimization procedure
\cite{mieussens2000discrete}, which is the more expensive option, or via an
$L^2$-based projection of the solution, which is analytic but may not preserve
positivity \cite{gamba2009spectral}. The problem of nonconservative spatial
discretizations can be addressed with a conservative correction at the cost of
introducing a \ac{cfl}-like time step constraint
\cite{boscarino2025conservative,xiong2019conservative}.

For fully implicit methods, the Eulerian approach allows for conservative
discretizations of spatial gradients, and it can be augmented via modifications
of the Maxwellian that ensure proper collision invariants. The main challenge
with implicit Eulerian methods is that, once discretized in time, they are
essentially steady-state equations that may require sophisticated iterative
solvers. The core of most solvers is a source iteration procedure in which the
Maxwellian in the \ac{bgk} operator is lagged while the remaining components of
the equation are inverted with sweeping algorithms that move sequentially
through the grid according to the direction of the kinetic velocity. This method
has been used \cite{aoki1997numerical} for two-dimensional steady state problems
and more recently in \cite{su2020implicit} as a pseudo-time stepping procedure
to solve the Boltzmann equation in one and two dimensions. An alternative to
sweeping is to use \ac{sgs} methods \cite{cai2025symmetric,zhu2021general},
which split linear systems in upper and lower triangular subsystems that are
solved in succession. \ac{sgs} methods are more flexible than sweeping but
typically require more iterations. Newton methods can also be used
\cite{mieussens2000discrete}, especially since the Jacobian of the Maxwellian
with respect to the kinetic solution is local in space. However, the coupling
induced between velocity degrees of freedom in the linearized system precludes
sweeping. For sweep-based and \acs{sgs}-based methods, moment-based acceleration
schemes \cite{hauck2025high,su2020fast,zeng2023general, taitano2014moment} have
been introduced to improved convergence behavior in highly collisional regimes.

In the current work, we propose a nodal \ac{dg} scheme in both space and
velocity, similar to \cite{su2020implicit,gerhard_parallel_2024}, for
time-dependent solutions of the \ac{bgk} equation, although the methodology in
\cite{gerhard_parallel_2024} is formulated as a kinetic scheme for general
conservation laws. The \ac{dg} method provides a locally conservative
discretization of the spatial gradients and, after an appropriate projection of
the Maxwellian, satisfies local conservation properties for the collision
invariants. We integrate the discretized system in time using \ac{dirk} methods.
For the resulting algebraic equations, the nodal form of the \ac{dg}
discretization provides a discrete velocity interpretation that is amenable to
sweeping. We allow for arbitrary unstructured spatial meshes that require unique
sweep orderings for each velocity unknown. Once accounting for any possible
cycles (see \cref{sec:transport-sweeps}) each ordering is encoded by a \ac{dag}
that can be traversed to invert the transport operator.

We employ a hybrid parallel performance strategy in which the spatial mesh is
domain decomposed into MPI-managed partitions, and local sweeps are performed
on CPU multi-core or GPU hardware. Domain boundary fluxes are communicated and
updated after each local sweep, and the process is iterated to convergence
during the non-linear solve using a Picard iteration. Due to both time-dependence
and the non-linear behavior of the \ac{bgk} collision operator, convergence of
the domain boundary fluxes generally occurs within or near the non-linear
tolerance, and advanced domain decomposition and cycle handling strategies that
have been developed within the linear transport community
\cite{adams_provably_2020,vermaak_massively_2021} are not required. Therefore,
computational efficiency of the method is dominated by the performance of the
local sweeps. We introduce two platform-portable local sweeping algorithms in
\cref{sec:transport-sweeps} that enable optimal partitioning of the \acp{dag}
depending on the hardware (multi-threaded CPU versus GPU) and geometric
configuration of the problem.

The \ac{bgk} solver is implemented in the \ac{ornl} \ac{cfd} Spinnaker codebase
\cite{spinnaker}. Spinnaker is built on the Trilinos \cite{trilinos} software
framework for \ac{hpc} scientific applications. Mesh-based data structures,
domain decomposition, and finite-element functions are handled by the Trilinos
infrastructure. The sweep solvers are implemented using the Kokkos library
\cite{kokkos-ecosystem} that enables platform-portable implementations across
diverse architectures from a single implementation. Accordingly, Spinnaker can
execute \ac{bgk} problems on multi-node systems using MPI for inter-node,
domain decomposition communication and mesh partitioning, and local sweeps are
performed within a node on either multi-threaded CPU or GPU hardware.
\Cref{sec:solver-performance} shows results using Spinnaker on \ac{hpc} systems
including the Frontier supercomputer \cite{frontier} at the \ac{olcf}.

To verify the implementation, we simulate Sod problems in one, two, and three
dimensions and compare with analytic, closed-form, and high-order continuum
solutions. We also examine the effects of the moment preserving projection of
the Maxwellian, investigate linear stability, and provide a preliminary study
of scaling behavior. Finally, we simulate a simple duct problem that, in the
three-dimensional setting, uses a phase space mesh with over 2.7 trillion
degrees of freedom.

The rest of the paper is organized as follows. In Section \ref{sec:bgk}, we
briefly recall the \ac{bgk} equation. In Section \ref{sec:discretization}, we
present the spatial discretization and the implicit time integrators used. In
Section \ref{sec:sol-methods}, we present the basic source iteration scheme and
the underlying sweeping technology. Section \ref{sec:results} contains
numerical results.

%%---------------------------------------------------------------------------%%

%%---------------------------------------------------------------------------%%
\section{BGK Model}
\label{sec:bgk}

For $d \in \{1,2,3\}$, let $X \subset \bbR^d$ be a spatial domain with
Lipschitz boundary; let $Z = X \times \bbR^d$ be the position-velocity phase
space; and let
\begin{equation}
    \partial Z^{\pm}
    = \{ (\bm{x},\bm{v}) \in \partial X \times \mathbb{R}^d
    : \pm \bm{v} \cdot \bm{n}(\bm{x}) > 0 \}
\end{equation}
be the outflow ($+$) and inflow ($-$) boundaries of $Z$, respectively.  For any
non-negative, measurable function $g = g(v)$, let
\begin{align}
    \label{eq:primitive_g}
    n_g = \intv g\,d\bm{v},
    \quad
    \bm{u}_g =
    \frac{1}{n_g}\intv\bm{v} g\,d\bm{v}, \quad \text{and} \quad
    \theta_g & = \frac{1}{d}\frac{1}{n_d}
    \intv\bigl|\bm{v}-\bm{u}_g\bigr|^2
    g \,d\bm{v}
\end{align}
be the density, bulk velocity, and temperature, respectively, associated with
$g$.  Moreover for any $n \geq 0$, $\bm{u} \in \mathbb{R}^d$, and $\theta > 0$,
define the Maxwellian
\begin{equation}
    M_{n,\bm{u},\theta} (\bm{v})
    =
    \frac{n}{\bigl(2\pi\theta)^{d/2}}
    \exp\Biggl(-\frac{\bigl|\bm{v}-\bm{u}\bigr|^2}
    {2\theta}\Biggr).
\end{equation}
Then the Maxwellian associated to $g$ is:
\begin{equation}
    \cM_g(\bm{v})
    = M_{n_g,\bm{u}_{g},\theta_g}(\bm{v})
    = \frac{n_g}{\bigl(2\pi\theta_g)^{d/2}}
    \exp\Biggl(-\frac{\bigl|\bm{v}-\bm{u}_g\bigr|^2}
    {2\theta_g}\Biggr).
    \label{eq:M_g}
\end{equation}
An important property of the $\cM_g$ is that it is the unique Maxwellian such that
\begin{equation}
    \label{eq:Mawellian-moment-condition}
    \intv \mbe \cM_g \,d\bm{v} =  \intv \mbe g \,d\bm{v},
\end{equation}
where $\mbe = (\mbe^0, \mbe^1,\mbe^2 ) = (1, \bm{v},
    \tfrac{1}{2}|\bm{v}|^2)^\top$.

The \ac{bgk} model \cite{bhatnagar_model_1954} is used to describe the
evolution of a kinetic distribution, $f = f(\bm{x}, \bm{v},t)$, that here gives
the number density, with respect to the measure $d \bm{v} d \bm{x}$, of
particles at position $\bm{x} \in X$ and time $t >0$ that move with microscopic
velocity $\bm{v} \in \bbR ^d$. The model takes the form
\begin{subequations}
    \label{eq:bgk}
    \begin{numcases}{}
        \partial_t f(\bm{x},\bm{v},t) +
        \bm{v}\cdot\nabla_{\bm{x}} f(\bm{x},\bm{v},t) =
        \nu \left(\cM_f(\bm{x},\bm{v},t) - f(\bm{x},\bm{v},t) \right) , & $(\bm{x},\bm{v}) \in Z, t >0$,
        \label{eq:bgk-eq}
        \\
        f(\bm{x},\bm{v},0) = f^{\rm{ic}}(\bm{x},\bm{v}), & $(\bm{x},\bm{v}) \in Z$,
        \label{eq:bgk-ic}
        \\
        f(\bm{x},\bm{v},t) = f^{\rm{bc}}(\bm{x},\bm{v},t), & $(\bm{x},\bm{v}) \in \partial Z^-, t >0$,
        \label{eq:bgk-bc}
    \end{numcases}
\end{subequations}
where the Maxwellian associated to $f$ is defined via
\eqref{eq:M_g}:
\begin{equation}
    \mathcal{M}_f(\bm{x},\bm{v},t)
    = M_{n_f(\bm{x},t), \bm{u}_f(\bm{x},t), \theta_f(\bm{x},t)}(\bm{v})
    =
    \frac{n_f(\bm{x},t)}{\bigl(2\pi\theta_f(\bm{x},t)\bigr)^{d/2}}
    \exp\Biggl(-\frac{\bigl|\bm{v}-\bm{u}_f(\bm{x},t)\bigr|^2}
    {2\theta_f(\bm{x},t)}\Biggr),
    \label{eq:M_f}
\end{equation}
and the fields $n_f$, $\bm{u}_f$, and $\theta_f$ are the density, bulk velocity,
and temperature associated to $f$, defined via \eqref{eq:primitive_g}:
\begin{align}
    n_f(\bm{x},t)      & := n_{f(\bm{x},\cdot,t)}
    = \intv f(\bm{x},\bm{v},t)\,d\bm{v},
    \label{eq:density}                                 \\
    \bm{u}_f(\bm{x},t) & := \bm{u}_{f(\bm{x},\cdot,t)}
    =\frac{1}{n_{f(\bm{x},\cdot,t)}} \intv\bm{v} f(\bm{x},\bm{v},t)\,d\bm{v},
    \quad \text{and}
    \label{eq:bulk_velocity}                           \\
    \theta_f(\bm{x},t) & := \theta_{f(\bm{x},\cdot,t)}
    =\frac{1}{d}\frac{1}{n_{f(\bm{x},\cdot,t)} }
    \intv\bigl|\bm{v}-\bm{u}_{f(\bm{x},\cdot,t)} \bigr|^2
    f(\bm{x},\bm{v},t)\,d\bm{v}.
    \label{eq:theta}
\end{align}
The collision frequency $\nu \ge 0$ is typically a function of $\theta_f$ and
$n_f$, but for the purposes of the current work, we assume that it is an
absolute constant. For $\nu \gg 1$,  the solution of $f$ is approximated by
$\cM_f$, whose parameters $n_f$, $\bm{u}_f$, and $\theta_f$ are, in turn,
approximated by solution of the compressible Euler equations \cite[Chapter
11]{cercignani2013mathematical}. Specifically, due to
\eqref{eq:Mawellian-moment-condition}, the velocity moments
\begin{equation}
    \mbU_f = \intv \mathbf{e}f(\bm{v})\,d\bm{v}
    = \begin{pmatrix}
        n_f          \\
        n_f \bm{u}_f \\
        \frac12 n_f |\bmu_f|^2  + \frac{d}{2} n_f \theta_f
    \end{pmatrix}
    \label{eq:moments}
\end{equation}
satisfy the conservation law
\begin{equation}
    \partial_t\mbU_f + \nabla_{\bmx}\cdot \left( \intv \bmv\mathbf{e}f \right) = 0.
    \label{eq:conservation}
\end{equation}
In the formal limit
$\varepsilon = \nu / L \to \infty$, $\mbU_f \to \mbU$ where $\mbU$ satisfies
the Euler equations
\begin{equation}\label{eq:euler_system}
    \partial_t \mbU(\bm{x},t) + \grad_{\bm{x}} \cdot \ve{F}(\mbU(\bm{x},t)) = 0,
\end{equation}
with
\begin{equation}
    \ve{U}
    = \begin{pmatrix}
        n        \\
        n \bm{u} \\
        \frac12 n |\bm{u}|^2  + \frac{d}{2} n \theta
    \end{pmatrix},
    \quad \text{and} \quad
    \ve{F}(\ve{U})
    =\begin{pmatrix}
        n \bm{u}                             \\
        n \bm{u} \otimes \bm{u} + n \theta I \\
        \frac12 n |\bm{u}|^2 \bm{u} + \frac{d+2}{2} n \theta \bm{u}
    \end{pmatrix}.
\end{equation}

%%- - - - - - - - - - - - - - - - - - - - - - - - - - - - - - - - - - - - - -%%

\section{Discretization}
\label{sec:discretization}

We discretize the \ac{bgk} model using a nodal \ac{dg} method in both position
and velocity space. The spatial mesh is general and may be unstructured, while
the velocity mesh is a tensor product of uniform intervals. Time integration is
performed using \ac{dirk} methods.

\subsection{Spatial discretization}

Let $\Tx$ be a mesh on $X$ with $N_{\bmx}$ polygonal cells characterized by a
mesh parameter $h_{\bmx} = \max_i{h_{\bmx,i}}$ where ${h_{\bmx,i}}$ is the
diameter of the smallest sphere circumscribed around $\Kx^i\in\Tx$.
Let $W_{\bmx}$ be the \ac{dgfem} space on $\Tx$ defined by
\begin{equation}\label{eq:discrete_space_x}
    W_{\bmx} = \left\{ w\in L^2(X): w\big|_{\Kx}\in \bbM^{\kx}(\Kx) \quad \forall \Kx \in \Tx \right\},
\end{equation}
where $L^2(X)$ is the space of square-integrable functions defined on $X$. Here
$\bbM^{\kx}(\Kx)$ is either $\bbP^{\kx}(\Kx)$, the space of polynomials on $\Kx$
up to total degree $\kx$, or $\bbQ^{\kx}(\Kx)$, the space of polynomials on
$\Kx$ up to degree $\kx$ in each component of $\bmx$. The total degrees of
freedom in $\Wx$ is denoted $\dofx$.

Because functions in $W_{\bmx}$ can be discontinuous at cell boundaries, it is
necessary to define averages and jumps along these boundaries. Let $\EIx$ be
the interior skeleton of $\Tx$; that is,
\begin{equation}
    \EIx=\{e:  e = \p{\Kx^+}\cap\p{\Kx^-} \ne \emptyset
    \text{ for some } \Kx^+, \Kx^- \in \cT_\bmx \}.
\end{equation}
For any $\bmx \in e =\p{\Kx^+}\cap\p{\Kx^-}\in\EIx$, the average and jump
operators of a function $a \in W_{\bmx}$ are given by
\begin{equation}\label{eq:avg_and_jmp}
    \avg{a}(\bmx) = \tfrac{1}{2}(a^+(\bmx) +a^-(\bmx)),
    \qquad
    \jmp{a}(\bmx) = a^+(\bmx)\bmn^+(\bmx) + a^-(\bmx)\bmn^-(\bmx),
\end{equation}
where $\bmn^\pm(\bmx)$ are the outward normals to $\Kx^\pm$ at $\bmx$ and
$a^\pm(\bmx) = \lim_{\bmx'\to \bmx} a(\bmx')$ where $\bmx'\in K_\bmx^\pm$.

A global basis for $\Wx$ is formed using a local Lagrange nodal basis on each
$\Kx \in \Tx$ that is defined via affine transformations from a fixed reference
element $\Kx^{\text{ref}}$ to $\Kx$. Reference elements and corresponding bases
that are currently used in the software are listed in \Cref{tab:pos-basis}. We
denote such an interpolatory basis for $W_{\bmx}$ by $\{\phi_i\}_i$ using
interpolation points by $\{\bmx_i\}_i$, such that $\phi_i(\bmx_i') =
    \delta_{i,i'}$.

To evaluate the bilinear forms in the variational formulation below, we use
quadrature operators $\Qx$ and $\Qex$. Denote $C(D)$ to be the set of
functions continuous on the open set $D$. The quadrature operators $\Qx$ and
$\Qex$ are defined such that for any $a\in C(\Kx)$ and $a^e\in C(e)$,
\begin{equation}
    \Qx a \approx \int_{\Kx} a (\bmx) d \bmx
    \qquad\text{and}{\qquad}
    \Qex a^e \approx \int_{e} a^e(\bmx) d \bmx.
\end{equation}
Moreover, define, for any $a,b \in C(\Kx)$ and for any $a^e,b^e \in C(e)$,
\begin{equation}
    (a,b)_{\Kx} = \Qx(ab)
    \qquad\text{and}{\qquad}
    \ledge a^e,b^e\redge_{e} = \Qx^e(a^e b^e).
\end{equation}

We choose quadratures that satisfy the following assumption:
\begin{ass}
    The quadratures $\Qx$ and  $\Qex$ are exact for polynomials of degree
    $2 \kx$.  In particular,
    \begin{equation}
        (a,b)_{\Kx}
        = \int_{\Kx} a(\bmx) b(\bmx) \, d \bmx
        \;
        \forall a,b \; \in \bbM^{p}(\Kx),
        \quad\text{and}\quad
        \ledge a^e,b^e \redge_{e}
        = \int_{e} a(\bmx) b(\bmx) \, d \bmx
        \;
        \forall \chi,\xi \; \in \bbM^{p}(e)
    \end{equation}
    and, as a consequence, $\int_X a(\bmx)b(\bmx) d\bmx = \sum_{\Kx \in \Tx}
        (a,b)_{\Kx}$ for all $a,b \in \Wx$.
\end{ass}

\begin{table}[h]
    \centering
    \caption{Various reference element and corresponding basis used.}
    \label{tab:pos-basis}
    \begin{tabular}{c|c|l|c|l}
        $d$ & $\kx$ & Reference element $\Kx^{\text{ref}}$                                                                  & $\bbM^{\kx}$ & Reference basis                            \\ \hline
        2   & 1     & $\{\tilde{\bmx}=(\xi,\eta): \xi \geq 0; \eta \geq 0, \xi + \eta \leq 1\}$                             & $\bbP^1$     & $\{1-\xi-\eta,\xi,\eta\}$                  \\
        2   & 1     & $\{\tilde{\bmx}=(\xi,\eta): \xi \in [-1,1],\eta \in [-1,1]\}$                                         & $\bbQ^1$     & $\{\tfrac{1}{4}(1 \pm \xi)(1 \pm \eta) \}$ \\
        3   & 1     & $\{\tilde{\bmx}=(\xi,\eta,\zeta): \xi \geq 0, \eta \geq 0, \zeta \geq 0, \xi + \eta + \zeta \leq 1\}$ & $\bbP^1$     & $\{1-\xi-\eta-\zeta,\xi,\eta,\zeta\}$
    \end{tabular}
\end{table}

\subsection{Velocity discretization}
\label{sec:velocity-discretization}

For computational purposes, we restrict $\bmv \in \bbR^3$ to a computational
domain $V = [-L,L]^d$ for some sufficiently large, but finite, constant $L >0$.
Let $\cT_{\bmv}$ be a mesh on $V$ with $N_{\bmv} = n^d_\bmv $ uniform Cartesian cells
$K_\bmv$ of side length $h_{\bmv}$ = $2L/n_\bmv$, and let
\begin{equation}\label{eq:discrete_space_v}
    \Wv = \left\{ w \in L^2(V): w\big|_{\Kv}\in \bbQ^{k_\bmv}(\Kv) \quad \forall
    \Kv \in \cT_{\bmv} \right\}.
\end{equation}
The total number of degrees of freedom in $\Wv$ is given by $\dofv=(k_\bmv +
    1)^d n_\bmv^d$.  While functions in $W_\bmv$ can also be discontinuous at cell
boundaries, the lack of velocity gradients in the \ac{bgk} model means that average
and jump operators with respect to $\bmv$ are not needed.

A global basis for $\Wv$ is formed using a tensor product Lagrange basis on
each $\Kv \in \cT_{\bmv}$ that is mapped by affine transformations from a
reference element $\Kv^{\text{ref}} = [-1,1]^d$ to $\Kv$. We denote the global
interpolation points are denoted by $\{\vj\}$, where $\mbj=(j_1,\ldots,j_d) \in
    \mbJ \subset \mathbb{N}^d$ is a multi-index and $|\mbJ| = \dofv$. We denote the
basis for $\Wv$ by $\{\psi_{\mbj}\}_{\mbj}$.

To evaluate the bilinear forms in the variational formulation below, we use a
quadrature operator $\Qv$ such that for all $a \in C(\Kv)$
\begin{equation}
    \Qv (a) \approx \int_{\Kv}a(\bmv) d \bmv
\end{equation}
and define, for any $a,b \in C(\Kv)$
\begin{equation}
    (a,b)_{\Kv} = \Qv(a b)
\end{equation}
We choose quadratures that satisfy the following assumption
\begin{ass}
    \label{ass:velocity-quadrature}
    The quadrature $\Qv$ is exact for polynomials of degree $2 \kv +1$.  In
    particular,
    \begin{equation}
        (\bmv a,b)_{\Kv}
        = \int_{\Kv} \bmv a(\bmv) b(\bmv) \, d \bmv,
        \;
        \forall a,b \; \in \bbQ^{k_\bmv}(\Kv)
    \end{equation}
    and, as a consequence, $\int_V \bmv a(\bmv) b(\bmv) d \bmv = \sum_{\Kv \in
            \Tv} (\bmv a,b)_{\Kv}$ for all $a,b \in \Wv$.
\end{ass}
In this work, we use $\kv = 2$ for all computations to ensure conservation, and
to satisfy Assumption \ref{ass:velocity-quadrature}, we use a quadrature
operator $\Qv$ which uses the standard three-point Gauss--Legendre
quadrature on each cell $\Kv$.

\subsection{Variational formulation in phase space}

At a semi-discrete level, the nodal \ac{dg} method is: Find
\begin{equation}
    f_h(\bmx,\bmv,t)
    = \sum_\mbj \fhj(\bmx,t)  \psi_{\mbj}(\bmv)
    = \sum_i \sum_\mbj f_{h,\mbj,i}(t) \psi_{\mbj}(\bmv) \phi_i(\bmx)
\end{equation}
such that for all $\mbj$, for all $\Kx \in \Tx$, and every $t$,
\begin{equation} \label{eq:bgk_cont_time}
    (\p_t \fhj, \zeta_h)_{\Kx}
    = (\cF_{h,\mbj}(f_h),\zeta_h)_{\Kx},
    \qquad
    f_{h,\mbj,i}|_{t=0}
    = f^{\rm{ic}}(\bmx_i, \vj) \quad \text{for all $i$ such that $\bmx_i \in \Kx$},
\end{equation}
where $\cF_h$ is defined by
\begin{equation}
    \label{eq:def-F}
    (\cF_{h,\mbj}(f_h),\zeta_h)_{\Kx}
    = (\fhj, \vj \cdot \grad_\bmx \zeta_h)_{\Kx}
    - (\hat f_{h,\mbj},\vj \cdot  \bmn \zeta_h)_{\p \Kx}
    - \nu (\fhj,\zeta_h)_{\Kx}
    + \nu (\mhj, \zeta_h)_{\Kx}.
\end{equation}
Here $\bmn(\bmx)$ is the unit outward normal at $\bmx \in \p \Kx$,
$(\cdot,\cdot)_{\p \Kx}$ is defined by quadrature over $\p \Kx$, and
\begin{equation}
    \hat f_{h,\mbj}(\bmx,\cdot)
    =  \lim_{\varepsilon\to 0}
    \begin{cases}
        f_{h,\mbj}(\bmx-\varepsilon \bmn,\cdot),   & \bmx \in \p_\mbj^{+}\Kx, \\
        f_{h,\mbj}(\bmx + \varepsilon \bmn,\cdot), & \bmx \in \p_\mbj^{-}\Kx,
    \end{cases}
\end{equation}
where
\begin{equation}
    \label{eq:cell-inflow-outflow-boundaries}
    \p_\mbj^\pm\Kx
    = \{\bmx\in \p \Kx:\pm \vj \cdot \bmn(\bmx) > 0\}.
\end{equation}
In addition, the discrete velocity Maxwellian is
\begin{equation}
    \mhj = \cP_{\bmv,h}\cM_{f_h} |_{\bmv = \vj},
\end{equation}
where the projection $\cP_{\bmv,h}\colon L^2(V) \to W_\bmv$ is a projection
operator defined such that
\begin{equation}
    \label{eq:Maxwellian-projection}
    \int_{V} \mbe \cP_{\bmv,h}\cM_{f_h} \,d\bm{v}
    = \int_{\bbR^3} \mbe \cM_{f_h} \,d\bm{v}.
\end{equation}
The condition in \eqref{eq:Maxwellian-projection} is to ensure that a numerical
form of the conservation laws \eqref{eq:conservation} are satisfied, even though
the computational velocity domain is bounded. The construction of $\cP_{\bmv,h}$
is provided in \ref{app:velocity-projection}.

\subsection{Basic properties of the discretization}

The nodal \ac{dg} formulation in velocity enables the use of sweeping
techniques discussed in the next section. However, with the use of
Gauss-Legendre quadrature for $\Qv$, the method satisfies a discrete version of
the conservation laws in \eqref{eq:conservation}

\begin{prop}[Conservation laws]
    Suppose that the quadrature operator $\Qv$ is exact for any $a,b \in
    W_\bmv$. Then the nodal \ac{dg} method satisfies a discrete version of the
    conservation law \eqref{eq:conservation} on each cell $\Kx$.
\end{prop}

\begin{proof}
    Let $\mbe_j = \mbe(\vj)$ and let $\wj$ be the weights of the velocity
    quadrature $\Qv$.  Then setting $\zeta_h=1$ in \eqref{eq:def-F} and applying the quadrature $\Qv$
    gives,
    \begin{equation}
        \label{eq:-cons}
        \sum_{\mbj} \wj \ej (\cF_{h,\mbj}(f_h),1)_{\Kx}
        =    - \sum_{\mbj} \wj \ej (\hat f_{h,\mbj},\vj \cdot  \bmn)_{\p \Kx}
        + \nu  \sum_{\mbj} \wj \ej (\fhj,1)_{\Kx}
        - \nu \sum_{\mbj} \wj \ej  (\mhj, 1)_{\Kx}.
    \end{equation}
    Since $\cP_{\bmv,h}\cM_{f_h^{(s)}}$, $f^{(s)}_h$, and the components
    of $\mbe$ are all in $W_\bmv$, it follows from
    \eqref{eq:Maxwellian-projection} that
    \begin{equation}
        \label{eq:collision-invariance}
        \sum_{\mbj} \wj \ej  \mhj
        = \int_{V} \mbe \cP_{\bmv,h}\cM_{f_h^{(s)}} \,d\bm{v}
        = \int_{\bbR^3} \mbe \cM_{f_h^{(s)}} \,d\bm{v}
        = \int_{\bbR^3} \mbe f_h^{(s)} \,d\bm{v}
        = \int_{V} \mbe f_h^{(s)} \,d\bm{v}
        = \sum_{\mbj} \wj \ej \fhj
    \end{equation}
    Let $\mbU_h(t) = \sum_{\mbj} \wj \ej \fhj$.  Then substitution of
    \eqref{eq:-cons} and \eqref{eq:collision-invariance} into
    \eqref{eq:bgk_cont_time} gives
    \begin{equation}
        \label{eq:discrete-cons}
        (\p_t \mbU_h, 1)_{\Kx}
        + \sum_{\mbj} \wj \ej (\hat f_{h,\mbj},\bmv_\mbj \cdot  \bmn)_{\p \Kx}=0.
    \end{equation}
\end{proof}
\begin{remark}\label{rmk:jump_quadrature_error}
    In terms of the average and jump terms,
    \begin{equation}
        \vj \hat f_{h,\mbj}(\bmx,\cdot)
        = \vj \avg{f_{h,\mbj}}(\bmx,\cdot) - \frac12 |\vj\cdot\bmn| \jmp{f_{h,\mbj}}(\bmx,\cdot).
    \end{equation}
    Due to the term $|\vj\cdot\bmn|$, the quadrature evaluation of the jump term
    is typically not exact.  However, it has been shown that such an
    approximation does not harm the accuracy of the method
    \cite{pazner_short_2021}.  Moreover, if the jump is zero, then because $\Qv$ can
    integrate $ \bmv \mbe f_h$ exactly, it follows that
    \begin{multline}
        \sum_{\mbj} \wj \ej (\bmv_j \cdot  \bmn) \hat f_{h,\mbj}(\bmx,\cdot)
        = \sum_{\mbj} \wj \ej \vj \cdot  \bmn \avg{f_{h,\mbj}}(\bmx,\cdot) \\
        = \int_V (\bmv\cdot \bmn) \mbe \avg{f_h} (\bmx,\bmv, \cdot) d \bmv
        = \int_V (\bmv\cdot \bmn) \mbe f_h (\bmx,\bmv, \cdot) d \bmv.
    \end{multline}
\end{remark}

The method is also linearly stable. Following the calculation in
\cite{hauck2025high}, we set $\zeta_h = f_{h,\mbj}$ and sum over all $\mbj$.
The result is an evolution equation for the $L^2$ phase space norm:
\begin{equation}
    \label{eq:l2-estimate}
    \begin{split}
        \frac12 \frac{\partial}{\partial t}
        \left\lVert f_{h} \right\rVert_{L^2(X
                                    \times V)}^2
        = & - \nu \sum_{\mbj}w_{\mbj}\| \mhj- \fhj\|_{V}^2
        + \nu \sum_{\mbj} w_{\mbj} (\mhj,\fhj - \mhj)_{X}                              \\
          & - \sum_{\mbj}w_{\mbj}\sum_{e \in {\EIx}}(|\vj \cdot\bmn| \jmp{f_{h,\mbj}},
        \jmp{f_{h,\mbj}})_e                                                            \\
          & - \sum_{\mbj} w_{\mbj}\ledge|\vj \cdot\bmn|   f_{h,\mbj},
        f_{h,\mbj} \redge_{\partial^{+}_{\mbj}X}
        + \sum_{\mbj} w_{\mbj}\ledge |\vj \cdot\bmn|  f^{\mathrm{bc}},
        f^{\mathrm{bc}}
        \redge_{\partial^{-}_{\mbj}X} .
    \end{split}
\end{equation}
The fact that the Maxwellian is nonlinear precludes an $L^2$ based estimate.
Thus to verify linear stability of the method in our calculations, we introduce
the simplified Maxwellian, which is linear:
\begin{equation}
    \label{eq:linear-Maxwellian}
    \nhj
    = N_h = \frac{n_{f_h}}{\lvert V \rvert}\quad\forall \mbj\in\mbJ.
\end{equation}
\begin{prop}[Linear Stability]\label{prop:linear-stability}
    Suppose that either (i) $\nu = 0$ or (ii) $\mhj$ is replaced by $\nhj$ in
    \eqref{eq:l2-estimate}. Then if
    \begin{equation}
        \label{eq:L2-boundary}
        \sum_{\mbj} \ledge (\vj \cdot \bmn)f_{h,\mbj},
        f_{h,\mbj} \redge_{\partial^{+}_{\mbj}X}
        \geq
        \sum_{\mbj} \ledge (\vj \cdot \bmn)  f^{\mathrm{bc}},
        f^{\mathrm{bc}}
        \redge_{\partial^{-}_{\mbj}X},
    \end{equation}
    it follows that $ \frac{\partial}{\partial t} \left\lVert f_{h}
        \right\rVert_{L^2(X\times V)}^2 \leq 0$.
\end{prop}
\begin{proof}
    Case (i) is trivial.  The proof for case (ii) follows immediately from the
    fact that $\sum_{\mbj} w_{\mbj} (\nhj,\fhj - \nhj)_{\Kx} = N_h \sum_{\mbj}
        w_{\mbj} (1,\fhj - \nhj)_{\Kx} = 0$.
\end{proof}

\subsection{Time integration}
\label{sec:time-integrators}

For the purposes of this paper, we consider \ac{dirk} schemes, although the
solver strategy described in the next section applies to other implicit
approaches, including \ac{bdf} and Crank-Nicholson. An $S$-stage \ac{dirk}
method is characterized by vectors $\mbb,\mbc \in \bbR^{S}$ and a lower
triangular matrix $\mbA \in \bbR^{S \times S}$. We use three different schemes;
the Butcher tableaus containing $\mbb$, $\mbc$, and $\mbA$ for these schemes
are given in Appendix \ref{sec:time-integrators}.

To present the method compactly, we write the phase space discretization
\eqref{eq:bgk} formally as an ODE for $f$:
\begin{equation}
    \label{eq:bak-ode}
    \p_t \fhj = \cF_{h,\mbj}(f_h,t),  \qquad \fhj|_{t=0} = \fhj^{\rm{ic}},
\end{equation}
where $\fhj^{\rm{ic}} \in \Wx$ satisfies $(\fhj^{\rm{ic}},\zeta_h)_{\Kx} =
    (f^{\rm{ic}}(\cdot, \vj),\zeta_h)_{\Kx}$.
When applied to \eqref{eq:bak-ode}, an
$S$-stage \ac{dirk} method with time step $\dt$ takes the form
\begin{subequations}
    \label{eq:dirk}
    \begin{align}
        \label{eq:dirk-0}
        f_h^{0}
         & = f_h^{\rm{ic}},
        \\ \label{eq:dirk-n}
        \fhj^{k+1}
         & = \fhj^k + \dt \sum_{s=1}^{S} b_s \Fhj(\fhs,t^{(s)}) , \quad k = 0, 1, 2, \dots,
    \end{align}
\end{subequations}
where, for each stage $s \in \{1 \ddd S\}$, $t^{(s)} = t^k + c_{s} \dt$ and
\begin{equation}
    \label{eq:dirk-stages}
    \fhjs
    = \fhj^k + \dt \sum_{r=1}^s A_{s,r} \Fhj(f^{(r)},t^{(r)}).
\end{equation}
From the definition of $\Fhj$, the stage equations in \eqref{eq:dirk-stages} can
be written in the steady-state form
\begin{equation}
    -(\fhjs, \vj \cdot \grad_\bmx \zeta_h)_{\Kx}
    + (\hat f_{h,j}^{(s)},\vj \cdot  \bmn \zeta_h)_{\p \Kx}
    + \nu^{(s)} (\fhjs,\zeta_h)
    = \nu (\mhjs,\zeta_h)_{\Kx}
    + (\qhjs,\zeta_h)_{\Kx},
    \label{eq:bgk-ss}
\end{equation}
where
\begin{subequations}
    \begin{align}
        \label{eq:lambda-stages}
        \nu^{(s)}
         & = \nu + \frac{1}{\dt A_{s,s}}\quad\text{and}
        \\
        \label{eq:q-stages}
        \qhjs
         & = \frac{1}{\dt A_{s,s}}
        \left(\fhj^k + \dt \sum_{r=1}^{s-1} A_{s,r} \Fhj(f_h^{(r)},t^{(r)}) \right).
    \end{align}
\end{subequations}
In practice, $\Fhj(f_h^{(r)},t^{(r)})$ in \eqref{eq:q-stages} can be evaluated by
applying $\Fhj$ to $f^{(r)}$ for $r < s$; alternatively, a recursive definition
of $\Fhj(f^{(r)},t^{(r)})$ can be obtained through \eqref{eq:dirk-stages}:
\begin{equation}\label{eq:dirk-recursive-def}
    \Fhj(f^{(r)},t^{(r)}) = \frac{1}{\dt A_{r,r}}
    \left(
    \fhjr - \fhj^k - \dt \sum_{r'=1}^{r-1} A_{r,r'} \cF(f^{(r')},t^{(r')})
    \right).
\end{equation}
In this work
\eqref{eq:dirk-recursive-def} is used.

%%---------------------------------------------------------------------------%%

%%---------------------------------------------------------------------------%%
\section{Sweep-based Picard iteration}
\label{sec:sol-methods}

Using the definition of $\p_{\mbj}^{\pm}\Kx$ in
\eqref{eq:cell-inflow-outflow-boundaries}, we write \eqref{eq:bgk-ss} as
\begin{multline}
    \label{eq:bgk-ss-inflow-outflow}
    -(\fhjs, \vj \cdot \grad_\bmx \zeta_h)_{\Kx}
    + (\hat f_{h,{\mbj}}^{(s)},\vj \cdot  \bmn \zeta_h)_{\p_{\mbj}^+ \Kx}
    + \nu^{(s)} (\fhjs,\zeta_h)_{\Kx} \\
    = \nu (\mhjs,\zeta_h)_{\Kx}
    + (\qhjs,\zeta_h)_{\Kx}
    - (\hat f_{h,{\mbj}}^{(s)},\vj \cdot  \bmn \zeta_h)_{\p_{\mbj}^-\Kx}.
\end{multline}
If the right-hand side of \eqref{eq:bgk-ss-inflow-outflow} is given, then the
left-hand side can be inverted to find $\fhjs$ independently across all $\mbj$.
This observation motivates the sweeping strategy in which, for each $\vj$, the
spatial cells are visited in an ordered way  so that the outflow $\hat
    f_{h,{\mbj}}^{(s)}|_{\p_\mbj^{+}\Kx}$ becomes the inflow for another cell later
in the ordering.  However, the sweeping is precluded by three issues
\begin{itemize}
    \item Evaluation of the discrete Maxwellian $\mhj$ requires the entire set
    of discrete velocities to compute moments. This introduces coupling in
    $\mbj$.
    \item If $\Kx$ is a boundary cell, i.e., $\p {K_x} \cap \p X \ne \emptyset$,
    then the inflow boundary data given in \eqref{eq:bgk-bc} may depend on the
    outflow. In such cases, $\hat f_{h,{\mbj}}^{(s)}|_{\p_\mbj^{-}\Kx}$ may
    depend on $\hat f_{h,{\mbj'}}^{(s)}|_{\p_{\mbj'}^{+}\Kx}$ for values of $\mbj'
    \ne \mbj$. This also causes coupling in $\mbj$.
    \item For each $\mbj$, the sweep ordering can be represented by a directed
    graph, and in some cases, the graph may have a small number of cycles. If
    $\Kx$ is part a cycle, there is no ordering of the cells that ensures the
    inflow $\hat f_{h,{\mbj}}^{(s)}|_{\p_\mbj^{-}\Kx}$ is known when trying to
    invert the left-hand side of \eqref{eq:bgk-ss-inflow-outflow}.
\end{itemize}

To address the scenarios above, we lag the discrete Maxwellian as well as any
inflow from the boundary $\p X$ or from cells that are later in the graph
ordering, which can happen whenever $\Kx$ is part of a cycle. In addition, for
large-scale problems that require domain decomposition, we lag the inflow from
one subdomain in another. For each $\Kx \in \Tx$, let
\begin{equation}
    \p_\mbj^{-}\Kx = \p_\mbj^{-}\Kx^\downarrow + \p_\mbj^{-}\Kx^\uparrow
\end{equation}
be a partition of the inflow boundary into regions where the solution must be
lagged ($\downarrow$) and regions where it is not ($\uparrow$). Then for each
integer $\ell \in \{0,\ldots,\ell_{\max}-1\}$, the Picard iteration for
\eqref{eq:bgk-ss-inflow-outflow} takes the form
\begin{multline}
    \label{eq:bgk-ss-picard}
    -(\fhjslp, \vj \cdot \grad_\bmx \zeta_h)_{\Kx}
    + (\hat f_{h,{\mbj}}^{(s),\ell+1},\vj \cdot  \bmn \zeta_h)_{\p_{\mbj}^+ \Kx}
    + \nu^{(s)}  (\fhjslp,\zeta_h)_{\Kx} \\
    = \nu (M_{h,{\mbj}}^{(s),\ell},\zeta_h)_{\Kx}
    + (\qhjs,\zeta_h)_{\Kx}
    + (\hat f_{h,{\mbj}}^{(s),\ell+1},\vj \cdot  \bmn \zeta_h)_{\p_\mbj^{-}\Kx^\uparrow}
    - (\hat f_{h,{\mbj}}^{(s),\ell},\vj \cdot  \bmn \zeta_h)_{\p_\mbj^{-}\Kx^\downarrow}.
\end{multline}
The iteration is initialized by the outcome of the previous \ac{dirk} stage; namely,
\begin{equation}
    f_{h,{\mbj}}^{(s),\ell=0} =
    \begin{cases}
        f_{h,{\mbj}}^{k},                  & s=0,  \\
        f_{h,{\mbj}}^{(s-1), \ell_{\max}}, & s >0,
    \end{cases}
\end{equation}
and it is terminated after the moments of the iterated kinetic solution satisfy
a convergence criterion based on a prescribed tolerance $\tau > 0$.  For $m\in
\{0,1,2\}$, let $\mbU^{\ell,m}_{h}(\bmx) = \sum_i \mbU^{\ell,m}_{h,i}
\phi_i(\bmx) = (\mbe^m,  f_{h}^{(s),\ell}(\bmx,\cdot))_V$. Then the convergence
criteria is
\begin{equation}
    \label{eq:convergence-criteria}
    \max_{m\in\{0,1,2\}}
    \frac{\|\mbU^{\ell+1,m}_{h}-\mbU^{\ell,m}_{h}\|_{\ell^2}}
    {\|{\mbU^{\ell,m}_{h}}\|_{\ell^2}}
    < \frac{\tau}{\max(1.0,\nu\Delta t)},
\end{equation}
where $\|{\mbU^{\ell,m}_{h}}\|^2_{\ell^2} = \sum_i |\mbU^{\ell,m}_{h,i}|^2$ and
$|\cdot|$ is the Euclidean norm. The inclusion of the collision frequency $\nu$
on right-hand side of \eqref{eq:convergence-criteria} is to ensure the iteration
does not terminate too quickly due to ill-conditioning when $\nu$ is large; see,
for example, \cite{hauck2025high} or \cite[Section I.C]{adams2002fast}.

% %%- - - - - - - - - - - - - - - - - - - - - - - - - - - - - - - - - - - - - -%%
\subsection{Transport sweeps}
\label{sec:transport-sweeps}

Parallel sweeping algorithms for solving \eqref{eq:bgk-ss-picard} have been
well established in the linear transport community
\cite{vermaak_massively_2021,adams_provably_2020,haut_efficient_2019,pautz_parallel_2017,evans_denovo:_2010,baker_sn_1998},
where steady-state equations and/or fully implicit methods are the norm. In the
problems presented here, we are allowed much more flexibility in handling
lagged values in \eqref{eq:bgk-ss-picard} imposed by domain boundaries and
cycles because the extra iterations they require are either damped due to
timestep restrictions or absorbed within the non-linear solver tolerance. Thus,
we are most concerned with developing highly efficient domain-local (on-node)
sweep algorithms.

Sweeping algorithms rely on an ordering of the mesh that follows the flow of
information along each velocity $\vj$. For each $\mbj$, the sweep ordering
depends on $\vj$ and induces a directed graph $\cG_\mbj = (\cV,\cE_\mbj)$. Each
vertex $a \in \cV$ is identified with a cell $\Kx^a$ in the mesh, and any edge
$(a,b) \in \cE_\mbj$ indicates that cell $\Kx^{a}$ is adjacent to and upwind
from cell $\Kx^{b}$ with respect to the velocity $\vj$.

\Cref{fig:sweep-graph} illustrates the how the mesh is ordered in the simple
case that  $\p_\mbj^{-}\Kx^\downarrow = \emptyset$; that is, prescribed inflow
at the boundary and no cycles.  The choice of $\vj$ in
\Cref{fig:unstruct-quad-sweep} induces the graph shown in
\Cref{fig:unstruct-graph}.  In this case,  $\cG_\mbj$ is a \ac{dag}, and a
\ac{dfs} algorithm provides the topological ordering of the cells as illustrated
in \cref{fig:unstruct-graph}. For graphs with cycles, the \ac{dfs} algorithm
detects \emph{back edges}  and removes them, as shown in
\cref{fig:unstructured-cycle-5-6-7}, to create a \ac{dag} \cite[Lemma
23.10]{algorithms_2000}.
\begin{figure}
    \centering
    \begin{subfigure}[b]{0.65\textwidth}
        \includegraphics[width=\textwidth]{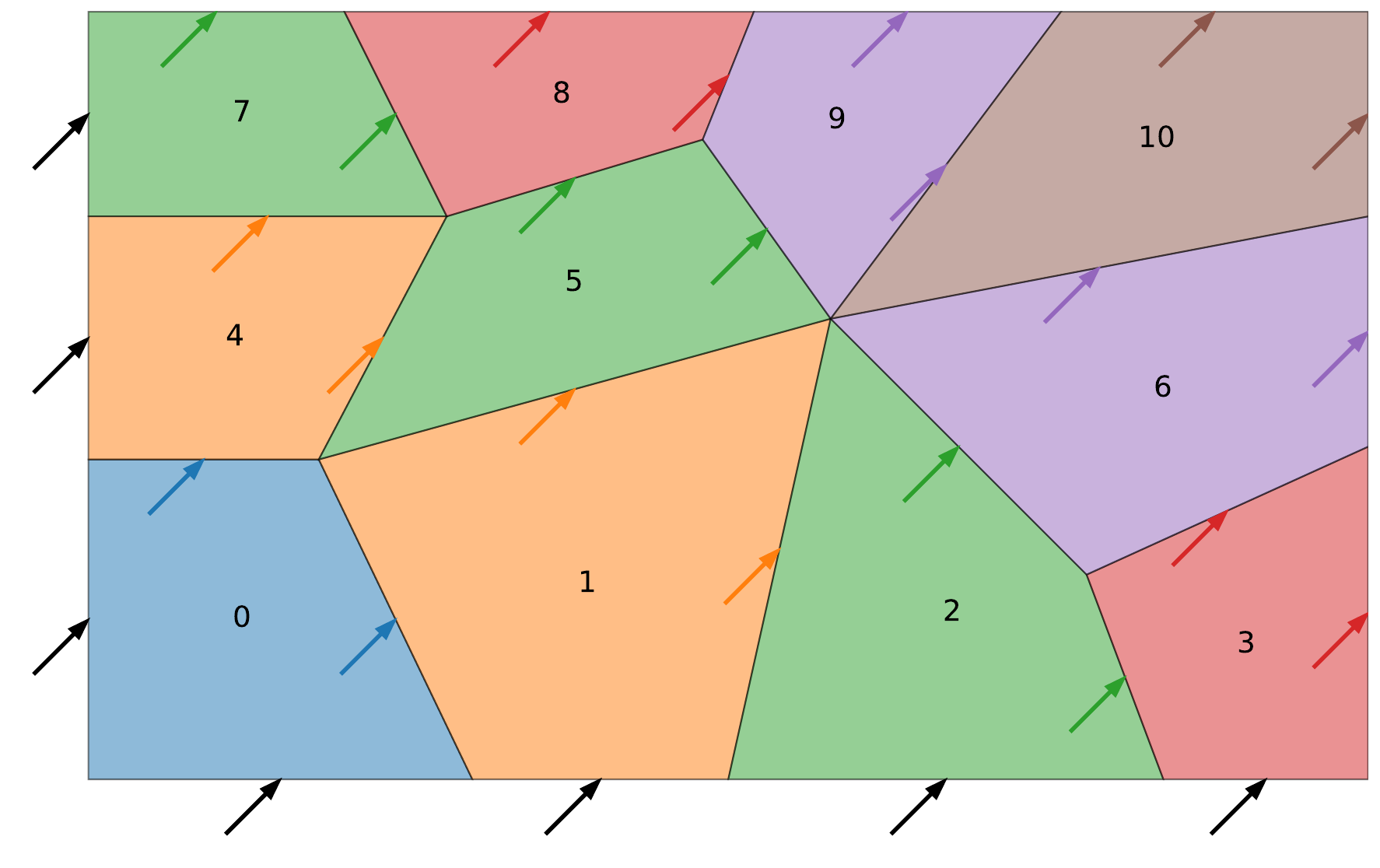}
        \caption{Mesh}
        \label{fig:unstruct-quad-sweep}
    \end{subfigure}
    \hfill
    \begin{subfigure}[b]{0.3\textwidth}
        \includegraphics[width=\textwidth]{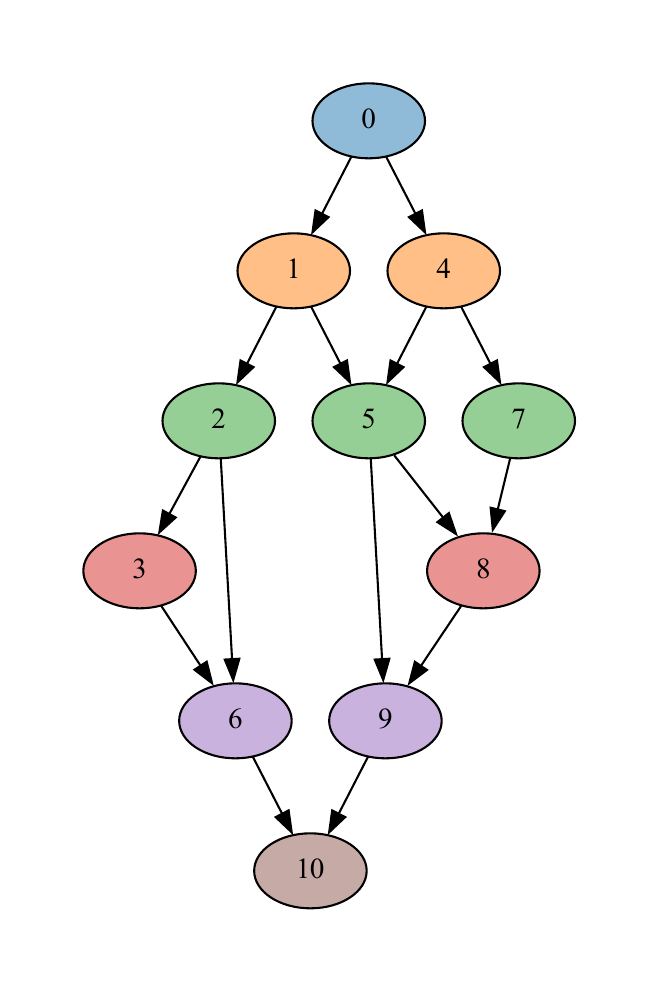}
        \caption{Induced graph}
        \label{fig:unstruct-graph}
    \end{subfigure}
    \caption{Transport sweep with no cycles. The label $i$ for each element
        corresponds to $\Kx^i$. Inflow and outflow faces for each mesh element
        are determined by the velocity $\vj= (0.56, 0.56)$, whose direction
        is indicated by the colored arrows. For this value of $\vj$, the sweep
        begins with $\Kx^0$. In particular, $\partial_\mbj^-\!\Kx^0$ requires
        only data from the boundary (black arrows); hence
        $f_{h,\mbj}^{(s),\ell+1}$ can be determined on $\Kx^0$ by inverting the
        left-hand side of \eqref{eq:bgk-ss-picard}, without using inflow
        information from any other cells. After the solve on $\Kx^0$, the flux
        on $\partial_\mbj^-\!\Kx^i$ is known for $i \in \{1,4\}$ (orange cells);
        thus, $f_{h,\mbj}^{(s),\ell+1}$ can be computed on these elements. The
        sweep continues to the next set of cells $\Kx^i$, $i \in \{2,5,7\}$
        (green cells), for which the inflow data is now known. The process
        continues until \eqref{eq:bgk-ss-picard} is solved for each $\Kx$. The
        process is then repeated for each $\mbj$.}
    \label{fig:sweep-graph}
\end{figure}
\begin{figure}
    \centering
    \begin{subfigure}[b]{0.45\textwidth}
        \centering
        \includegraphics[width=\textwidth]{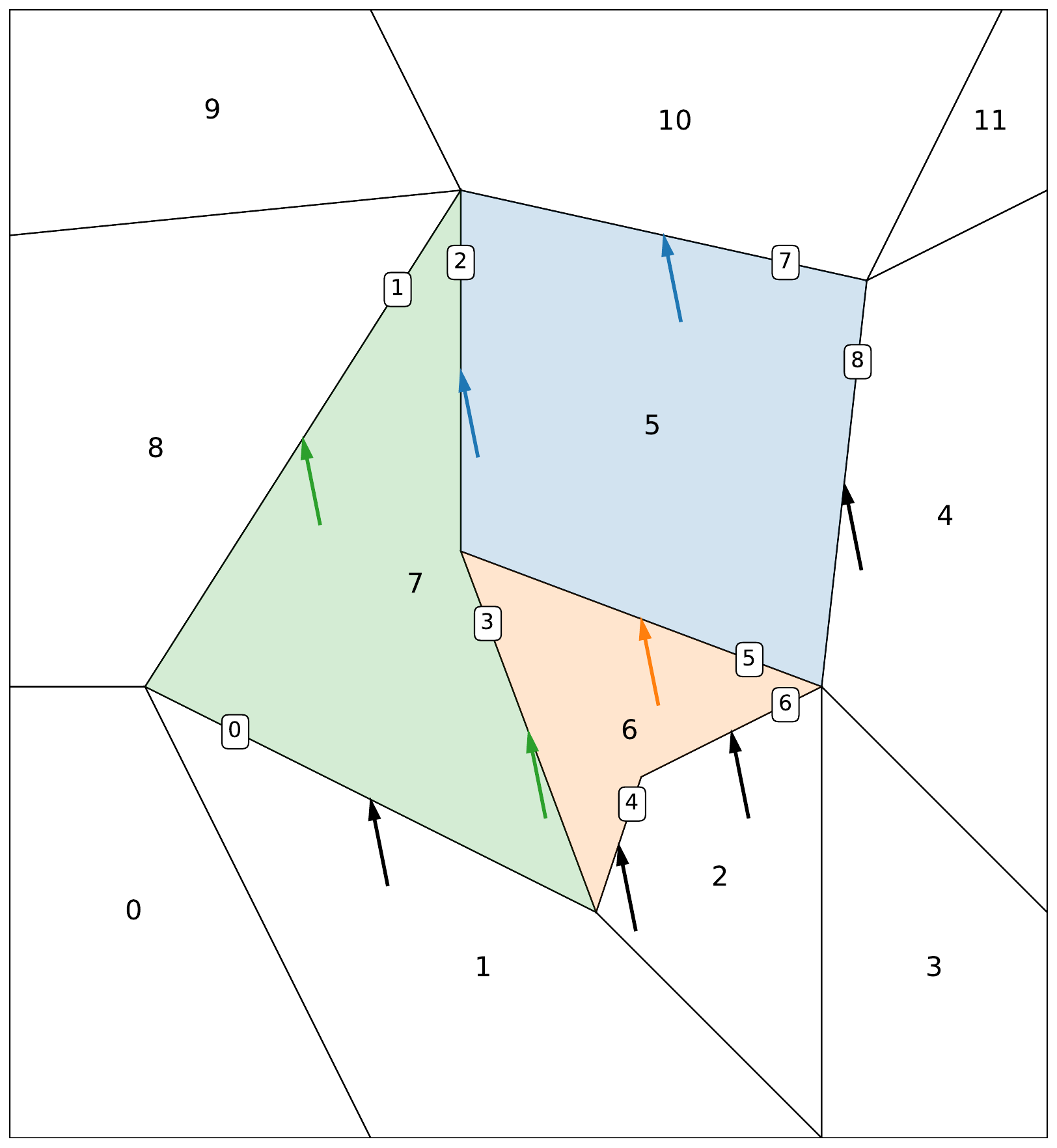}
        \caption{Mesh}
        \label{fig:cycle-5-6-7}
    \end{subfigure}
    \begin{subfigure}[b]{0.5\textwidth}
        \centering
        \includegraphics[width=.65\textwidth]{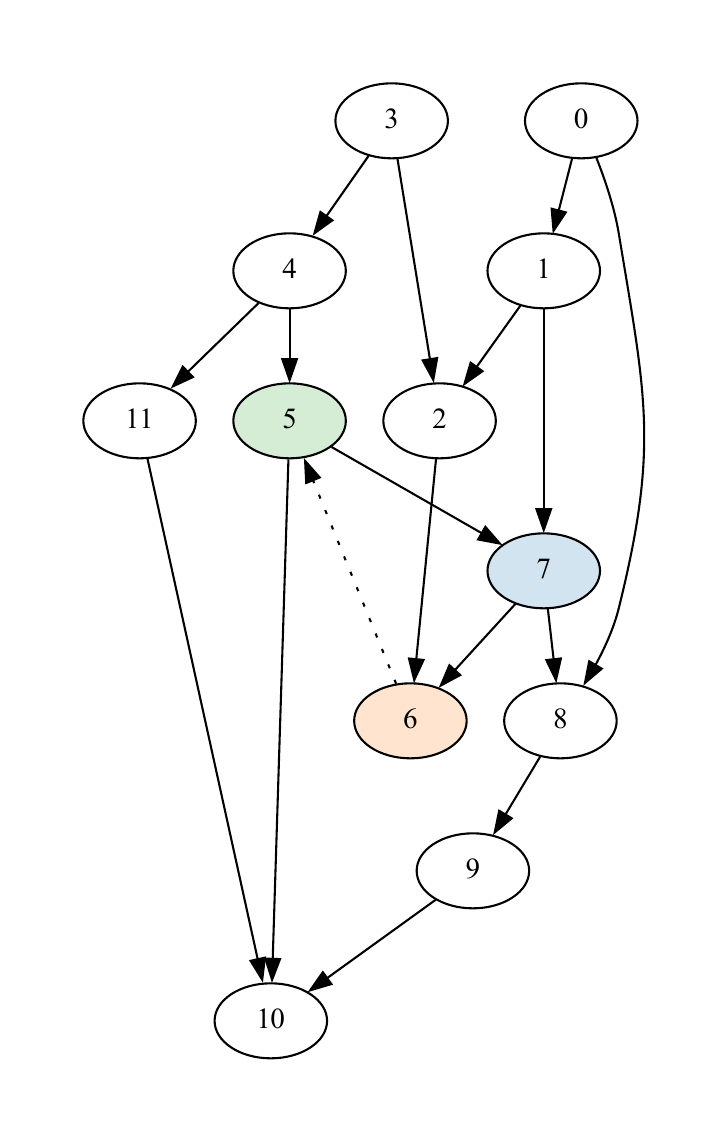}
        \caption{Induced graph}
    \end{subfigure}
    \caption{Transport sweep with a cycle.  The label $i$ for each element
        corresponds to the cell $\Kx^i$, while face $e^l$ between cells is
        identified by an integer $l$ inside a square box.  As before, the sweep
        velocity $\vj$ is indicated by colored arrows.  The cell $\Kx^5$
        receives fluxes from $\Kx^6$ (via face $e^5$); cell $\Kx^7$ receives
        fluxes from $\Kx^5$ (via $e^2$); and cell $\Kx^6$ receives fluxes from
        $\Kx^7$ (via $e^3$).  Hence  $\Kx^5\to\Kx^7\to\Kx^6 \to \Kx^5$ forms a
        cycle with this value of $\vj$.  The \ac{dfs} detects a cycle in edge
        $e^2 = (6, 5)$ of the induced graph. In order to perform the sweep, the
        cycle is broken by removing this edge and setting the numerical flux on
        $\partial_\mbj^-\!\Kx^5\cap\partial_\mbj^+\!\Kx^6$ using data from
        $f_{h,\mbj}^{(s),\ell}$; this action converts the graph to a \ac{dag}.}
    \label{fig:unstructured-cycle-5-6-7}
\end{figure}

We implement two platform portable sweep algorithms, both of which are
illustrated in \cref{fig:tg-v-tgs}. In the \ac{ts}, each thread traverses the
complete graph $\cG_\mbj$ induced by a velocity $\bmv_\mbj$, independently for
each $\mbj$. In the \ac{tgs}, an outer loop is performed over graph generations,
and in each generation the set of velocity points and cells are solved
concurrently. Here, the generation $\operatorname{gen}_\bmj(b)$ of any vertex $b
\in \cV$ is defined as the length of the longest directed path from any source
vertex, i.e.,
\begin{equation}
    \operatorname{gen}_\mbj(b)=
    \begin{cases}
        0,                                                       & b \in \cS_\mbj,   \\
        1+\max_{(a,b)\in \cE_{\bmj}} \operatorname{gen}_\mbj(a), & \text{otherwise}.
    \end{cases}
\end{equation}
The set $\cS_\mbj$ of source vertices, with respect to $\vj$, is given by
\begin{equation}
    \cS_\mbj = \{ a \in \cV :
    \text{data $f_{h,\mbj}^{(s),\ell+1}$ on $\p_\mbj^{-}{\Kx^a}^\uparrow$
        is known at the beginning of the sweep} \}.
\end{equation}
The set of all generation $g$ vertices is denoted by
\begin{equation}
    \operatorname{Gen}_\mbj(g)
    =\{a \in \cV : \operatorname{gen}_\mbj(a) = g\},
\end{equation}
and the generations for each graph are calculated using a longest-distance
algorithm on the \acp{dag}.

The total concurrent work in the \ac{ts} algorithm is
$\text{W}\textsuperscript{TS}=|\mbJ|=\dofv$, whereas the concurrent work in the
\ac{tgs} algorithm for generation $g$ is
\begin{equation}
    \text{W}\textsuperscript{TGS}(g)=\sum_\mbj |\operatorname{Gen}_\mbj(g)|,
\end{equation}
where $|\operatorname{Gen}_\mbj(g)|$ is the cardinality of
$\operatorname{Gen}_\mbj(g)$. For example, in \cref{fig:tg-v-tgs}
$|\operatorname{Gen}_{0,0}(3)|=3$ and $|\operatorname{Gen}_{4,0}(3)|=2$ for the
first two velocity points, respectively.
\begin{figure}
    \centering
    \includegraphics[width=\textwidth]{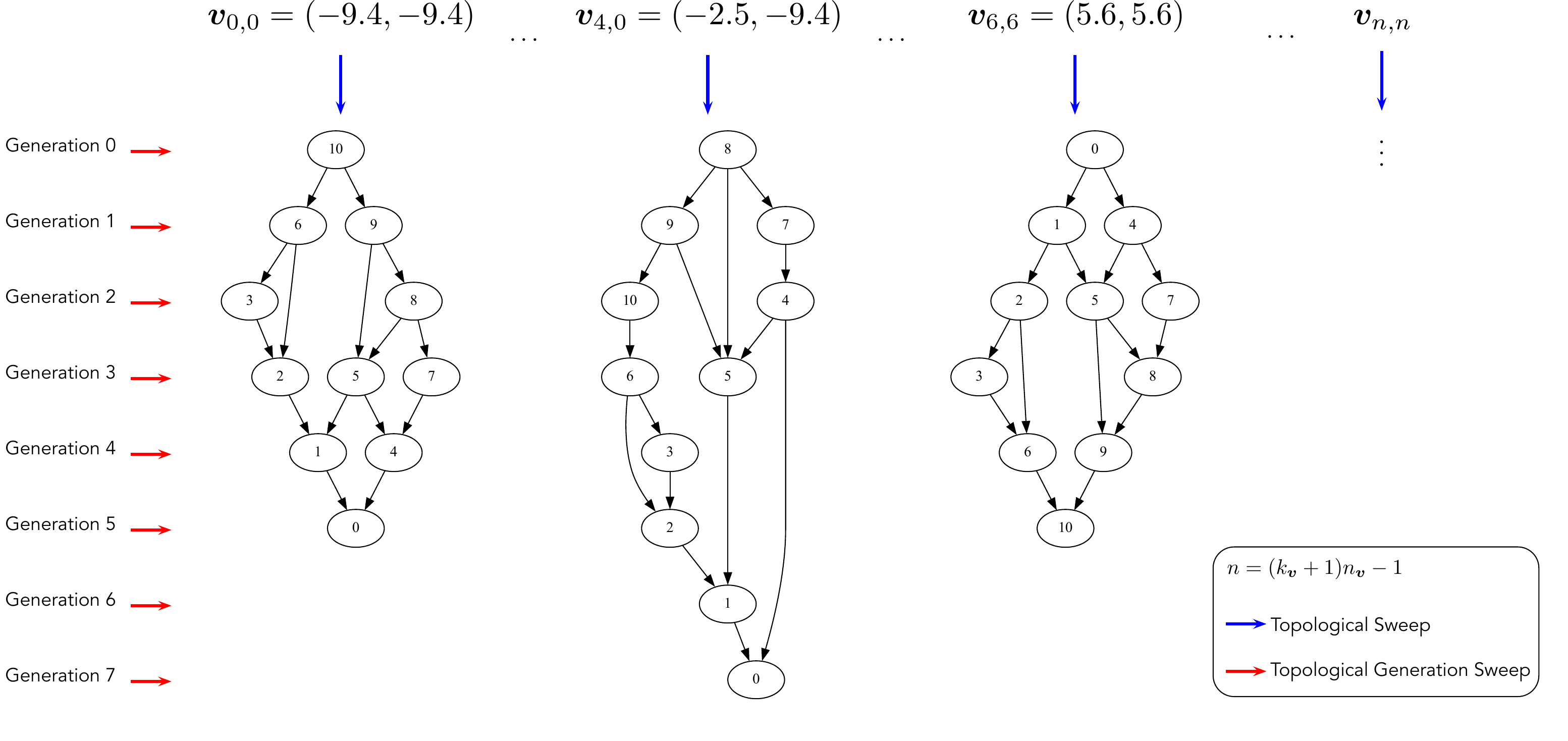}
    \caption{\acs{ts} versus \acs{tgs} sweep strategies using the unstructured
        grid in \cref{fig:unstruct-quad-sweep}.}
    \label{fig:tg-v-tgs}
\end{figure}

%%- - - - - - - - - - - - - - - - - - - - - - - - - - - - - - - - - - - - - -%%
\subsection{Boundary conditions}
\label{sec:boundary-conditions}

We consider three types of boundary conditions:
\begin{enumerate}
    \item Inflow condition
          \begin{equation}
              f^{\rm{bc}}(\bm{x},\bm{v},t) = g(\bm{x},\bm{v},t),
          \end{equation}
          where $g$ is given and independent of $f$.
    \item Bounce-back condition
          \begin{equation}
              f^{\rm{bc}}(\bm{x},\bm{v},t)= f(\bm{x},-\bm{v},t).
          \end{equation}
    \item Reflection condition
          \begin{equation}
              f^{\rm{bc}}(\bmx,\bmv,t)= f(\bm{x},R_{\bmn(\bmx)}(\bmv),t),
          \end{equation}
          where $R_{\bmn}(\bmv) = \bmv - 2(\bmv\cdot\bmn)\bmn$.
\end{enumerate}

\begin{prop}
    For any smooth function $\phi = \phi(\bmv)$ and any $\bmx \in \p X$ with
    well-defined normal $\bmn(\bmx)$, the bounce-back condition implies that
    \begin{equation}
        \label{eq:bounce-back-moments}
        \int_{\bbR^3} \phi(\bmv) f(\bm{x},\bm{v},t) d\bmv
        = \int_{\bmv \cdot \bmn(\bmx) > 0} [\phi(\bmv) + \phi(-\bmv)] f(\bm{x},\bm{v},t) d \bmv,
    \end{equation}
    and the reflection boundary condition implies that \begin{equation}
        \label{eq:reflection-moments}
        \int_{\bbR^d} \bmv \cdot \bmn(\bmx)
        \phi(\bmv) f(\bm{x},\bm{v},t) d \bmv
        = \int_{\bmv \cdot \bmn(\bmx) > 0} \bmv \cdot \bmn(\bmx) [\phi(\bmv) -
            \phi(R_{\bmn(\bmx)}(\bmv))] f(\bm{x},\bm{v},t) d \bmv.
    \end{equation}
    In particular, if $n(\bmx) > 0$%
    \footnote{$n(\bmx)$ is the density at $\bmx$; $\bmn(\bmx)$ is the normal.}
, then the bounce back condition implies
    $\bmu(\bmx) = 0$ and the reflection condition implies that $\bmn(\bmx) \cdot
        \bmu(\bmx) = 0$.
\end{prop}

\begin{proof}
    Fix $\bmx$ and let suppose that $f(\bm{x},\bm{v},t) = f(\bm{x},T \bm{v},t)$ whenever $\bmv \cdot \bmn(\bmx) <0$, where $T:\bbR^d \to \bbR^d$ is any linear transformation such
    that $T^2=I$, $(T \bmv) \cdot \bmn(\bmx) = - \bmv \cdot \bmn(\bmx)$,  The
    mappings $\bmv \mapsto -\bmv$ and $\bmv \mapsto R_{\bmn(\bmx)} \bmv$ both
    satisfy these conditions.  Let $\psi = \psi(\bmv)$.  Then it is easy to show
    that
    \begin{equation}
        \int_{\bbR^d} \psi(\bmv) f(\bm{x},\bm{v},t) d \bmv
        = \int_{\bmv \cdot \bmn(\bmx) > 0} [\psi(\bmv) +
            \psi(T \bmv)] f(\bm{x},\bm{v},t) d \bmv.
    \end{equation}
    Setting $\psi = \phi$ and $T = -I$ recovers \eqref{eq:bounce-back-moments}.
    In particular, if $\phi(\bmv) = \bmv$, then $\bmu (\bmx)  = 0$.  Setting
    $\psi(\bmv) = \bmv \cdot \bmn(\bmx) \phi(\bmv)$ and  $T = R_{\bmn}(\bmx)$
    recovers \eqref{eq:reflection-moments}.  In particular, if $\phi(\bmv) = 1$,
    then $\bmu(\bmx) \cdot  \bmn(\bmx) = 0$.
\end{proof}

The conditions $\bmu (\bmx) = 0$ and $\bmu(\bmx) \cdot \bmn(\bmx) = 0$
correspond to standard no-slip and free-slip boundary conditions for fluid
models. For simplicity, we assume in this work rectangular spatial domains such
that the normal $\bmn(\bmx)$ at every point $\bmx \in \p X$ is aligned with one
of the coordinates axes. Moreover, we assume for any interpolation point $\vj$,
that $-\vj$ and $R_{\bmn}(\vj)$ are also interpolation points. This assumption
ensures that \eqref{eq:bounce-back-moments} and \eqref{eq:reflection-moments}
hold at the discrete velocity level. After discretizing in space,
\eqref{eq:reflection-moments} will hold at each interpolation point on $\p X$.
In particular the number density flux $\phi(\bmv)=1$, and the energy density
flux $\phi(\bmv)=\frac12|\bmv|^2$ will be conserved. However, the conditions in
\eqref{eq:bounce-back-moments} can only be enforced weakly.

In future work, more general scenarios will be considered. This includes cases
for which the normal $\bmn(\bmx)$ does not align with the $\bmx$-coordinate
axes, as well as moving geometries and internal structures.

%%---------------------------------------------------------------------------%%

%%---------------------------------------------------------------------------%%
\section{Numerical Results}
\label{sec:results}

The complete discretization and iterative solver is implemented in the Spinnaker
\ac{cfd} codebase using the platform-portable Kokkos library
\cite{kokkos,kokkos-ecosystem}. This allows the sweep-based \ac{bgk} solver to
run optimally on serial, OpenMP, and GPU (\nvidia, \amd) devices, as discussed
in \cref{sec:solver-performance}. \amd GPU results were performed on the
Frontier supercomputer at the \ac{olcf} \cite{frontier}. Each node of Frontier
contains 4~\ac{mi250x} GPUs, each with 2 \acp{gcd}, which yields an effective
equivalent of 8 GPUs per node. In all Frontier simulations, GPU runs were
executed with 1 MPI rank per \ac{gcd}. Each \ac{mi250x} \ac{gcd} has
\SI{64}{\giga\byte} of memory. Full descriptions of the Frontier system are
available in the Frontier User Guide \cite{frontier-user-guide}. CPUs in all
cases are \amd EPYC 64- and 96-core processors, split into \ac{numa} partitions
of 16 cores each. Each physical core has 2 hardware threads; however, all the
results that follow were obtained utilizing only a single thread per physical
core. \nvidia results were executed on local clusters containing V100
(\SI{32}{\giga\byte}), A100 (\SI{80}{\giga\byte}), and H100
(\SI{80}{\giga\byte}) cards.

The stability restraint for explicit \ac{rk} \ac{dgfem} schemes is
\begin{equation}
    \dt_\mathrm{e} \le
    \frac{1}{d} \frac{1}{2 \kx+1}\frac{h^{\min}_\bmx}{\max_\mbj |\vj|},
    \label{eq:explicit-dt}
\end{equation}
where, unless otherwise defined,
$h^{\min}_\bmx\sim(\text{vol}^{\min}_\bmx)^{1/d}$ and $L$ is the size of
computational velocity domain $V$ as described in
\cref{sec:velocity-discretization}, i.e., $V = [-L,L]^d$.  All simulations use a
constant timestep, $\dt = C\dt_\mathrm{e}$, in which $C>1$ is specified at
runtime.

In Section~\ref{sec:sod}, we consider variations of the Sod shock problem
\cite{toro2013riemann}. We perform simulations in 2-D and 3-D geometries using
triangular, quadrilateral, and tetrahedral spatial elements and compare to
solutions of the corresponding Euler equations in collisional regimes.
\Cref{sec:duct} shows flow through a duct, and \cref{sec:solver-performance}
analyzes solver performance on different architectures.

%%---------------------------------------------------------------------------%%
\subsection{Sod shock problems}
\label{sec:sod}

All variations of the Sod problem involve a shock wave propagating from a high
($H$) density region to a low density region ($X\setminus H$), and the velocity
domain $V = [-7,7]^d$. In all versions of the problem, the initial condition is
a local equilibrium: $f^{\rm{ic}} = M_{\nic,\uic,\tic}$, where
\begin{equation}
    \nic(\bm{x}) = \begin{cases}
        1     , & \bm{x} \in H,             \\
        0.125 , & \bm{x} \in X \setminus H,
    \end{cases}
    \qquad
    \uic(\bm{x}) = \begin{pmatrix}
        0 \\ 0 \\ 0
    \end{pmatrix} \, \quad \bm{x} \in X,
    \qquad
    \tic(\bm{x}) = \begin{cases}
        1   & \bm{x} \in H,             \\
        0.8 & \bm{x} \in X \setminus H.
    \end{cases}
\end{equation}
Unless otherwise stated, the boundary condition is
\begin{equation}
    f^{\rm{bc}}(\bm{x},\bm{v},t) =
    f^{\rm{ic}}(\bm{x},\bm{v}).
\end{equation}

%% - - - - - - - - - - - - - - - - - - - - - - - - - - - - - - - - - - - - - %%
\subsubsection{1D Sod and Maxwellian projection}
\label{sec:1d-sod}

The 1D Sod problem is embedded in a 2D geometry
\begin{equation}
    \begin{aligned}
        X & = \{\bm{x} = (x_1,x_2) : -L_1 \leq x_1 \leq L_1,
        -L_2 \leq x_2 \leq L_2\},                                                    \\
        H & = \{\bm{x} = (x_1,x_2) : -L_1 \leq x_1 \leq 0, -L_2 \leq x_2 \leq L_2\},
    \end{aligned}
\end{equation}
where $(L_1,L_2) = (1,3/128)$.   The boundary conditions on the bottom and top
boundaries $\{\bmx = (x_1,x_2) \in \p X: x_2 = \pm L_2 \}$ are reflection, as
defined in \cref{sec:boundary-conditions}. The mesh is a \num{256x6} uniform
grid with using 2D quadrilaterals

In \Cref{fig:1d-sod-results}, we present first- and second-order \ac{dirk}
results for the Sod problem with $\nu \in \{1,1000\}$ and $n_\bmv = 128$. The
problem is run for \num{85} timesteps with $\dt=1.76\times10^{-3}$ to reach a
final time of $t=0.15$. The timestep is approximately $7\times$ the explicit
\ac{dgfem} timestep restriction (\cref{eq:explicit-dt}). Reference solutions in
the fluid limit ($\nu = 1000$) are generated using the \emph{ToroExact} package
\cite{toroexact} that employs the exact Riemann solvers from
\cite{toro2013riemann} with $\gamma = \frac{d+2}{2} = 2$.
\Cref{fig:1d-sod-results-1000-d1,fig:1d-sod-results-1000-d2} show that the
\ac{bgk} solver properly captures the shock fronts in the fluid limit as $\nu
\gg 1$.  \Cref{tab:1d-sod-iterations} summarizes the non-linear Picard
iterations for DIRK-1 and DIRK-2 schemes as a function of $n_\bmv$ and $\nu$. In
all cases, the total number of iterations is relatively insensitive to the
number of velocity points until $n_\bmv=128$. The number of iterations are
roughly \numrange{5.5}{5.8} times greater in the fluid limit compared to the
rarefied gas limit for DIRK-1 and \numrange{2.7}{3.3} times greater for DIRK-2.
\begin{figure}
    \centering
    \begin{subfigure}{0.45\textwidth}
        \centering
        \includegraphics[width=\textwidth]{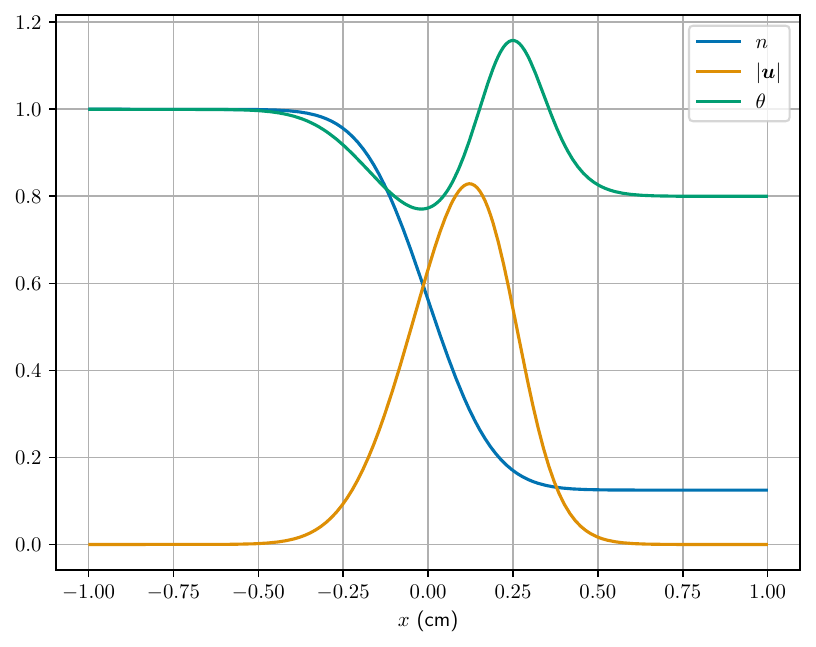}
        \caption{DIRK-1, $\nu = 1$}
    \end{subfigure}
    \begin{subfigure}{0.45\textwidth}
        \centering
        \includegraphics[width=\textwidth]{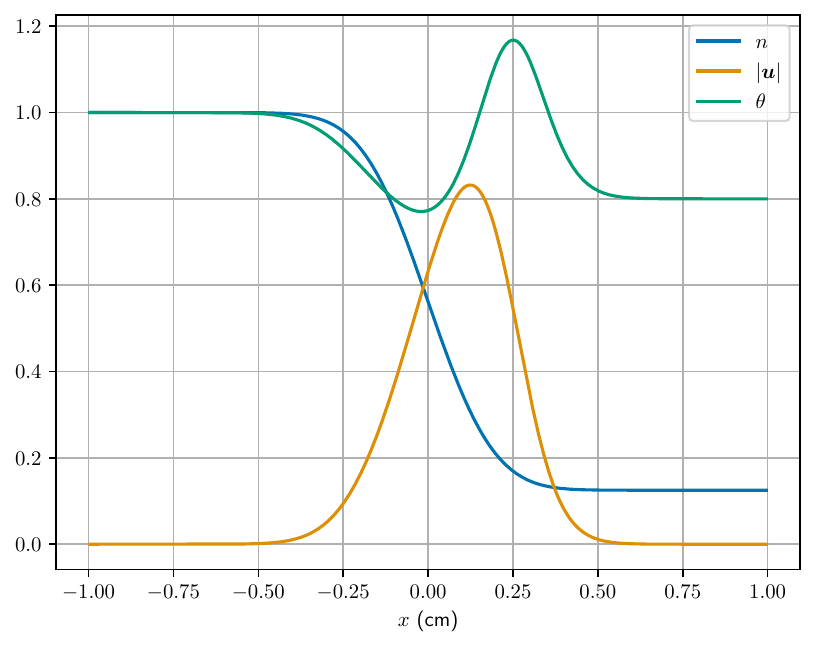}
        \caption{DIRK-2, $\nu = 1$}
    \end{subfigure}\\
    \begin{subfigure}{0.45\textwidth}
        \centering
        \includegraphics[width=\textwidth]{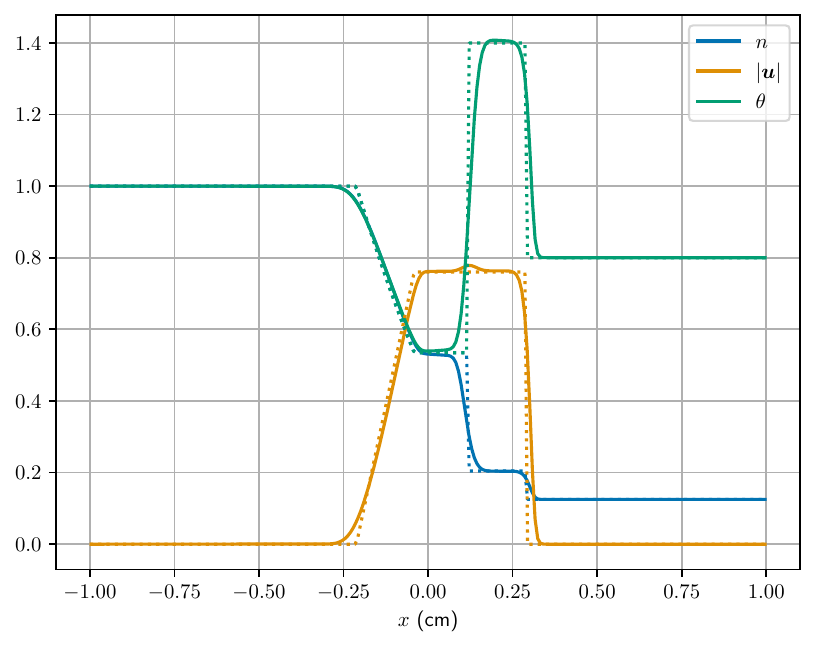}
        \caption{DIRK-1, $\nu = 1000$}
        \label{fig:1d-sod-results-1000-d1}
    \end{subfigure}
    \begin{subfigure}{0.45\textwidth}
        \centering
        \includegraphics[width=\textwidth]{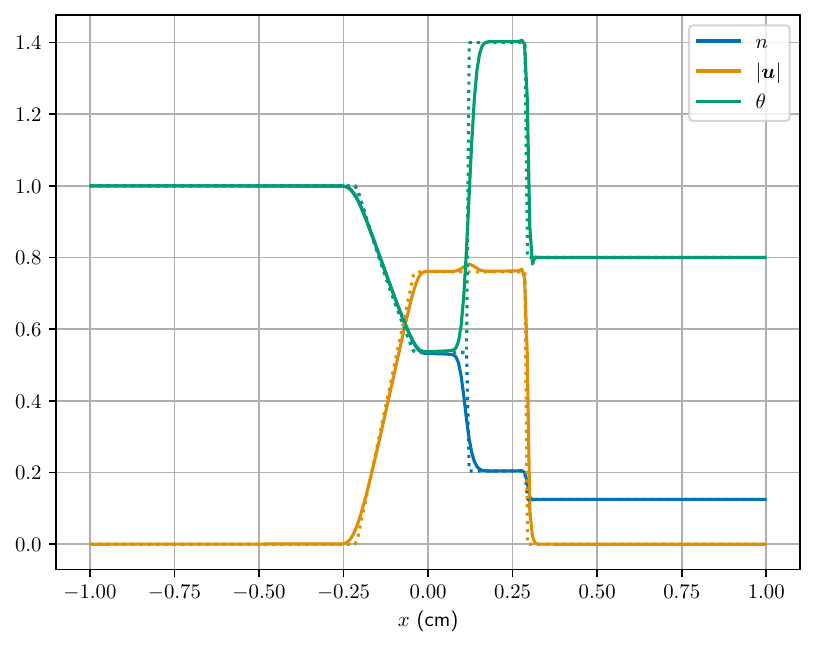}
        \caption{DIRK-2, $\nu = 1000$}
        \label{fig:1d-sod-results-1000-d2}
    \end{subfigure}\\
    \caption{1D Sod Problem results with $\nu \in \{1,1000\}$.  In (c-d)
        dotted lines represent the reference solutions for the Euler equations.}
    \label{fig:1d-sod-results}
\end{figure}

\begin{table}
    \centering
    \caption{Iteration counts for the 1D SOD problem.}
    \label{tab:1d-sod-iterations}
    \begin{tabular}{llllllllll}\toprule
                                                      &          & \multicolumn{6}{c}{per timestep} &    &                                       \\\cmidrule(lr){3-8}
        $\nu $                                        & $n_\bmv$ &
        \multicolumn{2}{c}{min}                       &
        \multicolumn{2}{c}{max}                       &
        \multicolumn{2}{c}{avg}                       &
        \multicolumn{2}{c}{tot}                                                                                                                  \\\cmidrule(lr){3-4} \cmidrule(lr){5-6}
        \cmidrule(lr){7-8} \cmidrule(lr){9-10}
        \multicolumn{2}{l}{$\Delta t\ \text{order}=$} & 1        & 2                                & 1  & 2  & 1  & 2     & 1     & 2           \\\midrule
        1                                             & 4        & 4                                & 7  & 5  & 9  & 4.32  & 7.38  & 367  & 627  \\
        1                                             & 8        & 4                                & 6  & 5  & 9  & 4.22  & 6.75  & 359  & 574  \\
        1                                             & 16       & 4                                & 6  & 5  & 9  & 4.21  & 6.54  & 358  & 556  \\
        1                                             & 32       & 4                                & 6  & 5  & 9  & 4.21  & 6.53  & 358  & 555  \\
        1                                             & 64       & 4                                & 6  & 5  & 9  & 4.21  & 6.53  & 358  & 555  \\
        1                                             & 128      & 4                                & 8  & 7  & 11 & 4.59  & 8.07  & 390  & 686  \\
        \cmidrule(lr){1-10}
        1000                                          & 4        & 23                               & 20 & 30 & 28 & 24.29 & 21.28 & 2065 & 1809 \\
        1000                                          & 8        & 23                               & 20 & 30 & 28 & 24.32 & 21.38 & 2067 & 1817 \\
        1000                                          & 16       & 23                               & 20 & 30 & 28 & 24.32 & 21.39 & 2067 & 1818 \\
        1000                                          & 32       & 23                               & 20 & 30 & 28 & 24.32 & 21.39 & 2067 & 1818 \\
        1000                                          & 64       & 23                               & 20 & 30 & 28 & 24.32 & 21.39 & 2067 & 1818 \\
        1000                                          & 128      & 23                               & 21 & 32 & 28 & 25.06 & 21.84 & 2130 & 1856 \\
        \bottomrule
    \end{tabular}
\end{table}

We repeat the Sod test for different velocity quadratures to show the benefits
of using the projected Maxwellian described in \ref{app:velocity-projection},
which guarantees the discrete \ac{bgk} operator maintains the proper collision
invariants as opposed to a standard interpolation in velocity. In
\cref{fig:1d-sod-projection} we show $L^2$ errors for both the Maxwellian
projection and interpolation schemes as a function of $n_\bmv$ and with $\nu \in
\{1,1000\}$. In the rarefied gas region, the interpolation and projection
schemes perform equally well at $n_\bmv\ge 8$. However, in the fluid limit,
there is a pronounced effect below $n_\bmv = 16$, and the projected scheme
yields roughly an order-of-magnitude smaller error out to $n_\bmv = 64$. When
$\nu = 1$, it takes significantly more velocity degrees of freedom to obtain
accurate results regardless of the Maxwellian approximation used.
\begin{figure}
    \centering
    \begin{subfigure}{0.45\textwidth}
        \centering
        \includegraphics[width=\textwidth]{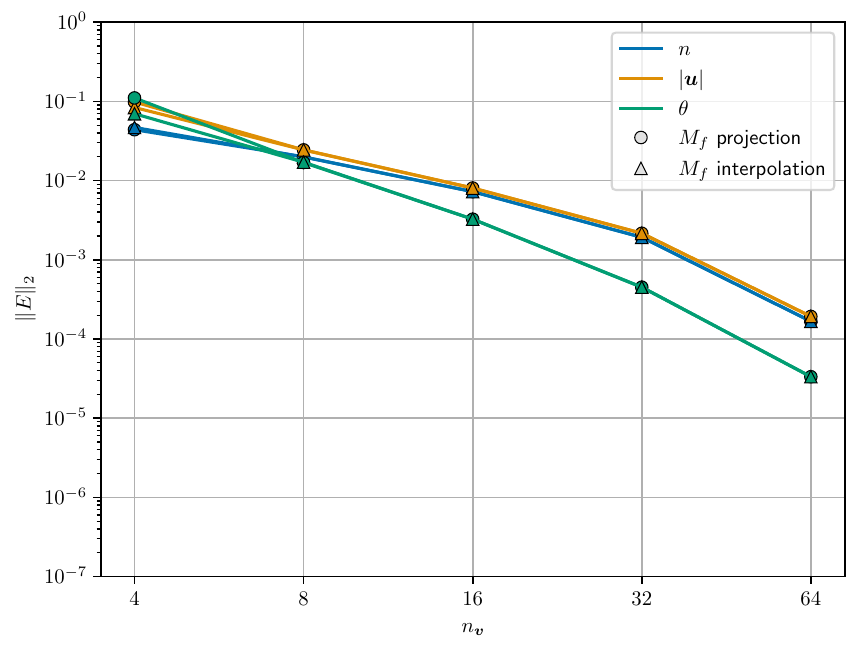}
        \caption{DIRK-1, $\nu = 1$}
    \end{subfigure}
    \begin{subfigure}{0.45\textwidth}
        \centering
        \includegraphics[width=\textwidth]{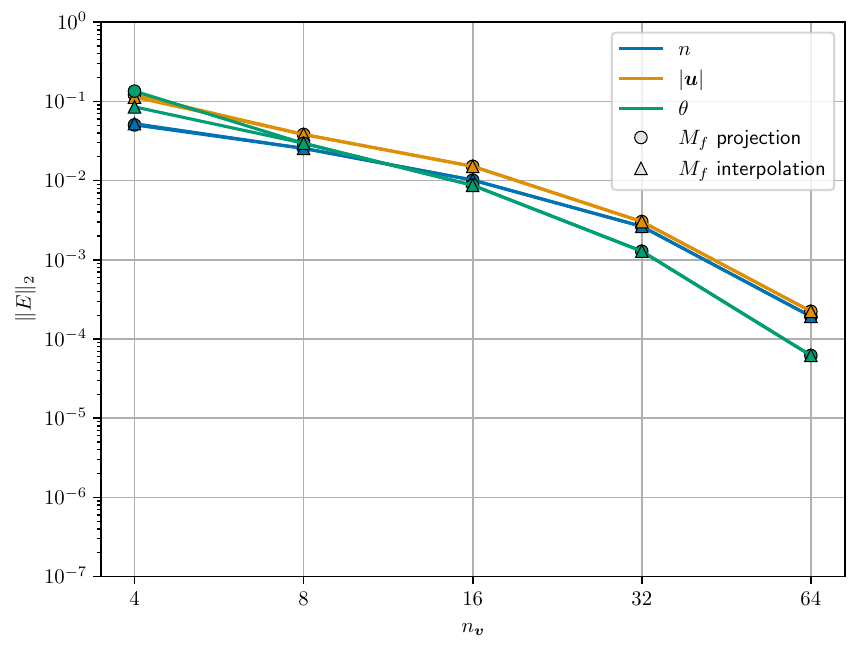}
        \caption{DIRK-2, $\nu = 1$}
    \end{subfigure}\\
    \begin{subfigure}{0.45\textwidth}
        \centering
        \includegraphics[width=\textwidth]{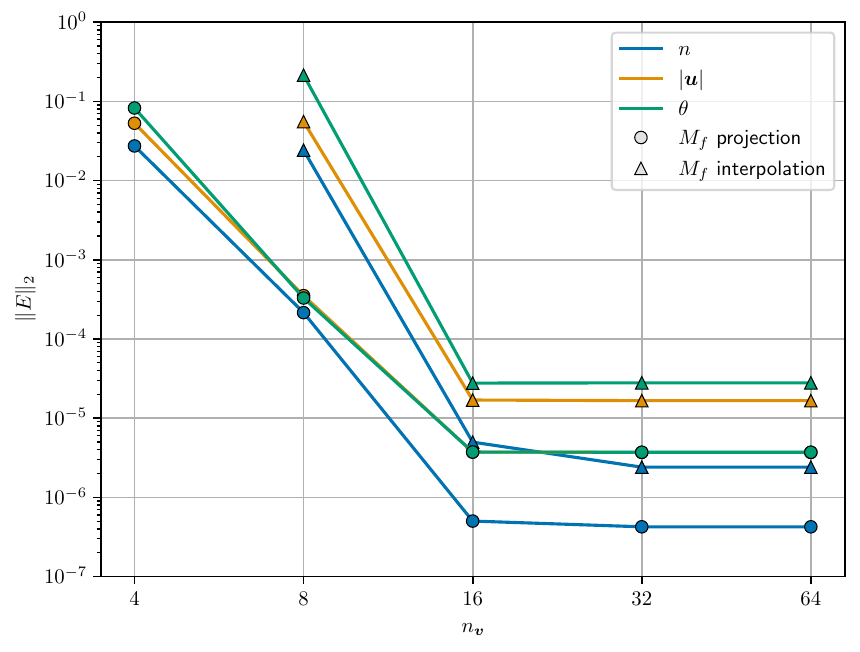}
        \caption{DIRK-1, $\nu = 1000$}
    \end{subfigure}
    \begin{subfigure}{0.45\textwidth}
        \centering
        \includegraphics[width=\textwidth]{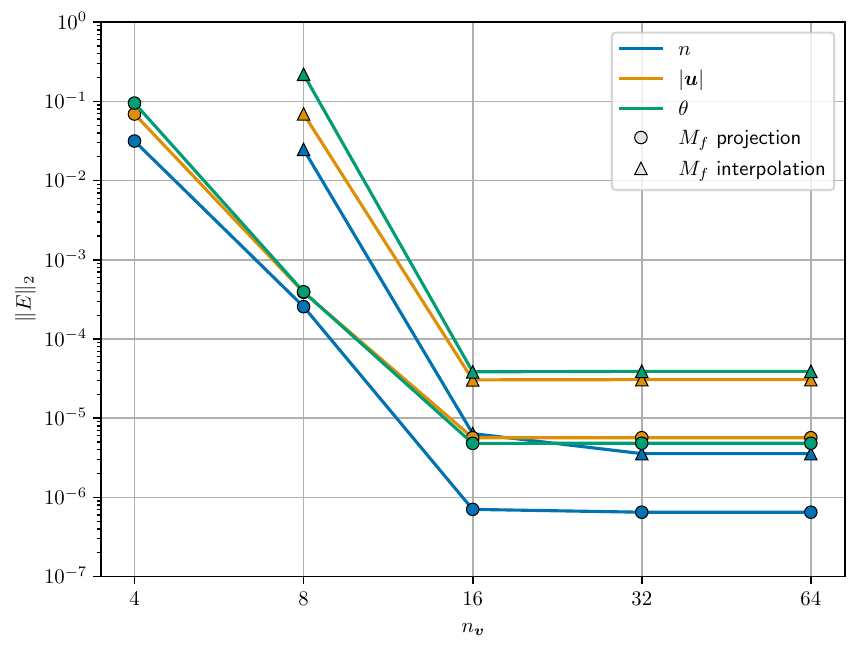}
        \caption{DIRK-2, $\nu = 1000$}
    \end{subfigure}\\
    \caption{1D Sod Problem results with $\nu \in \{1,1000\}$ showing $L^2$
        error norms for Maxwellian projection versus interpolation.  The
        reference cases for the error norms use $n_\bmv=128$ velocity quadrature
        cells.}
    \label{fig:1d-sod-projection}
\end{figure}

%% - - - - - - - - - - - - - - - - - - - - - - - - - - - - - - - - - - - - - %%
\subsubsection{Multi-dimensional Sod and stability}

The multi-dimensional Sod problem is defined on a the spatial domain
$X=[0,2]^d$, and $H$ is a closed ball centered at $\bm{x}_{\rm{c}} =
(1,1,1)^\top$ with radius $r = 0.4$. Reference solutions for highly collisional
regimes were obtained by solving the Euler equations with a gamma-law equation
of state in one spatial dimension, using cylindrical coordinates with $\gamma=2$
for the 2D problem and spherical-polar coordinates with $\gamma=5/3$ for the 3D
problem. Reference solution simulations employed the \ac{rk} \ac{dgfem} as
implemented in \cite{endeve_etal_2019,pochik_etal_2021}. For these calculations,
the computational domain was discretized using \num{2000} quadratic elements,
and time integration was performed using a third-order strong
stability-preserving \ac{rk} scheme.

Kinetic simulation of the 2D Sod problem is performed on triangular and
quadrilateral grids shown in \cref{fig:2d-grids} at $t =0.15$ and $n_\bmv =
    16$. \Cref{fig:2d-quad-results} show results for DIRK-1 and DIRK-2 with
$\nu\in\{1000,10000\}$ and $\Delta t=3.7 \times 10^{-4}$ ($7\times$ the
explicit timestep). As $\nu$ increases, the solution reproduces the fluid limit
reference more closely, and the higher order \ac{dirk} scheme resolves shock
fronts slightly better.

The ability to capture shock fronts is affected by both the resolution and shape
of the grids. In particular, the triangular grids create a sawtooth pattern at
the edges of the shock that can be observed in \cref{fig:2d-grids-tri}. The
quadrilateral grids, on the other hand, are able to smoothly follow the
cylindrical shock fronts as observed in \cref{fig:2d-grids-quad}. Thus, the
triangular grids require nearly four times as many elements as the quadrilateral
grid to achieve similar resolution. \Cref{fig:2d-quad-tri-comparison} shows that
with significant mesh resolution, the triangular grids produce nearly identical
solutions as the quadrilateral grids.
\begin{figure}
    \centering
    \begin{subfigure}{0.48\textwidth}
        \centering
        \includegraphics[width=\textwidth]{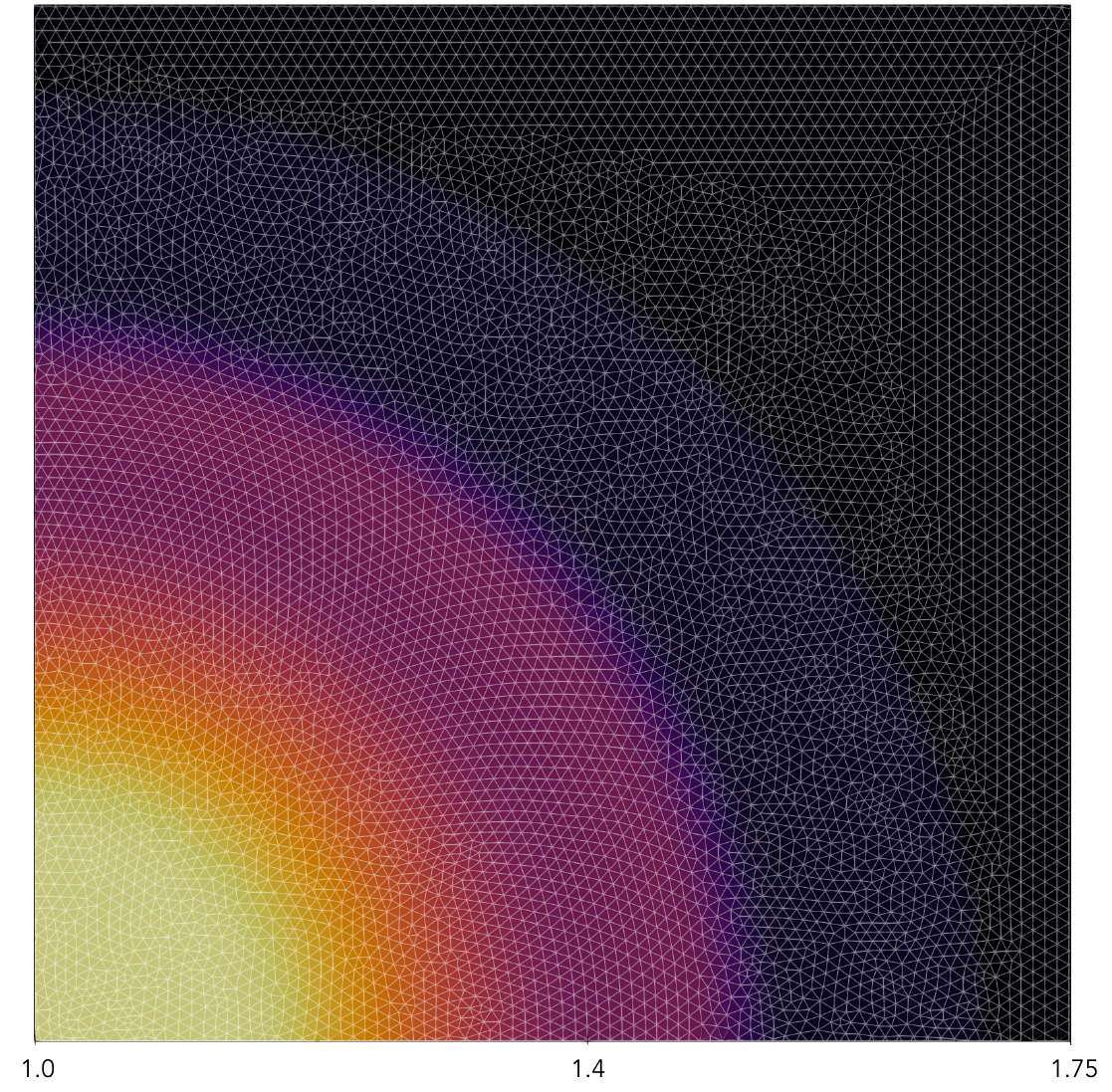}
        \caption{}
        \label{fig:2d-grids-tri}
    \end{subfigure}
    \begin{subfigure}{0.48\textwidth}
        \centering
        \includegraphics[width=\textwidth]{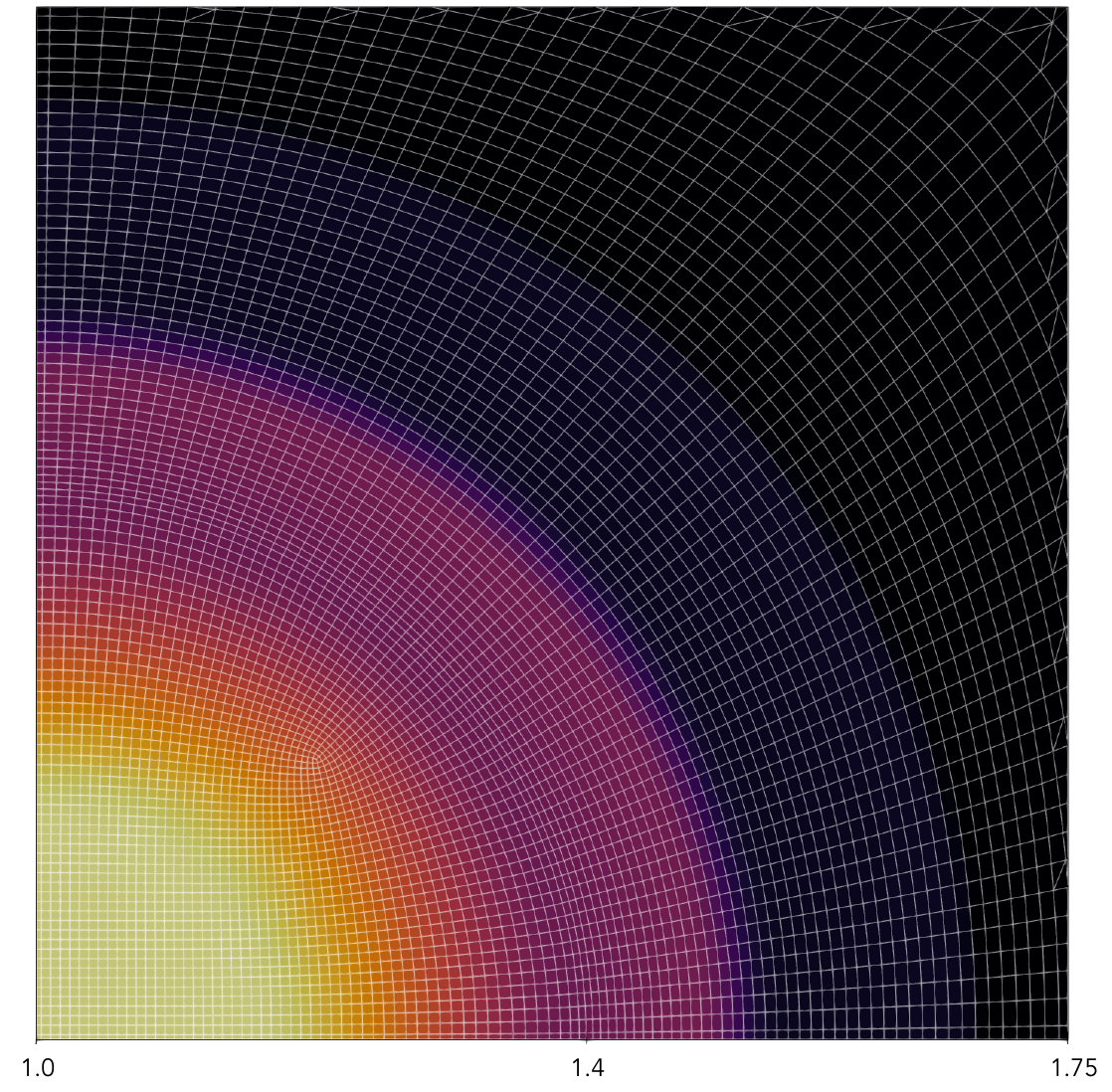}
        \caption{}
        \label{fig:2d-grids-quad}
    \end{subfigure}
    \caption{2D grid slices used for the Sod problem showing density at
        \SI{0.15}{\second}, (a) triangular grid with \num{98880} elements and (b)
        quadrilateral grid with \num{25545} elements.  The full extents of both
        grids are $X = [0,2]^2$.}
    \label{fig:2d-grids}
\end{figure}
\begin{figure}
    \centering
    \begin{subfigure}{0.45\textwidth}
        \centering
        \includegraphics[width=\textwidth]{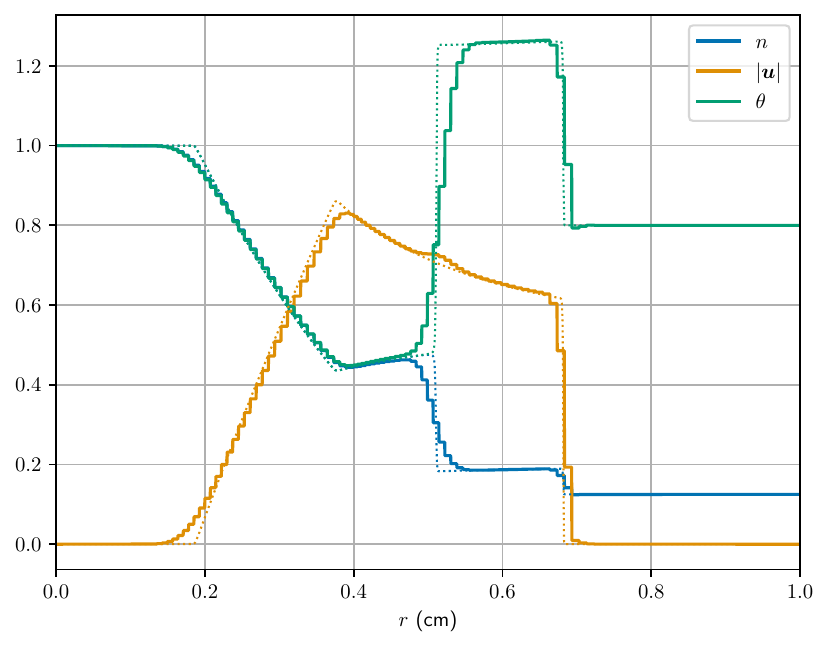}
        \caption{DIRK-1, $\nu = 1000$}
    \end{subfigure}
    \begin{subfigure}{0.45\textwidth}
        \centering
        \includegraphics[width=\textwidth]{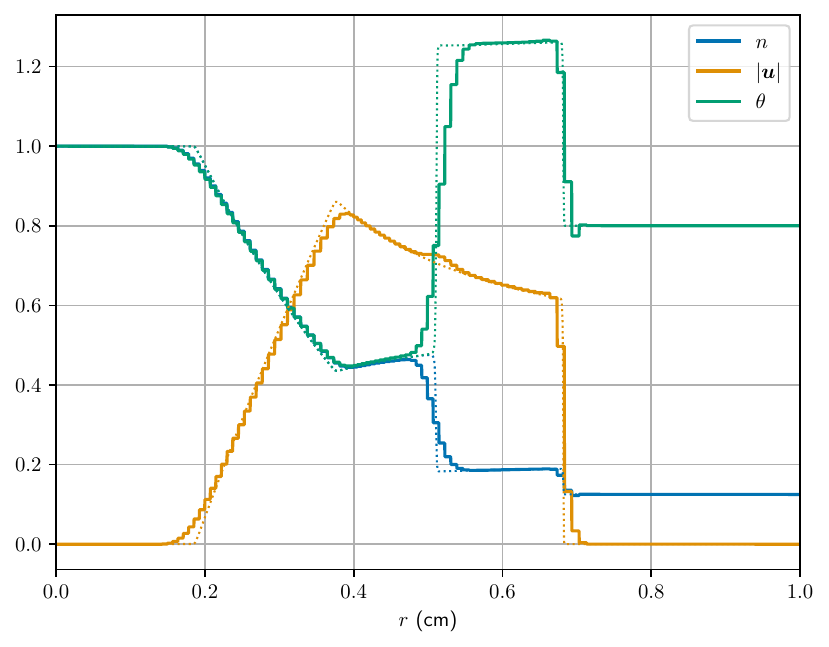}
        \caption{DIRK-2, $\nu = 1000$}
    \end{subfigure}\\
    \begin{subfigure}{0.45\textwidth}
        \centering
        \includegraphics[width=\textwidth]{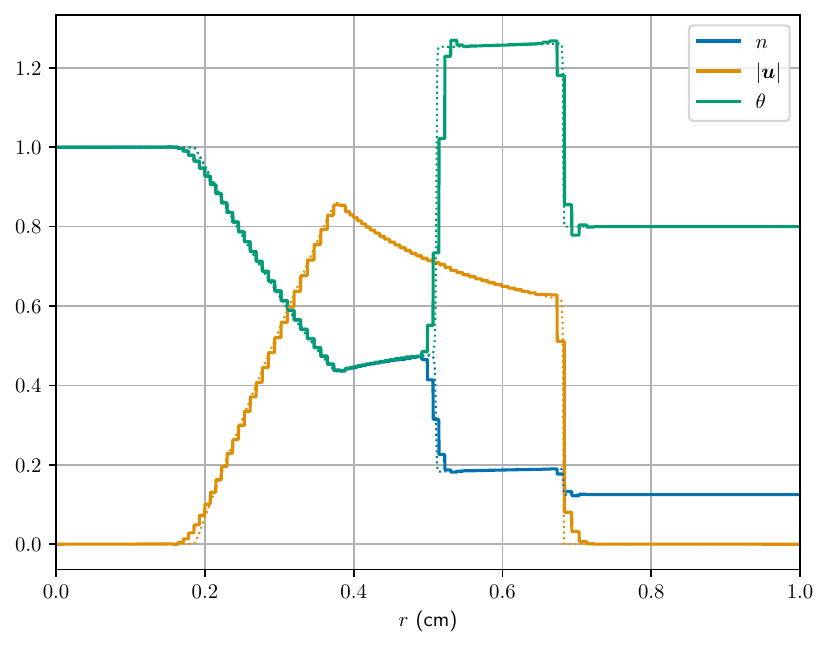}
        \caption{DIRK-1, $\nu = \num{10000}$}
    \end{subfigure}
    \begin{subfigure}{0.45\textwidth}
        \centering
        \includegraphics[width=\textwidth]{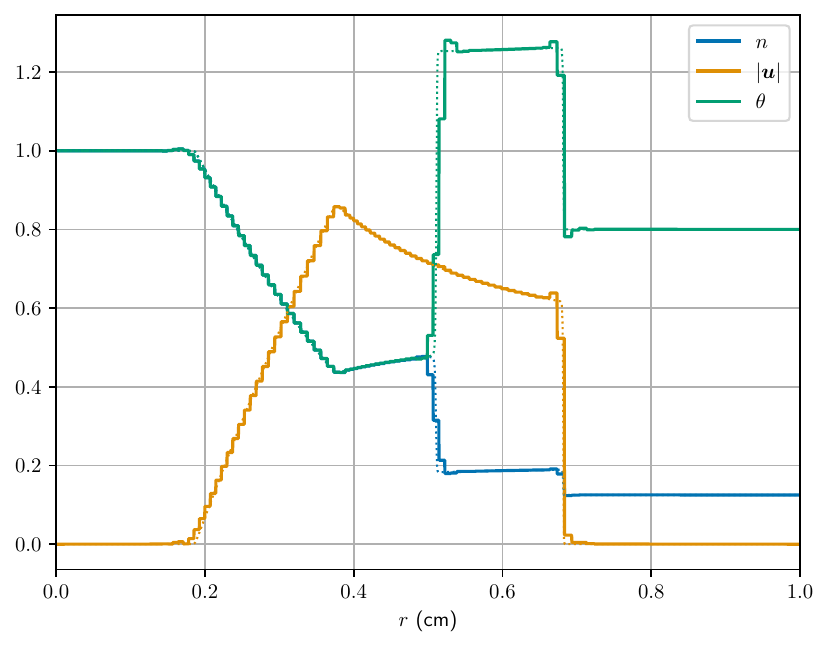}
        \caption{DIRK-2, $\nu = \num{10000}$}
    \end{subfigure}\\
    \caption{2D Sod Problem results with $\nu = \{1000,\num{10000}\}$.
        Reference solutions for the Euler equations are shown using dotted lines.}
    \label{fig:2d-quad-results}
\end{figure}
\begin{figure}
    \centering
    \includegraphics[width=0.45\textwidth]{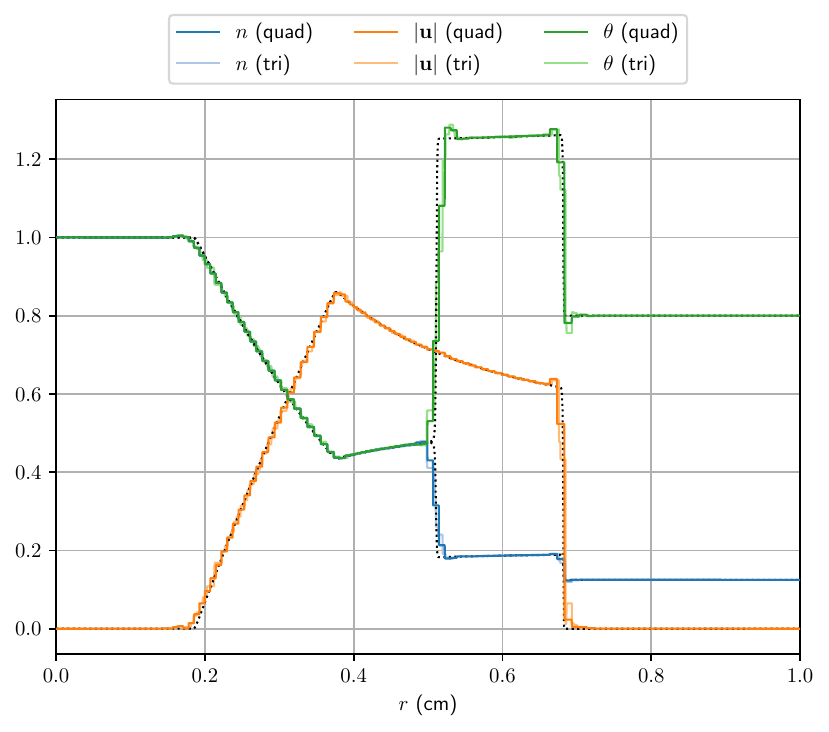}
    \caption{Comparison of 2D Sod problem solutions on quadrilateral and
        triangular grids with $\nu = \num{10000}$.}
    \label{fig:2d-quad-tri-comparison}
\end{figure}

As observed in the 1D Sod problem, solutions in the rarefied gas regime require
significantly more velocity quadrature points to accurately integrate the
kinetic distribution. \Cref{fig:2d-sod-velocity-resolution} shows the solution
to the 2D Sod shock problem using DIRK-2 with $\nu=1$. The dispersive, smooth
profile illustrated in \cref{fig:2d-sod-lambda1} is not attained until
$n_\bmv\ge 32$ as shown in the $L^2$ error norms in
\cref{fig:2d-sod-lambda1-l2norm}. Additionally, ray-effects, which are a
well-known pathology for velocity-collocation methods of the multi-dimensional
transport operator in \eqref{eq:bgk}
\cite{morel_analysis_2003,camminady_ray_2019}, are observed in the solution as
shown in \cref{fig:2d-sod-lambda1-nv4,fig:2d-sod-lambda1-nv64}. For $n_\bmv=8$,
the solution is dominated by rays along directions in the velocity quadrature.
As the velocity quadrature is refined to $n_\bmv=128$, the solution is smooth
and the ray effects are not observable.
\begin{figure}
    \centering
    \begin{subfigure}{0.45\textwidth}
        \centering
        \includegraphics[width=\textwidth]{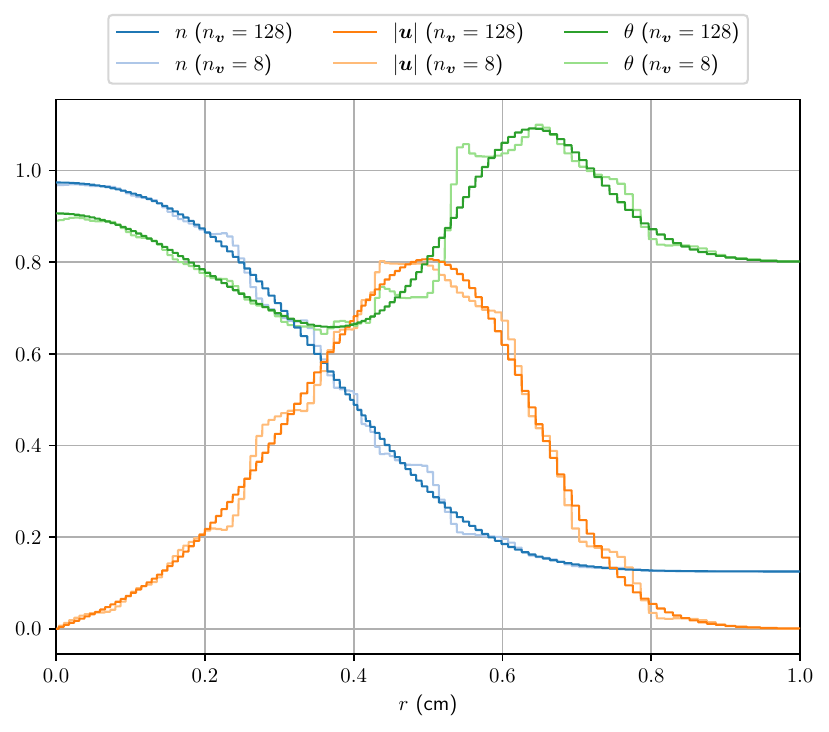}
        \caption{$n_\bmv\in\{8,128\}$}
        \label{fig:2d-sod-lambda1}
    \end{subfigure}
    \begin{subfigure}{0.45\textwidth}
        \centering
        \includegraphics[width=\textwidth]{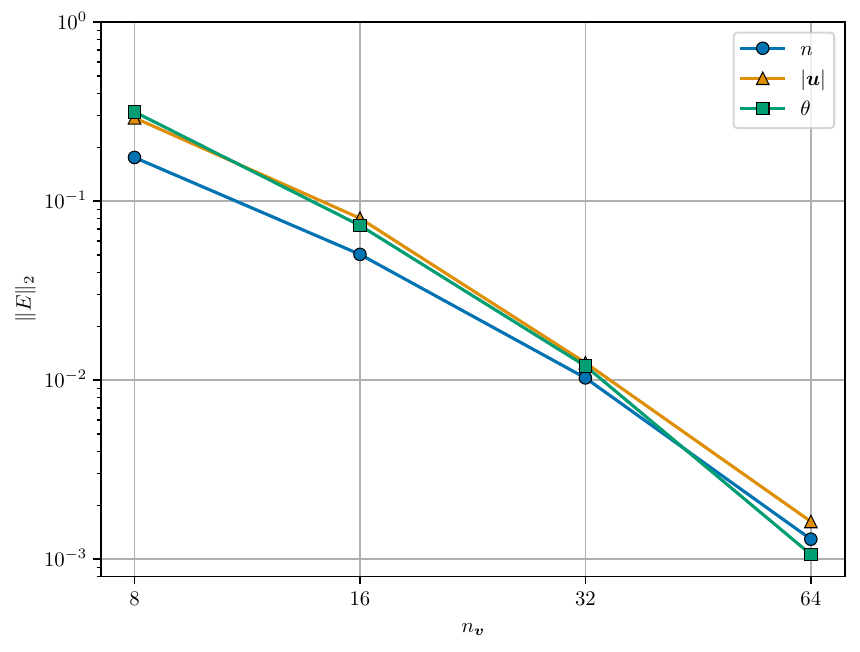}
        \caption{$L^2$ error norms vs $n_\bmv$}
        \label{fig:2d-sod-lambda1-l2norm}
    \end{subfigure}
    \begin{subfigure}{0.45\textwidth}
        \centering
        \includegraphics[width=\textwidth,clip]{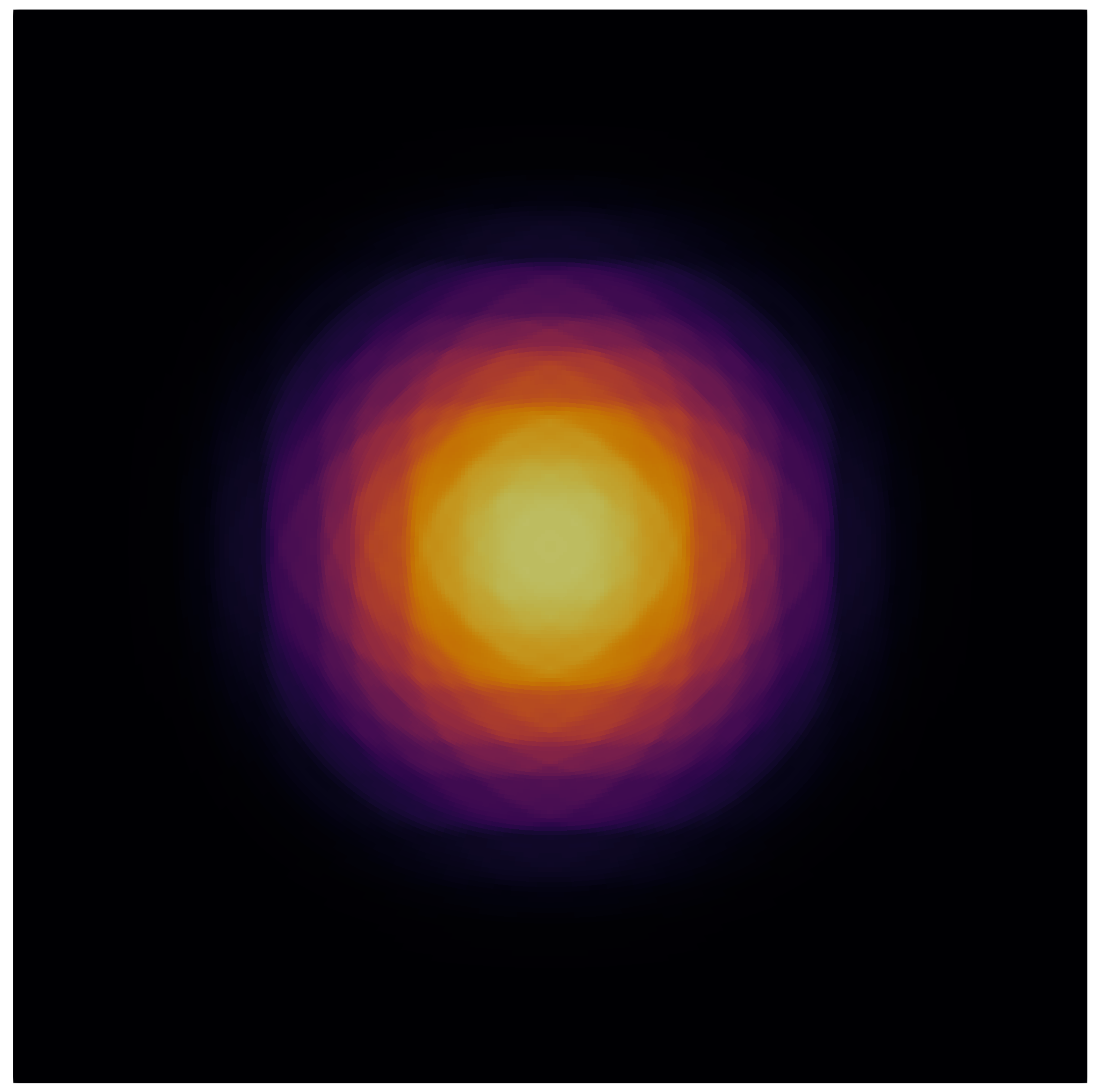}
        \caption{$n$ with $n_\bmv = 8$}
        \label{fig:2d-sod-lambda1-nv4}
    \end{subfigure}
    \begin{subfigure}{0.45\textwidth}
        \centering
        \includegraphics[width=\textwidth,clip]{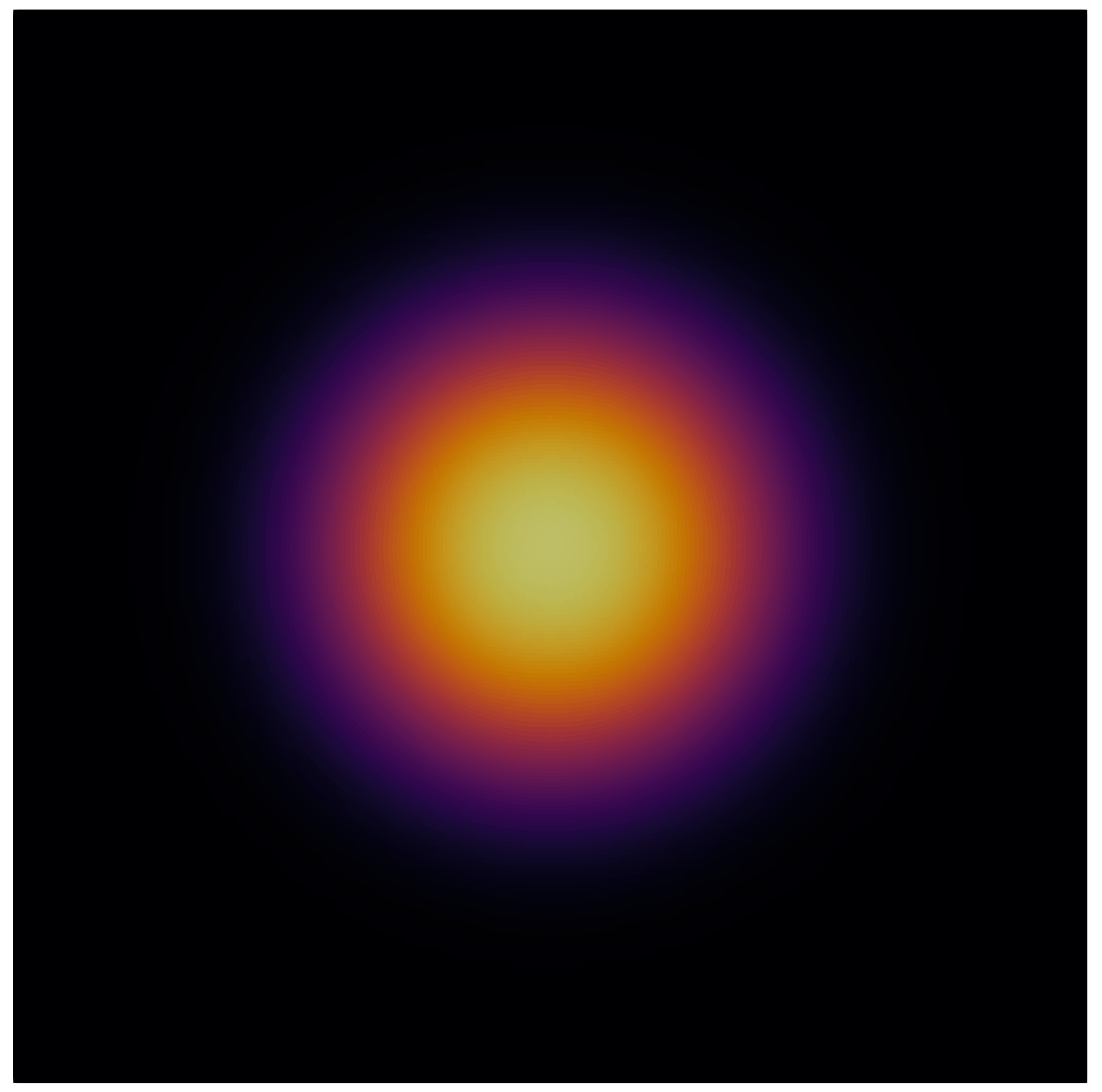}
        \caption{$n$ with $n_\bmv = 128$}
        \label{fig:2d-sod-lambda1-nv64}
    \end{subfigure}
    \caption{2D SOD explosion problem using DIRK-2 with $\nu = 1$.}
    \label{fig:2d-sod-velocity-resolution}
\end{figure}

The 3D Sod problem is solved on a tetrahedral grid containing \num{2806301}
elements (\cref{fig:3d-tet-grid}) and $n_\bmv=16$ (\num{110592} points) velocity
quadrature. The large size of the state vector, \SI{9}{\tebi\byte}, required the
use of \num{512} nodes of Frontier using a total of \num{4096} MI250X \acp{gcd}.
The solution at $t=0.15$ using DIRK-2 with $\nu=1000$ and $\dt=1.3\times
10^{-3}$ ($6\times$ the explicit timestep) is shown in \cref{fig:3d-tet-sod}.
\begin{figure}
    \centering
    \begin{subfigure}{0.55\textwidth}
        \centering
        \includegraphics[width=\textwidth]{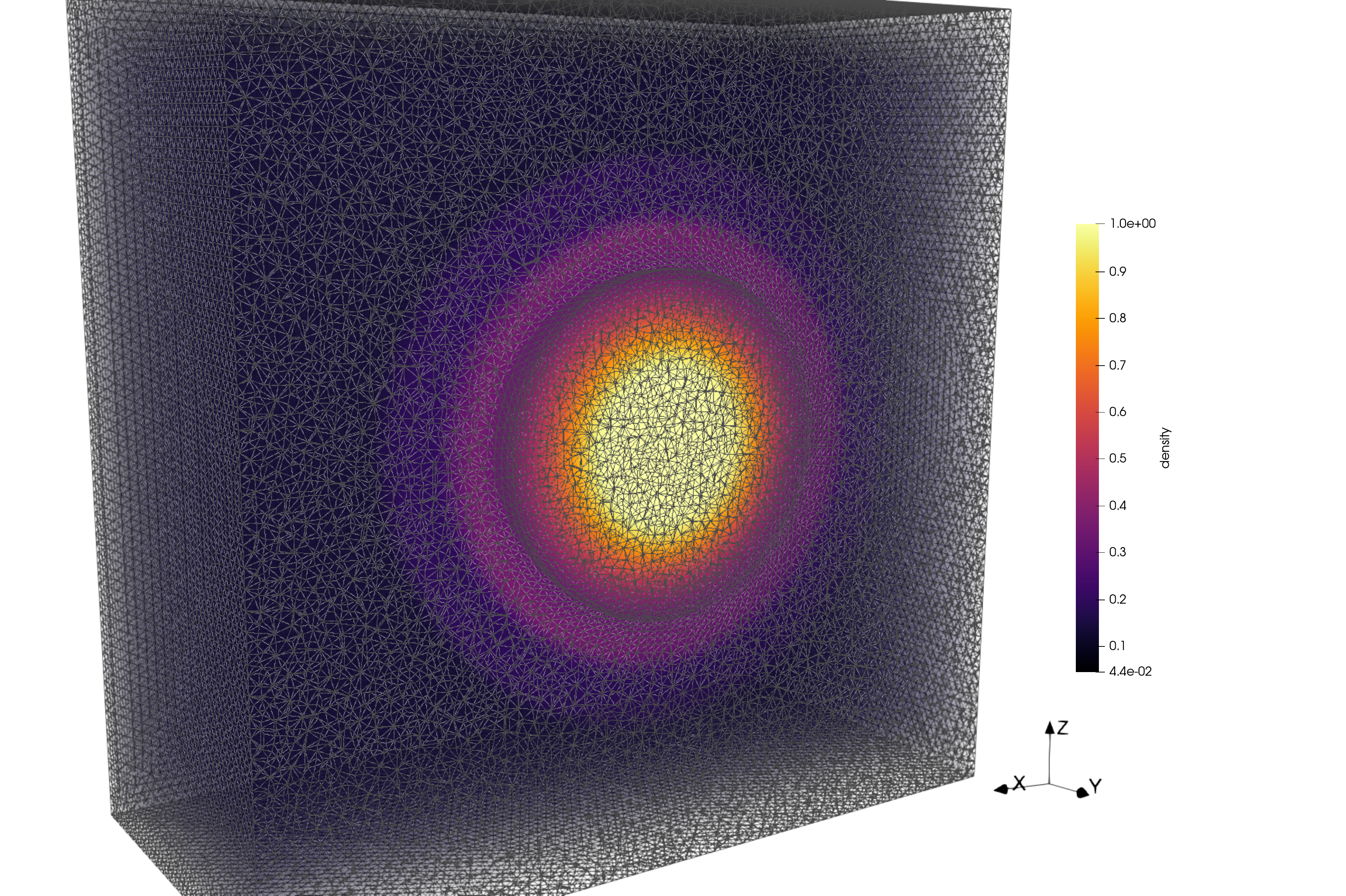}
        \caption{}
        \label{fig:3d-tet-grid}
    \end{subfigure}
    \begin{subfigure}{0.44\textwidth}
        \centering
        \includegraphics[width=\textwidth]{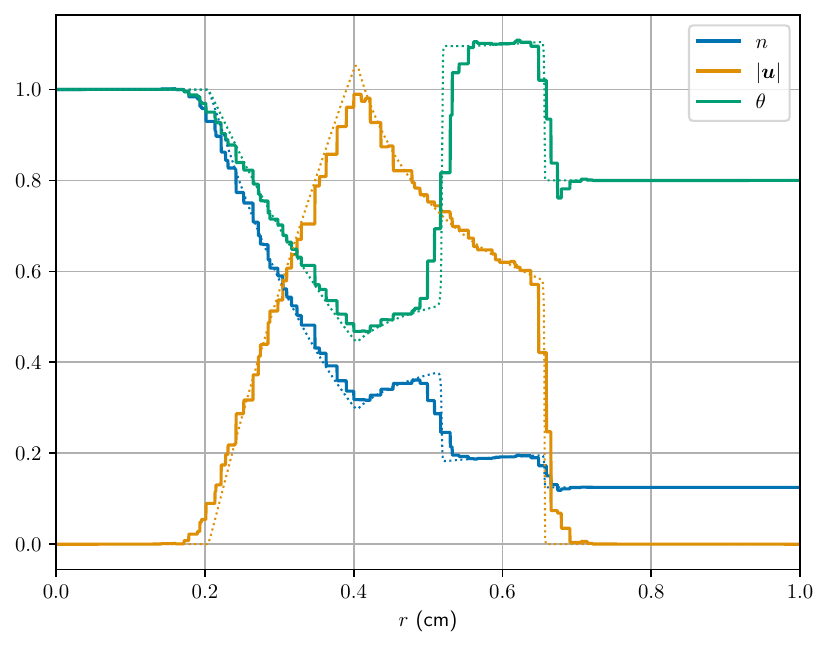}
        \caption{}
        \label{fig:3d-tet-sod}
    \end{subfigure}
    \caption{3D Sod problem on tetrahedral grid.}
\end{figure}

We next use the 2D sod problem to verify demonstrate numerically the linear
stability properties established in \cref{prop:linear-stability}. We repeat the
test on the quadrilateral grid in \cref{fig:2d-grids-quad} using the usual
Maxwellian, the linear Maxwellian from \eqref{eq:linear-Maxwellian}, and a
purely streaming case where $\nu = 0$ and the form of the Maxwellian is
irrelevant. For time integration, we consider several \ac{dirk} methods,
including the algebraically stable (or B-stable) DIRK method provided in
\Cref{sec:time-integrators}. Such methods are known to preserve monotonic
properties of underlying ODEs; see \cite[Theorem 3.8]{rk_dissipation} for
additional details.

Due to the boundary condition, the terms in \eqref{eq:L2-boundary} balance
exactly. Thus we expect to see the norm of the solution decrease monotonically
with time for B-stable time integrators. Of the L-stable integrators in
\cref{eq:L-RK}, only backward Euler is B-stable. For comparison, we also plot
the \ac{bgk} solution with $\nu = 1000$. For each model, we run with each of the
time integrators in \eqref{eq:L-RK} and \eqref{eq:B-RK}.

We plot the difference of the $L^2(X\times V)$ norm of the solution between any
two adjacent timesteps as a function of time. See \cref{fig:2d-l2-stability}. In
each plot, the higher order methods are nearly coincident.  As expected, for the
simplified Maxwellian and the pure streaming problem, all time integrators
dissipate the solution. Additionally, as may be expected, backward Euler is
significantly more dissipative than the higher order time integrators. However,
on the nonlinear \ac{bgk} problem, none of the time integrators give a solution
that dissipates monotonically.
\begin{figure}
    \centering
    \begin{subfigure}{0.32\textwidth}
        \centering
        \includegraphics[width=\textwidth]{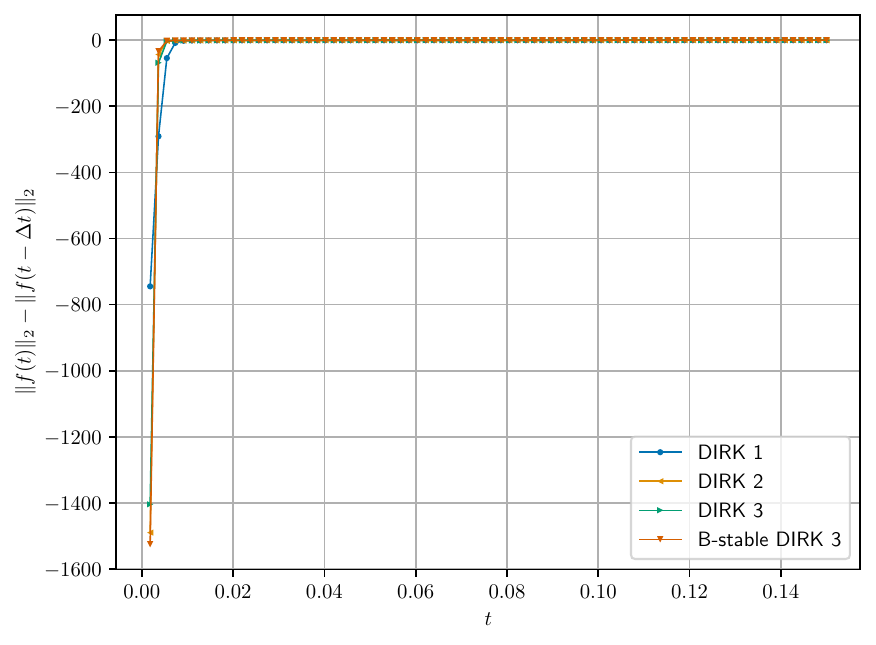}
        \caption{Simplified \acs{bgk} ($\nu = 1000$)}
    \end{subfigure}
    \begin{subfigure}{0.32\textwidth}
        \centering
        \includegraphics[width=\textwidth]{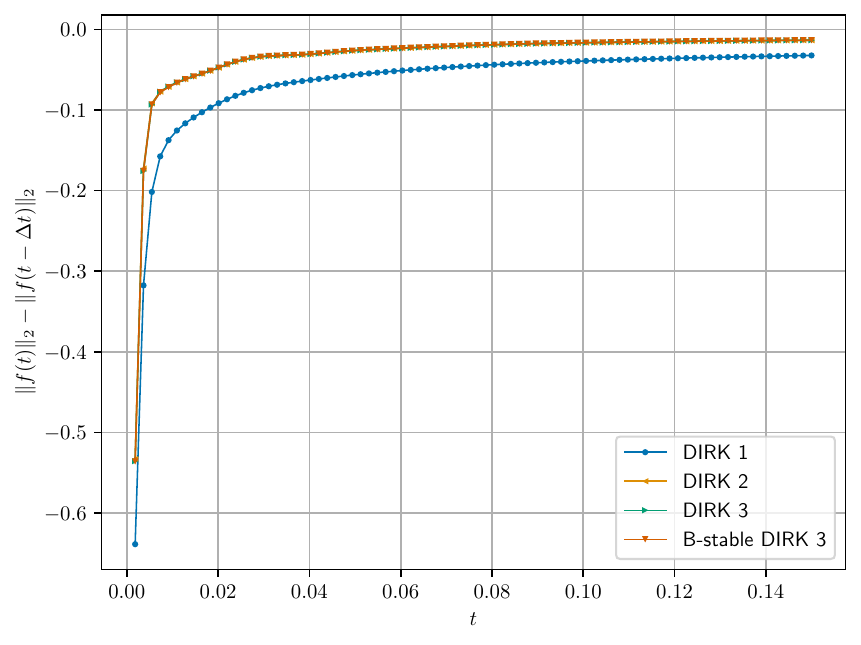}
        \caption{Streaming \acs{bgk} ($\nu = 0$)}
    \end{subfigure}
    \begin{subfigure}{0.32\textwidth}
        \centering
        \includegraphics[width=\textwidth]{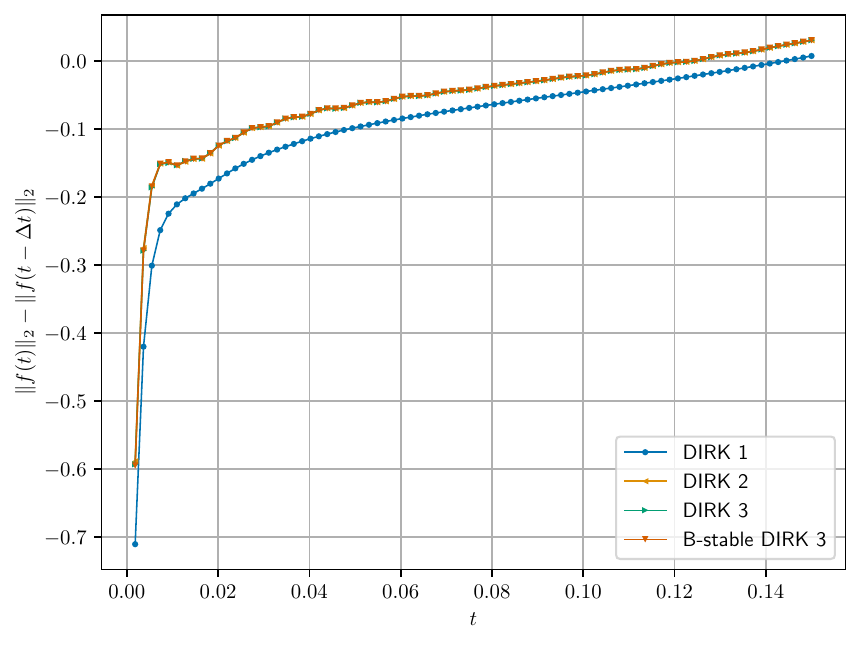}
        \caption{Non-linear \acs{bgk} ($\nu = 1000$)}
    \end{subfigure}
    \caption{Change in the $L^2(X\times V)$ norm of the sweep solution as a function of time.
        The DIRK-2/3 solutions are nearly superimposed.}
    \label{fig:2d-l2-stability}
\end{figure}

%%---------------------------------------------------------------------------%%
\subsection{Flow through rectangular duct}
\label{sec:duct}

This test problem has been designed to showcase the benefit of an implicit
solver and the ability to scale to large problem sizes. It uses a spatial mesh
with length scales spanning about two orders of magnitude, resolving fine
spatial features near boundaries while using a timestep appropriate for tracking
waves in the interior, where the mesh is significantly larger. In practice, a
rectangular domain could be discretized in space with a structured grid.
However, we use unstructured grids to demonstrate that the solver does not
require any special structure in the spatial mesh.

We consider a three-dimensional duct with dimensions $L_1 = 0.3$, $L_2=0.1$, and
$L_3=0.1$:
\begin{equation}
    X = \{\bm{x} = (x_1, x_2, x_3) :
    0 \le x_1 \le L_1, 0 \le x_2 \le L_2, 0 \le x_3 \le L_3\}.
\end{equation}
The initial condition in the duct is a Maxwellian:
\begin{equation}
    \fic(\bm{x},\bm{v})
    = M_{\nic, \uic, \tic} (\bm{v}),
\end{equation}
with
\begin{equation}
    \nic = 1,
    \qquad
    \uic = (0,0,0)^\top,
    \qquad\text{and}\qquad
    \tic = 1.
\end{equation}
To prescribe boundary data, let $\partial D =  \partial D^{\rm{left}} \cup
    \partial D^{\rm{right}} \cup \partial D^{\rm{bottom}} \cup \partial
    D^{\rm{top}}$, where
\begin{subequations}
    \begin{alignat}{2}
        \partial D^{\rm{left}} &= \{ (\bm{x},\bm{v}) \in \partial D : x_1 =0\},
        &\qquad
        \partial D ^{\rm{right}} &= \{ (\bm{x},\bm{v}) \in \partial D : x_1 = L_1\},
        \\
        \partial D ^{\rm{bottom}} &= \{ (\bm{x},\bm{v}) \in \partial D :x_2 = 0\},
        &\qquad
        \partial D ^{\rm{top}} &= \{ (\bm{x},\bm{v}) \in \partial D : x_2 = L_2\},
        \\
        \partial D^{\rm{back}} &= \{ (\bm{x},\bm{v}) \in \partial D : x_3 = 0\},
        &\qquad\text{and}\qquad
        \partial D ^{\rm{front}} &= \{ (\bm{x},\bm{v}) \in \partial D : x_3 = L_3\}.
    \end{alignat}
\end{subequations}
The boundary conditions are Maxwellian inflow on the left and zero inflow on the
right
\begin{equation}
    \fbc(\bm{x},\bm{v},t) = \begin{cases}
        M_{\nl, \ul, \tl} (\bm{v}) ,
         & (\bm{x},\bm{v}) \in  \partial D^{\rm{left}},
        \quad
        \\
        0,
         &
        (\bm{x},\bm{v}) \in  \partial D^{\rm{right}},
    \end{cases}
\end{equation}
where
\begin{equation}
    \nl = 1,
    \qquad
    \ul = (2,0,0)^\top,
    \qquad\text{and}\qquad
    \tl = 1.
\end{equation}
We impose bounce-back boundary conditions on the top and bottom walls of the
duct and reflection boundary conditions on the front and back walls.

We simulate the flow up to $t=0.1$, using backward Euler (DIRK-1) time
discretization with timestep of $0.001$. We consider two problem regimes: a
transition regime with $\nu = 1$ and a collision regime with $\nu = 1000$. The
velocity domain is $[-12,12]^3$, with $24$ and $16$ velocity cells per
dimension for $\nu = 1$ and $\nu = 1000$, respectively. Each each velocity cell
supports a local $\bbQ^2$ quadratic basis, resulting in \num{373248} and
\num{110592} velocity degrees of freedom, respectively.

We consider three meshes; see \Cref{tab:mesh_data}. Mesh 1 (shown in
\Cref{fig:3d-channel-mesh}) is a relatively coarse, nonuniform mesh that used in
both moderately ($\nu = 1$) and highly collisional ($\nu=1000$) settings. The
smallest cells in the mesh are clustered near the bounce-back boundaries, which
are expected to generate more refined flow structures.  The mesh contains
\num{684207} tetrahedral cells, decomposed into \num{4096} subdomains that own
approximately \num{160} cells each. For the moderately collisional problem, we
also use Mesh 2, a finer mesh with \num{1861080} elements across \num{12288}
domains, keeping the ratio of cells to subdomains roughly constant. For the high
collisional problem, we also use Mesh 3, a mesh that is refined near the inlet.
The result is a mesh with \num{1632897} cells, also split across \num{12288}
domains. The relatively small number of physical cells per subdomain is
primarily due to storage requirements. Using a linear basis in space for the
$\nu = 1$ velocity grid, each tetrahedral cell requires the storage of
\num{1492992} double precision floats requiring about $\SI{11.4}{MiB}$.

For the \ac{dg} solver, which does not assemble a global matrix, two cells need
to be connected only if they share a face. Unfortunately, domain decomposition
algorithms designed for continuous finite element methods, as provided through
Trilinos' mesh-data structures and partitioning algorithms, consider two cells
connected if they share any nodes. The issue is especially pronounced in 3D, in
which case it is likely that many tetrahedra share any given node. Aside from
storage concerns, small physical domains are likely to increase the number of
nonlinear iterations required to converge subdomain boundaries. Addressing the
storage limitation is an important topic for future study; possible solutions
include a tailored spatial mesh decomposition and decomposition in velocity
space. Despite these concerns, we are able to successfully run a problem in
excess of \num{2.77} trillion phase space degrees of freedom, solving for
\num{2.78e14} total unknowns over all timesteps. This particular simulation was
completed in 4 hours and 42 minutes on \num{1536} nodes on Frontier using
\num{12288} MI250X \acp{gcd}.

The timestep is selected to roughly correspond to the minimum time it
takes a particle to traverse the largest spatial cells for all three meshes:
\begin{equation}
    \dt = 0.001 \approx \frac{h_{\bmx}^{\max}}{\max_{\mbj} |\vj|},
\end{equation}
where $h_{\bmx}^{\max}$ is the largest edge length in the mesh.
Using \eqref{eq:explicit-dt}, this results in a timestep that is approximately $291\times$,
$568\times$, and $521\times$ the explicit timestep on Mesh 1,
Mesh 2, and Mesh 3, respectively.
\begin{table}
    \centering
    \caption{Mesh data for 3D duct flow.}
    \begin{tabular}{cllllll}
        \toprule
        number & description        & $h_{\bmx}^{\min} $ & $h_{\bmx}^{\max} $ & $N_\bmx$       & $\Delta t_e$     & $\Delta t / \Delta t_e$
        \\ \midrule
        1      & coarse mesh        & $\num{6.43e-4}$    & $\SI{2.19e-2}{}$   & $\SI{6.8e5}{}$ & $\SI{3.43e-6}{}$ & 291
        \\
        2      & fine, $\nu = 1$    & $\SI{3.30e-4}{}$   & $\SI{1.71e-2}{}$   & $\SI{1.8e6}{}$ & $\SI{1.76e-6}{}$ & 568
        \\
        3      & fine, $\nu = 1000$ & $\SI{3.60e-4}{}$   & $\SI{2.03e-2}{}$   & $\SI{1.6e6}{}$ & $\SI{1.92e-6}{}$ & 521
        \\ \bottomrule
    \end{tabular}
    \label{tab:mesh_data}
\end{table}
\begin{figure}
    \centering
    \includegraphics[width=0.6\textwidth]{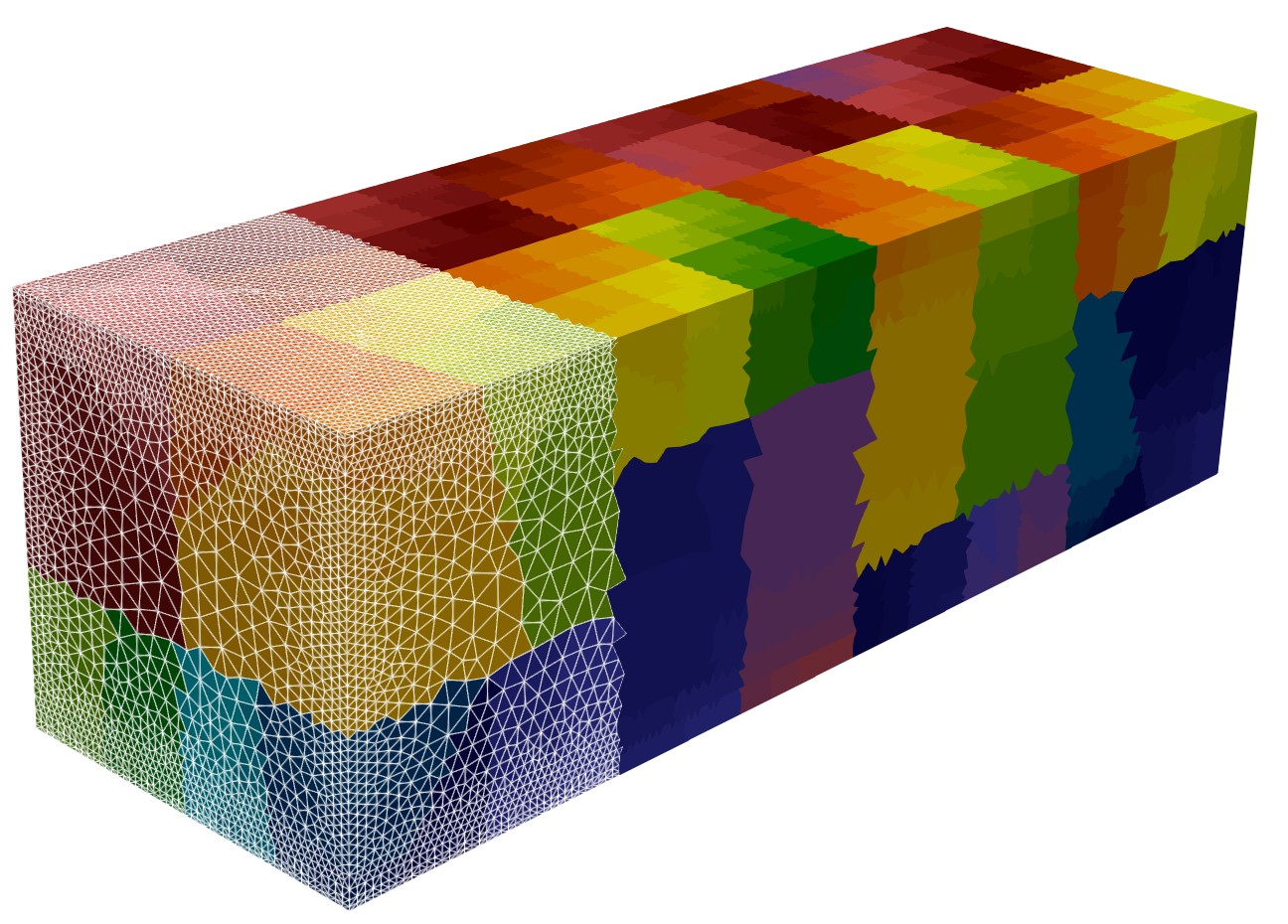}
    \caption{3D boundary layer mesh and decomposition showing \num{12288}
        subdomains by index.}
    \label{fig:3d-channel-mesh}
\end{figure}

In \Cref{fig:lamdba-1000-channel-u} a snapshot of the bulk velocity magnitude
plotted at the spanwise midplane with $\nu=1000$ on Mesh $1$ and Mesh $3$ at
$t=0.07$ is shown.
\begin{figure}
    \centering
    \begin{subfigure}{0.45\textwidth}
        \centering
        \includegraphics[width=\textwidth]{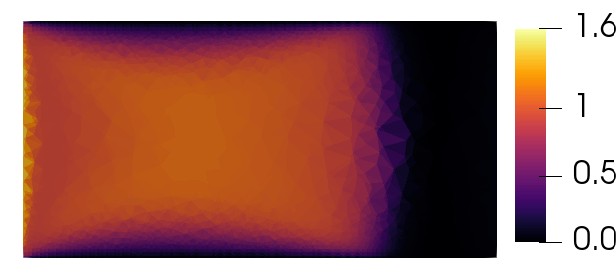}
        \caption{Mesh 1}
    \end{subfigure}
    \begin{subfigure}{0.45\textwidth}
        \centering
        \includegraphics[width=\textwidth]{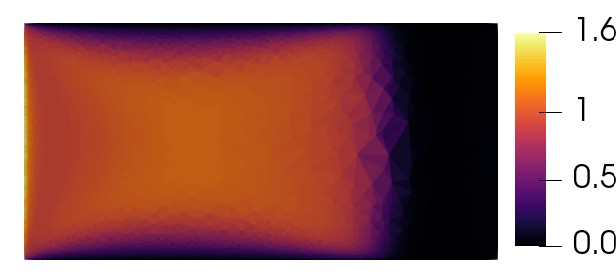}
        \caption{Mesh 3}
    \end{subfigure}
    \caption{Bulk velocity magnitude at $z=0.05$ and $t=0.07$ with $\nu = 1000$.}
    \label{fig:lamdba-1000-channel-u}
\end{figure}
Near $\partial D^{\rm left}$, $\partial D^{\rm top}$, and $\partial D^{\rm
bottom}$, the solver resolves a strong boundary layer. The resolution at
$\partial D^{\rm left}$ improves markedly upon near-inlet
refinement that occurs between Mesh $1$ and Mesh $3$.

\Cref{fig:lamdba-1-channel-u} shows a snapshot of the bulk velocity magnitude at
the spanwise midplane with $\nu=1$ and $t=0.1$ on Mesh $1$ and Mesh $2$.
\begin{figure}
    \centering
    \begin{subfigure}{0.45\textwidth}
        \centering
        \includegraphics[width=\textwidth]{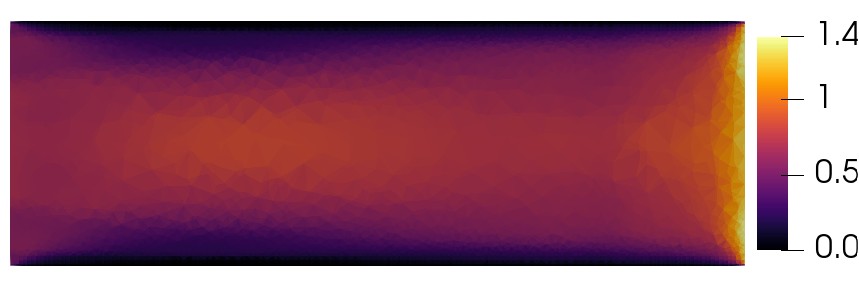}
        \caption{Mesh 1}
    \end{subfigure}
    \raisebox{0.05cm}{
        \begin{subfigure}{0.45\textwidth}
            \centering
            \includegraphics[width=\textwidth]{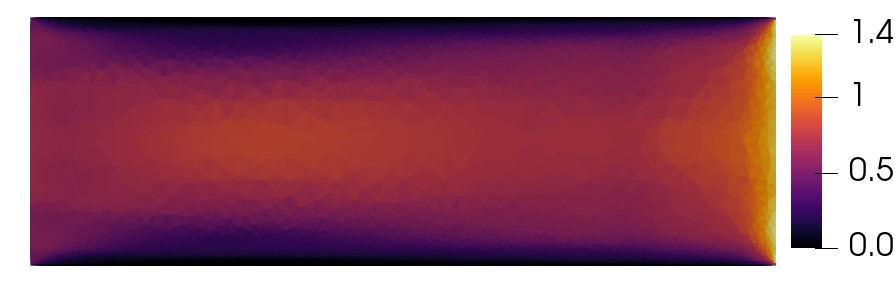}
            \caption{Mesh 2}
        \end{subfigure}}
    \caption{Bulk velocity magnitude at $z=0.05$ and $t=0.07$ with $\nu = 1$.}
    \label{fig:lamdba-1-channel-u}
\end{figure}
As in the $\nu=1000$ case, \Cref{fig:lamdba-1000-channel-u}, this solution also
has a strong boundary layer near $\partial D^{\rm top}$ and $\partial D^{\rm
bottom}$. Overall, in the moderately collisional limit, the shape of the
solution is well captured by both the coarse (Mesh 1) and fine (Mesh 2)
resolutions.

While in the moderately collisional case there does not appear to be an
under-resolved layer near $\partial D^{\rm left}$, we do see some faint structures
that could be ray effects; see \cite{camminady_ray_2019}. There is also a
wedge-shaped artifact near $\partial D^{\rm right}$. Investigation with high
resolution on an analogous 2D problem indicates that this is primarily caused by
the choice of pure outflow on $\partial D^{\rm right}$.

\Cref{fig:channel-gradP} shows the pressure gradient magnitude, $\left\lvert
\nabla_\bmx (n \theta)\right\rvert$, for the moderately and highly collisional
problems on Mesh $2$ and Mesh $3$, respectively. This quantity is not direct
simulation output, but is approximated in post-processing with ParaView; see
\cite{paraview}. While the pressure gradient when $\nu=1000$ primarily shows a
traveling pressure wave, there is more structure near the top and bottom walls
when $\nu=1$, especially near the left inlet.
\begin{figure}
    \centering
    \begin{subfigure}{0.45\textwidth}
        \centering
        \includegraphics[width=\textwidth]{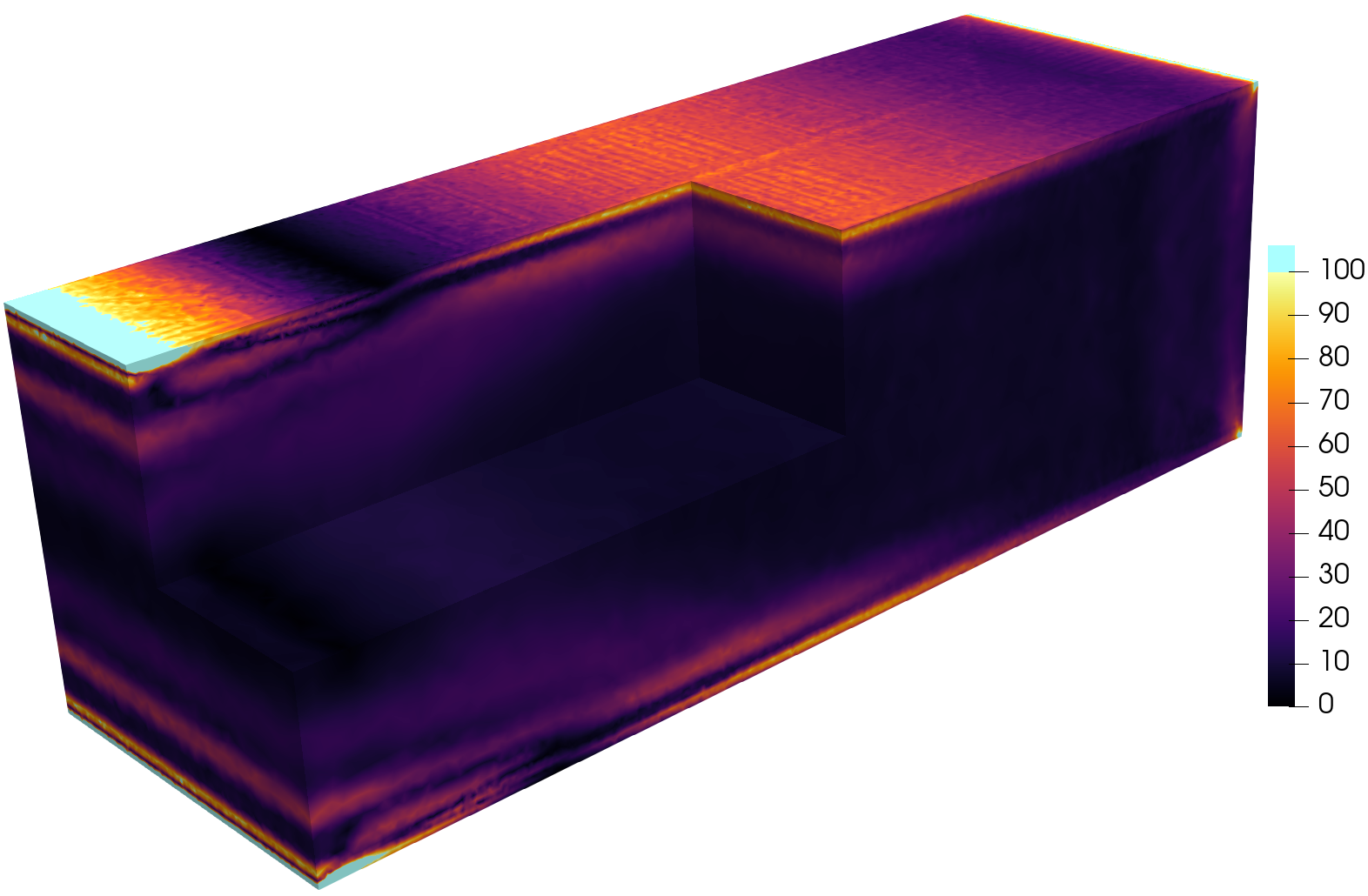}
        \caption{$\nu = 1$}
    \end{subfigure}
    \raisebox{0.05cm}{
        \begin{subfigure}{0.45\textwidth}
            \centering
            \includegraphics[width=\textwidth]{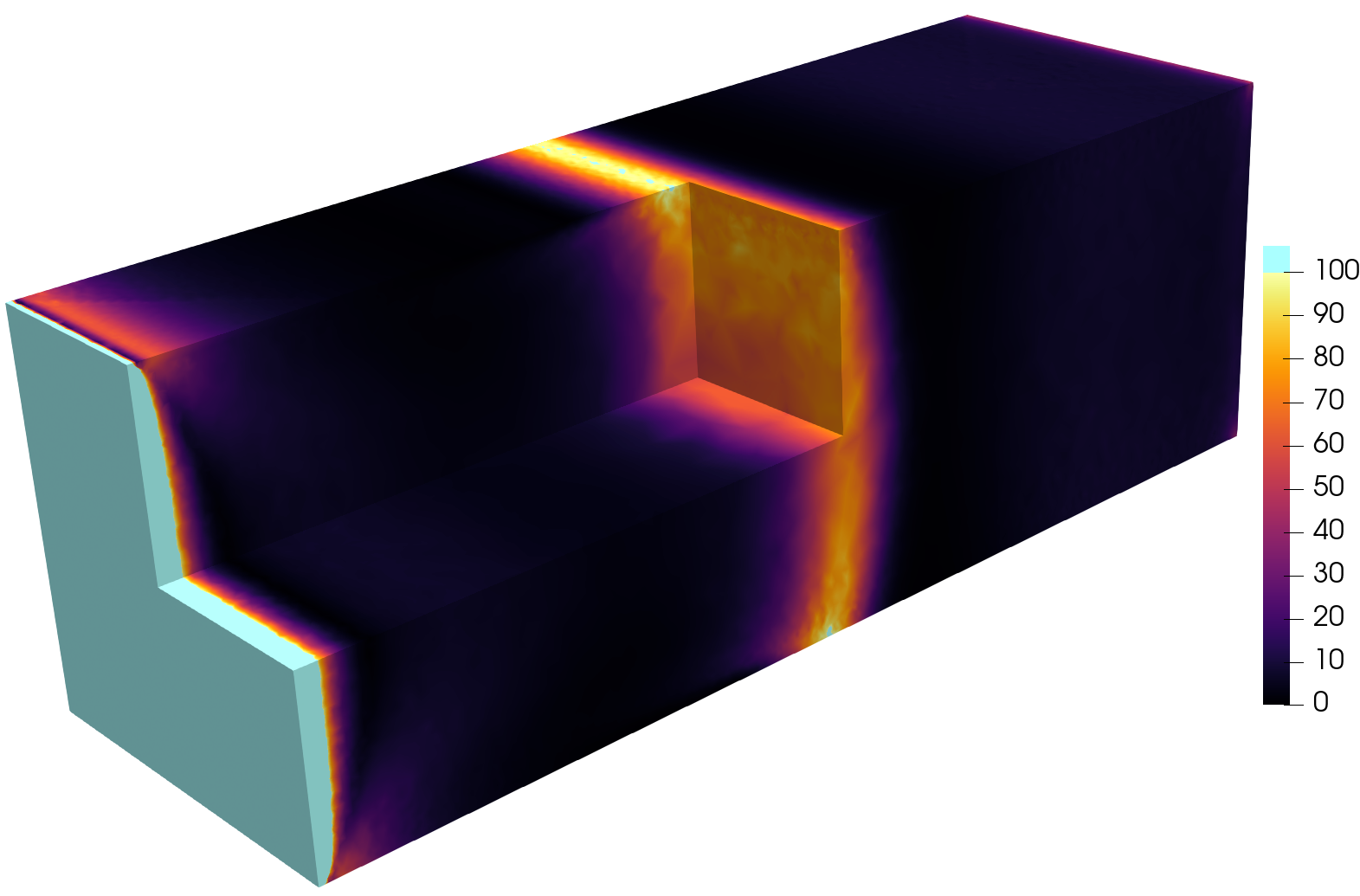}
            \caption{$\nu = 1000$}
        \end{subfigure}} \caption{Gradient of pressure magnitude at $t=0.07$ for
    Mesh 2 (a) and Mesh 3 (b), clipped at \num{100} and with one octant removed
    to show detail. Clipped values are shown in light blue. The true maximum for
    both plots is near \num{550}.}
    \label{fig:channel-gradP}
\end{figure}

%%- - - - - - - - - - - - - - - - - - - - - - - - - - - - - - - - - - - - - -%%
\subsection{Solver performance}
\label{sec:solver-performance}

The single-domain, device performance of the solver is tested on the 2D Sod
problem with a quadrilateral grid containing \num{6144} elements and is
illustrated in \Cref{fig:2d-single-domain-performance}. The $\dofx$ for
quadrilaterals with $\bbM^{k_\bmx}=\bbQ^1$ is equivalent to $\bbP^1$ on
tetrahedrons, so the performance characteristics as a function of velocity
quadrature points will be identical. Each non-linear Picard iteration consists
of a local sweep on domain $\cT_{\bmx,i}$, communication of domain boundary
value intensities, update of problem boundary $f^{\mathrm{bc}}$, integration of
the moments given by \eqref{eq:moments}, and update of the non-linear Maxwellian
source. On a single domain, convergence is dictated by the non-linear source
and problem boundary conditions.

\Cref{fig:gpu-solver-time-velocity} shows the solver performance as a function
of velocity quadrature points for GPU and multi-core CPU architectures. On GPUs
the solve time for the \ac{tgs} is universally faster than the \ac{ts}; however
the difference between the two reduces as the number of quadrature points
increases. On the \nvidia H100, this speedup ranges from \num{1.4} at
\num{36864} quadrature points ($n_\bmv=64$) to \num{10.4} at \num{2304}
quadrature points ($n_\bmv=16$). This occurs because the total concurrent work
in the \ac{ts} is limited by the number of quadrature points, whereas the
\ac{tgs} is limited by the product of quadrature points and elements within a
single generation of the graph.  As the number of quadrature points increases,
the work load on the GPU becomes more balanced, and the performance difference
between the two algorithms decreases. We note that as the aspect ratios of the
domains become large, the maximum size of $W_g$ decreases, and the performance
of \ac{ts} has been observed to surpass \ac{tgs}.  For example, on the 1D Sod
problem (\cref{sec:1d-sod}) the \ac{ts} is greater than 3 times faster than the
\ac{tgs} with \num{9216} velocity quadrature points.
% (1.579623/0.517028) - see the verify.ipynb notebook in bgk
%%
\begin{figure}
    \centering
    \begin{subfigure}{0.46\textwidth}
        \centering
        \includegraphics[width=\textwidth]{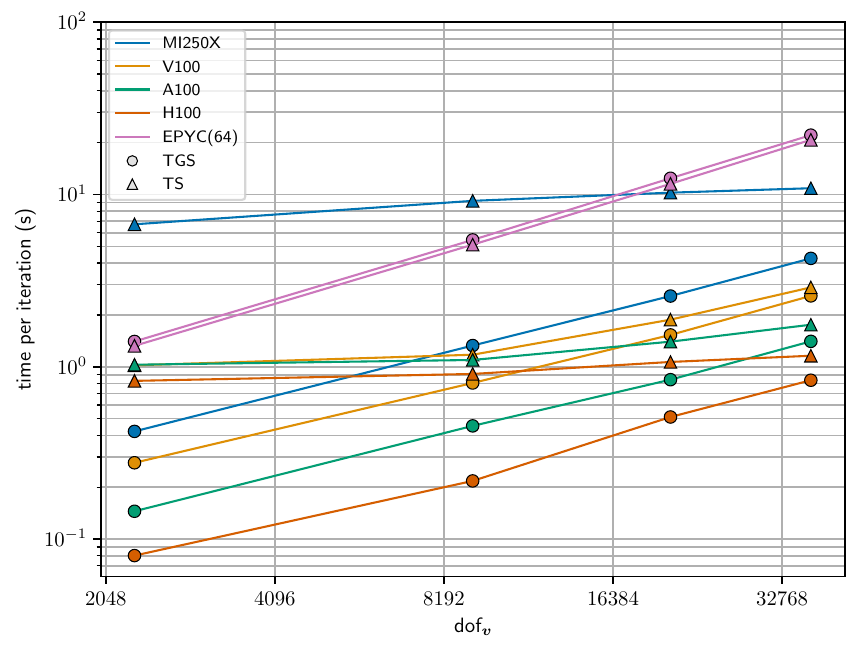}
        \caption{Time per iteration}
        \label{fig:gpu-solver-time-velocity}
    \end{subfigure}
    \begin{subfigure}{0.45\textwidth}
        \centering
        \includegraphics[width=\textwidth]{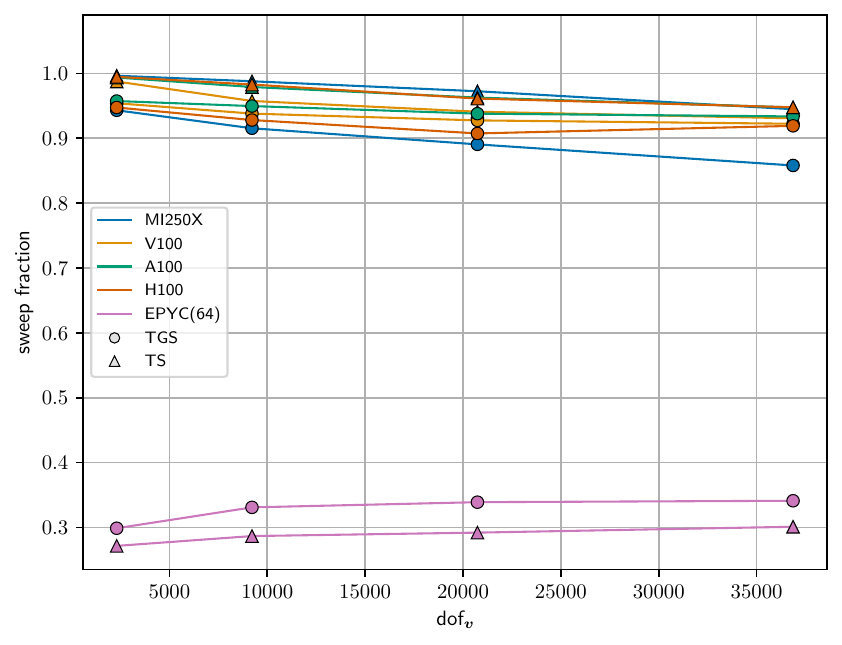}
        \caption{Sweep fraction}
        \label{fig:sweep-fraction-velocity}
    \end{subfigure}
    \caption{Single domain performance as function of number of quadrature
        points on different architectures.}
    \label{fig:2d-single-domain-performance}
\end{figure}

The multi-threaded CPU performance on a 64-core EPYC processor, using all 64
threads with OpenMP, is shown to be significantly slower than the GPU
performance. Additionally, the CPU scales nearly proportionally with the number
of velocity points, such that $t\propto\dofv^{\,0.99}$. In contrast, the GPU is
less sensitive to the number of velocity points. On the H100 the \ac{tgs} is
characterized by $t\propto\dofv^{\,0.85}$ while the \ac{ts} is highly
insensitive, $t\propto\dofv^{\,0.12}$. At low numbers of quadrature points,
$\text{W}^\text{TGS} > \text{W}^\text{TS}$ because $\text{W}^\text{TS}$ is
limited by the velocity quadrature size, whereas $\text{W}^\text{TGS}$ can
exploit the product of the mesh and velocity within
$\operatorname{Gen}_\mathbf{j}(g)$. Ultimately, the amount of concurrent work on
the GPU is saturated at high numbers of quadrature points in the \ac{ts}, and
its performance will equal the \ac{tgs}.  This occurs because the moment
integrations on the GPU are highly efficient due to nearly perfectly coalesced
memory access patterns, and the time spent performing the moment integrations is
nearly independent of the number of velocity points.
\Cref{fig:sweep-fraction-velocity} shows this effect where the percent of time
spent in the sweep on the EPYC ranges from \SIrange{27}{34}{\percent}. The time
spent doing sweeps on the GPU is never less than \SI{85}{\percent} and is
generally greater than \SI{90}{\percent} on \nvidia GPUs. This confirms that the
moment integrations consume a much larger percentage of runtime on multi-core
CPUs, and this fraction of time increases proportionally with the number of
velocity quadrature points. On the GPU both the sweeps and moment integrations
increase more slowly with velocity points than the CPU solver.

\Cref{fig:max-speedup} shows the maximum speedup attained on the H100 relative
to the best performance attained on other architectures. The overall solve time
on the H100 is up to \num{1.7} times faster than an A100 and \num{25} times
faster than the best performance on a 64-core EPYC CPU.  Thus, a single H100 is
roughly equivalent to \numrange{17}{25} 64-core EPYC CPUs.  The \ac{bgk} solver
is shown to be both platform-portable across all current architectures, tunable
to different mesh domains, and highly efficient on modern GPUs.
\begin{figure}
    \centering
    \includegraphics[width=0.45\textwidth]{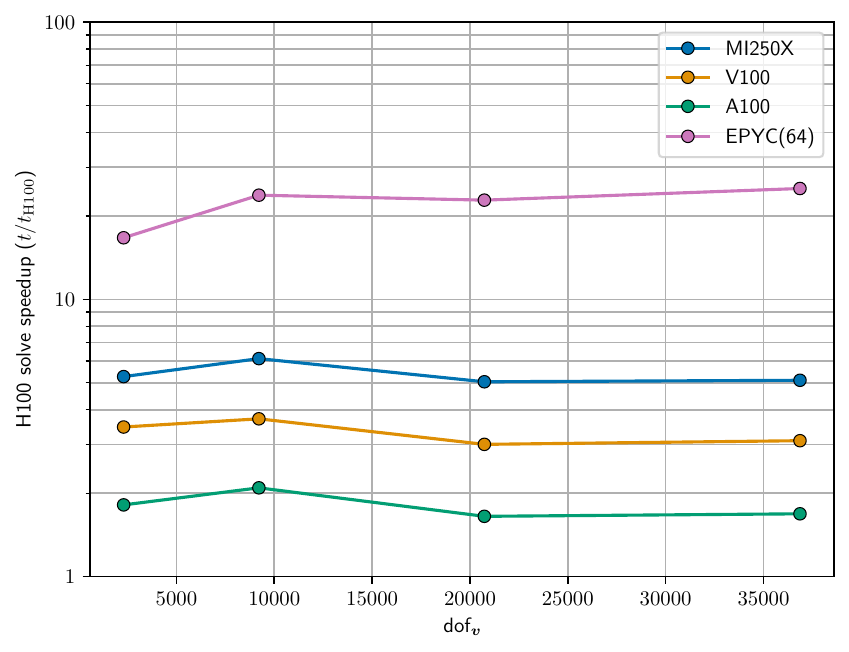}
    \caption{Single domain H100 speedup relative to the best sweep performance
        on other architectures.}
    \label{fig:max-speedup}
\end{figure}

Domain decomposition parallel efficiency is tested on Frontier using the
triangular grid shown in \cref{fig:2d-grids-tri}. The strong scaling
performance of the solver is shown in \cref{fig:strong-scaling} for $\nu\in\{1,
    10, 1000\}$ using DIRK-1 time integration with $n_\bmv=16$ and \ac{tgs}. The
spatial decomposition is generated using ParMetis with default weighting, so no
optimization based on particle mean free path is considered
\cite{gerhard_parallel_2024}. Overall strong scaling is efficient for all fluid
regimes below 16 domains (\num{6180} elements per domain). In the streaming
limit, the convergence of the non-linear Picard iterations is more strongly
influenced by domain boundary intensities, and additional iterations are
required to converge these terms as shown in \cref{fig:dd-iterations}.
Additionally, the fraction of time spent doing sweeps,
shown in \cref{fig:sweep-fraction}, decreases significantly as the number of domains
increases, which is an affect of both increased communication time and less
efficient GPU utilization on smaller domains.

In the fluid limit, the non-linear source dominates convergence because very
little flux crosses domain boundaries. The local sweep time as a function of
domain size is shown in \cref{fig:2d-sweep-cell}. Here we note that the GPU
scales nearly proportionally with mesh size, $t\propto \dofx^{\,0.99}$ on the
GPU. This is in contrast to the scaling with velocity quadrature illustrated in
\cref{fig:2d-single-domain-performance}.
\begin{figure}
    \centering
    \begin{subfigure}{0.45\textwidth}
        \centering
        \includegraphics[width=\textwidth]{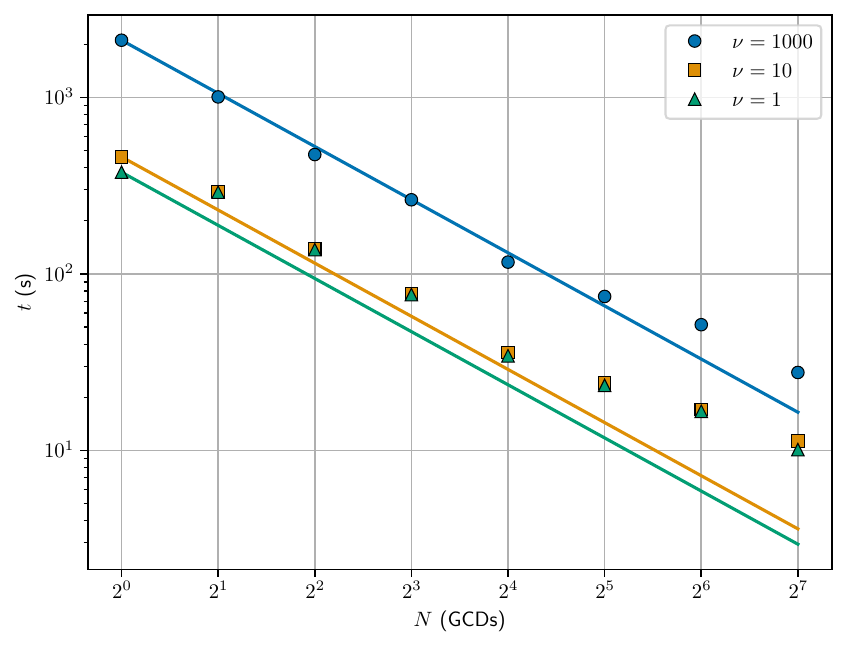}
        \caption{}
    \end{subfigure}
    \begin{subfigure}{0.45\textwidth}
        \centering
        \includegraphics[width=\textwidth]{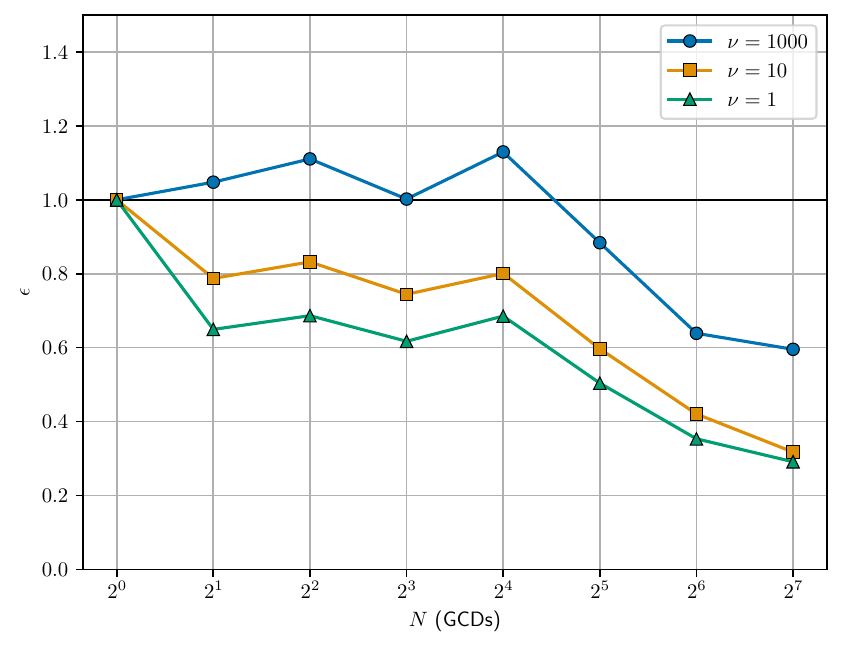}
        \caption{}
        \label{fig:strong-scaling-efficiency}
    \end{subfigure}
    \caption{Strong scaling on Frontier.}
    \label{fig:strong-scaling}
\end{figure}
\begin{figure}
    \centering
    \begin{subfigure}{0.32\textwidth}
        \includegraphics[width=\textwidth]{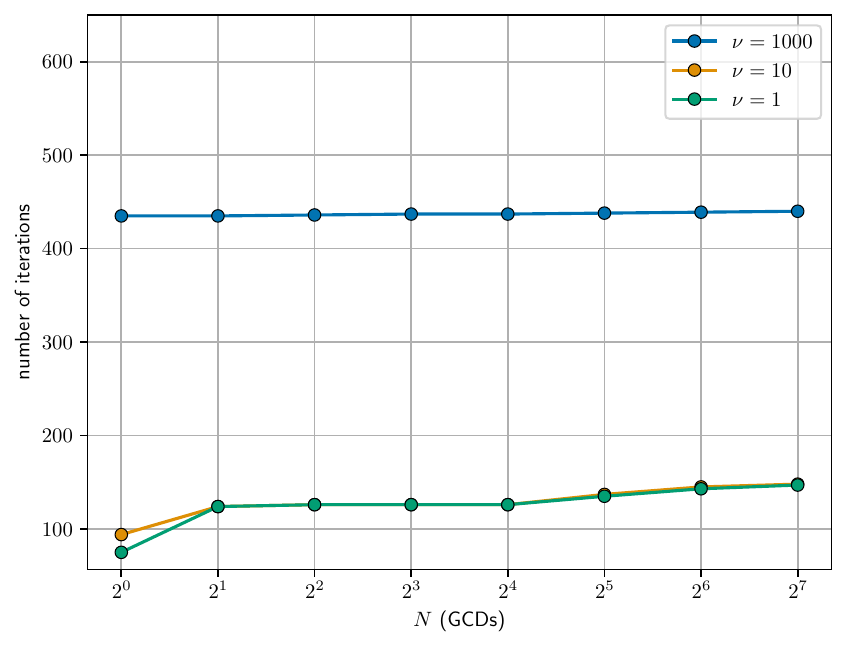}
        \caption{Picard iterations.}
        \label{fig:dd-iterations}
    \end{subfigure}
    \begin{subfigure}{0.32\textwidth}
        \includegraphics[width=\textwidth]{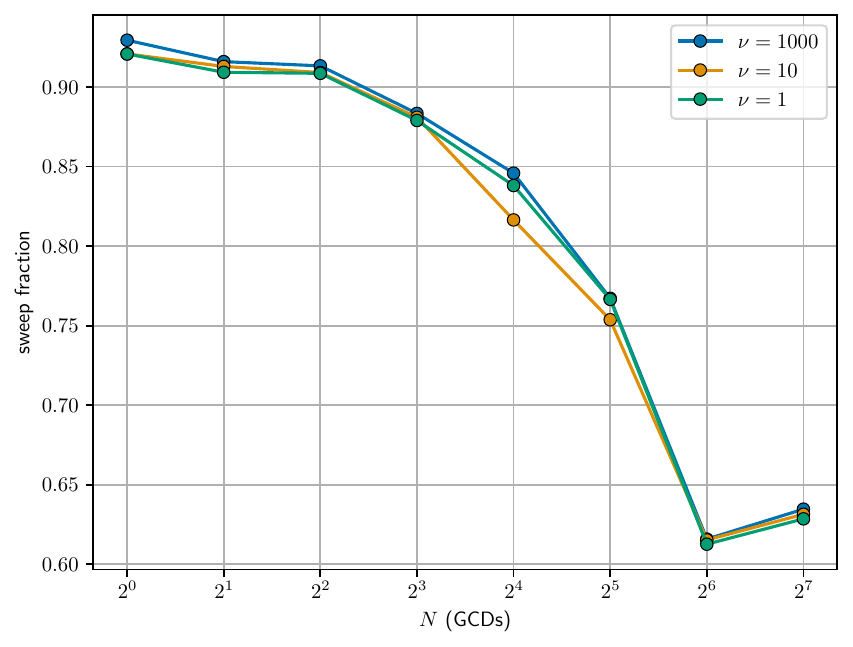}
        \caption{Sweep fraction.}
        \label{fig:sweep-fraction}
    \end{subfigure}
    \begin{subfigure}{0.32\textwidth}
        \includegraphics[width=\textwidth]{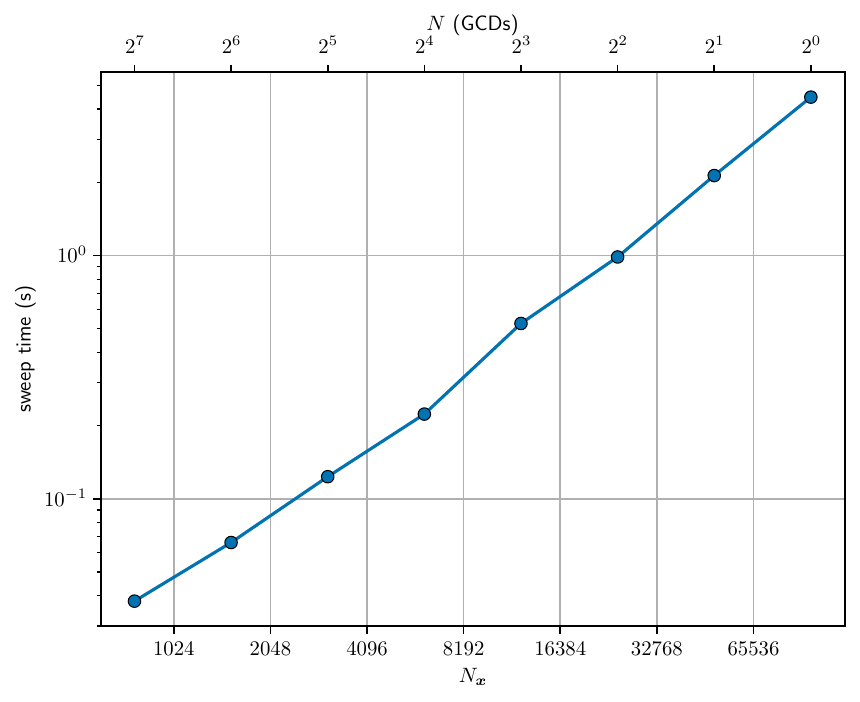}
        \caption{Sweep time.}
        \label{fig:2d-sweep-cell}
    \end{subfigure}
    \caption{Frontier performance as a function of number of domains (domain size).}
    \label{fig:frontier-performance}
\end{figure}

%%---------------------------------------------------------------------------%%

%%---------------------------------------------------------------------------%%
\section{Conclusion}
\label{sec:conclusion}

In this paper, we have presented a nodal \ac{dg} discretization for solving the time
dependent \ac{bgk} equation on unstructured spatial meshes. The method is linearly
stable and, with an appropriate discretization of the Maxwellian, conserves
mass, momentum and energy.

We have integrated the resulting time-dependent ODEs implicitly in order to
avoid restrictive time steps induced by boundary layers and other
geometry-induced features. We have solved the resulting algebraic equations
using a well-known Picard iteration algorithm known as source iteration. Here
the main contribution is a scalable implementation of the workhorse sweeping
algorithm that is commonly used in the radiation transport community. We
consider two types of graph-based sweeping strategies and demonstrate the
implementation is portable and performs well on both CPUs and GPUs.

The method has been tested on several variations of the Sod shock tube problem,
to demonstrate correctness and accuracy of the discretization. The benefits of
the projection-based treatment of the Maxwellian and the linear stability of
the discretization have also been demonstrated. Finally, we have analyzed a
three-dimensional duct flow problem, using various meshes with cells that vary
by orders of magnitude, especially in more collisional regimes where boundary
layers develop. In such cases the implicit time integration and subsequent
solver enable efficient solutions with time steps determined by interior
transient flows.

The current algorithm and code base can benefit from a number of improvements.
To handle memory in large problems, methods for phase space compression and/or
velocity space decomposition will be needed. Modifications to the underlying
software, particularly on mesh data structures, to take advantage of the \ac{dg}
discretization can help improve parallel performance. We also expect to expand
the capability of the code to handle more complex and evolving boundaries. The
source iteration algorithm requires an acceleration strategy; some
possibilities are discussed in the introduction. Finally, we intend to extend
the physical fidelity of the underlying solver by (i) allowing models that
capture the correct Prandtl number in the fluid limit, (ii) extending the
capability to polyatomic atoms, and (iii) extending the capability to
multispecies settings.

%%---------------------------------------------------------------------------%%

\section*{Acknowledgements}

This work was supported under DOE NNSA NA-22 programs. This research used
resources of the Oak Ridge Leadership Computing Facility at the Oak Ridge
National Laboratory, which is supported by the Office of Science of the U.S.
Department of Energy under Contract No. DE-AC05-00OR22725. The authors thank
Dr.~Eirik Endeve of the Computer Science and Mathematics Division at Oak Ridge
National Laboratory for providing reference solutions to the cylindrical and
spherical Sod shock problems. The \nvidia results in this work were performed on
the narsil, dirac, and destiny computers at \ac{ornl}.

\section*{Declaration of AI-assisted technologies in the manuscript preparation process}

During the preparation of this work the authors used Github Copilot to assist
with the development of Python/matplotlib code that produced the plots that
appear in this paper. After using this service, the authors reviewed and edited
the content as needed and take full responsibility for the content of the
published article.

\bibliographystyle{elsarticle-num-names}
\bibliography{references}

\appendix
%%---------------------------------------------------------------------------%%
\section{DIRK Schemes}

An $S$-stage \ac{dirk} method is characterized by vectors $\mbb,\mbc \in
    \bbR^{S}$ and a lower triangular matrix $\mbA \in \bbR^{S \times S}$. We
consider the following L-stable \ac{dirk} schemes \cite{kennedy2016diagonally}.
\begin{equation}
    \label{eq:L-RK}
    \begin{array}{c|c}
        \mathbf{c} & \mathbf{A}     \\\hline
                   & \ \mathbf{b}^T
    \end{array} =
    \begin{cases}
        \begin{array}{c|c}
            1 & 1 \\\hline
              & 1
        \end{array}                                                                                    & O(\dt), \\
                                                                                                              &    \\
        \begin{array}{c|c}
            \begin{matrix}
                1-\frac{\sqrt{2}}{2} \\ 1
            \end{matrix} &
            \begin{matrix}
                1-\frac{\sqrt{2}}{2} & 0                    \\
                \frac{\sqrt{2}}{2}   & 1-\frac{\sqrt{2}}{2}
            \end{matrix}                               \\\hline
                                          & \begin{matrix}
                                                \frac{\sqrt{2}}{2} & 1-\frac{\sqrt{2}}{2}
                                            \end{matrix}
        \end{array}                    & O(\dt^2),                          \\
                                                                                                              &    \\
        \begin{array}{c|c}
            \begin{matrix}
                \num{0.435867} \\ \num{0.717933}\\ 1
            \end{matrix} &
            \begin{matrix}
                \num{0.435867} & 0               & 0              \\
                \num{0.282067} & \num{0.435867}  & 0              \\
                \num{1.208497} & \num{-0.644363} & \num{0.435867}
            \end{matrix}                                            \\\hline
                                                     & \begin{matrix}
                                                           \num{1.208497} & \num{-0.644363} & \num{0.435867}
                                                       \end{matrix}
        \end{array} & O(\dt^3),
    \end{cases}
\end{equation}
A \ac{rk} method is algebraically stable (or B-stable) if the matrices
\begin{align}
    B & = \mathop{\mathrm{diag}}(\mathbf{b})\qquad\text{and}     \\
    M & = BA + A^{T}B - \mathbf{b} \mathbf{b}^{T}
\end{align}
are positive semi-definite.
These properties are readily verified for the backward Euler method as well as the
$O(\Delta t^3)$ scheme, see \cite{kennedy2016diagonally}, given by

\begin{equation}
    \label{eq:B-RK}
    \begin{array}{c|c}
        \mathbf{c} & \mathbf{A}     \\\hline
                   & \ \mathbf{b}^T
    \end{array} =
    \begin{array}{c|c}
        \begin{matrix}
            \frac{3+\sqrt{3}}{6} \\ 1-\frac{3+\sqrt{3}}{6}
        \end{matrix} &
        \begin{matrix}
            \frac{3+\sqrt{3}}{6}   & 0                    \\
            1-\frac{3+\sqrt{3}}{3} & \frac{3+\sqrt{3}}{6}
        \end{matrix}                  \\\hline &
        \begin{matrix}
            0.5 & 0.5
        \end{matrix}
    \end{array} .
\end{equation}

%%---------------------------------------------------------------------------%%
\section{Velocity projection}
\label{app:velocity-projection}
Due to the finite velocity domain, we define a projection operator that exactly
integrates the moments for functions without compact support. Recall that $V =
        [-L,L]^d$ and
\begin{equation}
    \label{eq:discrete_space_v_repeat}
    \Wv = \left\{ w \in L^2(V): w\big|_{\Kv}\in
    \bbQ^{k_\bmv}(\Kv) \quad \forall \Kv \in \cT_{\bmv} \right\}.
\end{equation}
Let $L$ be the non-overlapping union of uniform cells $I_j$ for $j\in
    \{1, \dots, n_{\bmv}\}$ and define the space
\begin{equation}
    W_v = \left \{
    w \in L^2([-L,L]): w\big|_{I_j}\in \bbP^{k_\bmv}(I_j) \quad
    \forall j\in \{1, \dots, n_{\bmv}\} .
    \right\}
\end{equation}
Let $L_\mathrm{loc}^2(\bbR)$ be the space of locally square integrable functions
on $\bbR$.  Define the extension $E: W_v \to L_\mathrm{loc}^2(\bbR)$ by
extending the polynomial defined on $I_1$ for $v < -L$ and $I_{n_{\bmv}}$ for $v
    > L$.  Specifically, let $E_j: \bbP^{k_\bmv}(I_j) \to \bbP^{k_\bmv}(\bbR)$ be
the unique polynomial extension such that
\begin{equation}
    E_j w \big|_{I_j} = w \quad \forall w \in \bbP^{k_\bmv}(I_j).
\end{equation}
Then set
\begin{equation}
    [E w](v) =
    \begin{cases}
    [E_1 w](v)                 & , v < -L,           \\
                         [E_{n_{\bmv}} w](v) & , v > L ,           \\
                         w(v)                & , \text{otherwise}.
    \end{cases}
\end{equation}
Because of the tensor structure of $\Wv = W_{v_1} \otimes W_{v_2} \otimes
    W_{v_3}$, any function $\xi_{h} \in \Wv$ can be written as $\xi_{h}(\bmv) =
    \xi_{h,1}(v_1) \xi_{h,2}(v_2) \xi_{h,3}(v_3)$ where $\xi_{h,\alpha} \in
    W_{v_\alpha}$ and $\alpha \in \{1,2,3\}$. We define the extension $E_\bmv: \Wv
    \to L_\mathrm{loc}^2(\bbR^3)$ by $E_{\bmv} \xi_h(\bmv) = E \xi_{h,1}(v_1) E
    \xi_{h,2}(v_2) E \xi_{h,3}(v_3)$.  Then for any function $ \zeta \in
    L^2(\bbR^3)$ with compact support or sufficient decay at $\infty$ define
$\cP_{\bmv,h}\zeta \in W_{v}$ by
\begin{equation}\label{eqn:cons_projection_equation}
    \int_V \cP_{\bmv,h}\zeta(\bmv) \xi_h(\bmv) d \bmv =
    \int_{\bbR^3} \zeta(\bmv) E_{\bmv} \xi_h (\bmv) d \bmv \qquad\forall \xi_h \in W_\bmv.
\end{equation}
Thus for $n>0$, $\theta >0$, and $\bm{u}\in\mathbb{R}^d$
\begin{equation}\label{eqn:cons_projection_exactness}
    \int_{V} \ve{e}(\bmv)\cP_{\bmv,h}M_{n,\bm{u},\theta}(\bmv)~\dx{\bmv} =
    \ve{U}_{M_{n,\bm{u},\theta}}.
\end{equation}
The left-hand side of \eqref{eqn:cons_projection_exactness} is built exactly for
$P_{\bmv,h}M_{n,\bm{u},\theta}$ using error functions to ensure
\eqref{eqn:cons_projection_exactness} is satisfied independent of the quadrature
rule used for the other velocity integrals.

%%---------------------------------------------------------------------------%%

%%---------------------------------------------------------------------------%%
\end{document}